\documentclass[journal,10pt]{IEEEtran}

\usepackage[nocompress]{cite}
\usepackage{url}
\usepackage[pdftex]{graphicx}
\usepackage[cmex10]{amsmath}
\usepackage[linesnumbered,lined,vlined,ruled,commentsnumbered]{algorithm2e}
\usepackage{array}
\usepackage[caption=false,font=footnotesize,labelfont=sf,textfont=sf]{subfig}
\usepackage{fixltx2e}
\usepackage{stfloats}
\usepackage{multirow}
\usepackage{color}
\usepackage{amssymb}
\usepackage[english]{babel}
\usepackage{amsthm}
\usepackage{bm}

\usepackage{bbm}

\usepackage[T1]{fontenc}
\usepackage[numbers,sort&compress]{natbib}
\usepackage{url}
\usepackage{booktabs}

\newtheorem{theorem}{Theorem}
\newtheorem{lemma}{Lemma}

\newtheorem{definition}{Definition}

\newtheorem{game}{Game}

\def\S{\mathcal{D}}

\def\I{\mathcal{I}}

\begin{document}
\title{Selfish Caching Games on Directed Graphs}

\author{Qian~Ma,~\IEEEmembership{Member,~IEEE,}
		Edmund~Yeh,~\IEEEmembership{Senior Member,~IEEE,}
        and~Jianwei~Huang,~\IEEEmembership{Fellow,~IEEE}
\thanks{Qian Ma is with the School of Intelligent Systems Engineering, Sun Yat-sen University. E-mail: maqian25@mail.sysu.edu.cn.}
\thanks{Edmund Yeh is with the Department of Electrical and Computer Engineering, Northeastern University, Boston, USA 02215. E-mail: eyeh@ece.neu.edu.}
\thanks{Jianwei Huang is with the Shenzhen Institute of Artificial Intelligence and Robotics for Society, the School of Science and Engineering, The Chinese University of Hong Kong, Shenzhen, Shenzhen 518172, China. E-mail: jianweihuang@cuhk.edu.cn. (\emph{Corresponding author: Jianwei Huang.})}
\thanks{This work is supported by the National Natural Science Foundation of China under Grant 62002399, the Shenzhen Institute of Artificial Intelligence and Robotics for Society, the Presidential Fund from the Chinese University of Hong Kong, Shenzhen, the National Science Foundation grant CNS-1718355 and a grant from Intel Corp.}
\thanks{Part of this work has been published at the Twentieth ACM International Symposium on Mobile Ad Hoc Networking and Computing (Mobihoc '19), July 2--5, 2019, Catania, Italy \cite{QianMobiHoc2019}.}
}

\maketitle

\begin{abstract}
Caching networks can reduce the routing costs of accessing contents by caching contents closer to users. However, cache nodes may belong to different entities and behave selfishly to maximize their own benefits, which often lead to performance degradation for the overall network. While there has been extensive literature on allocating contents to caches to maximize the social welfare, the analysis of selfish caching behaviors remains largely unexplored. 
In this paper, we model the selfish behaviors of cache nodes as selfish caching games on arbitrary directed graphs with heterogeneous content popularity. We study the existence of a pure strategy Nash equilibrium (PSNE) in selfish caching games, and analyze its efficiency in terms of social welfare. We show that a PSNE does not always exist in arbitrary-topology caching networks. However, if the network does not have a mixed request loop, i.e., a directed loop in which each edge is traversed by at least one content request, we show that a PSNE always exists and can be found in polynomial time. Furthermore, we can avoid mixed request loops by properly choosing request forwarding paths. We then show that the efficiency of Nash equilibria, captured by the price of anarchy (PoA), can be arbitrarily poor if we allow arbitrary content request patterns, and adding extra cache nodes can make the PoA worse, i.e., cache paradox happens. However, when cache nodes have homogeneous request patterns, we show that the PoA is bounded even allowing arbitrary topologies. 
We further analyze the selfish caching games for cache nodes with limited computational capabilities, and show that an approximate PSNE exists with bounded PoA in certain cases of interest. Simulation results show that increasing the cache capacity in the network improves the efficiency of Nash equilibria, while adding extra cache nodes can degrade the efficiency of Nash equilibria. 
\end{abstract}

\begin{IEEEkeywords}
Caching networks, selfish caching games, Nash equilibrium, price of anarchy.
\end{IEEEkeywords}

\section{Introduction}

\IEEEPARstart{C}{aching} networks can reduce the routing costs for accessing contents by caching the requested contents as close to the requesting users as possible. 
Prevailing caching networks include content delivery networks (CDN)~\cite{CDN,CDN2}, information-centric networks (ICN)~\cite{YehICN2014}, femtocell networks~\cite{femtocell}, web caching networks~\cite{DSR}, and peer-to-peer networks~\cite{p2p}. 
There has been extensive previous work (e.g., \cite{poularakis2018distributed, YehSigmetrics, ao2015distributed}) on how to optimally allocate contents to available caches. 
However, most existing work assumes that cache nodes are altruistic and cooperate with each other to optimize an overall network performance objective.

In practice, cache nodes may belong to different entities \cite{Milking}. 
For example, in wireless community mesh networks such as Google WiFi \cite{afanasyev2010usage} and Guifi \cite{vega2012topology}, individual users contribute their wireless routers (as caches) to the community. 
On the Internet, different operators and providers deploy their own caching infrastructures and services. 
Examples include AT$\&$T Content Delivery Network Service, Google Global Cache, Netflix Open Connect, and Akamai.\footnote{In this paper, we treat one provider as one cache node. Another example is a caching network where each provider owns multiple cache nodes. In such a case, we can model the interactions among multiple cache nodes belonging to the same provider as a coalititional game.}

In caching networks where different entities operate their own caches, cache nodes may behave selfishly to maximize their own benefits. 
For example, in multi-hop wireless community mesh networks \cite{draves2004routing}, a cache node has an incentive to cache the content items to minimize its own routing cost, which may not always maximize the social welfare. 
This motivates us to study the selfish caching behaviors through a game-theoretic approach.

To our best knowledge, this is the first paper to examine \emph{selfish caching games} on arbitrary directed graphs with heterogeneous content popularity. 
We focus on the pure strategy Nash equilibrium (PSNE),\footnote{The main reason for implementing PSNE in practice is simplicity \cite{ValidUtilityGame}.} and address two fundamental questions. 
\emph{First, is a PSNE guaranteed to exist in any selfish caching game?}
\emph{Second, if a PSNE exists, does it have a guaranteed efficiency in terms of social welfare?}
The short answers to the above two questions are ``No'' and ``No''. 
In other words, the selfish caching game does not always admit a PSNE. Even if a PSNE exist, its efficiency in terms of social welfare can be very poor.

In this paper, we characterize the conditions under which (i) a PSNE exists, and (ii) a PSNE has a guaranteed efficiency. 
We characterize the efficiency of PSNE by the \emph{price of anarchy} (PoA), which is the ratio of the social welfare achieved by the worst PSNE to that achieved by a socially optimal strategy  \cite{CSgame, roughgarden2002bad}. 
The analysis of PSNE and PoA takes into account the asymmetric and node-specific interdependencies among cache nodes, which reflect the network topology and content request patterns. 
Our analysis will help the network designer understand when the network behaves with certain performance guarantees, and how to create these conditions in the network.

We analyze the selfish caching game in two scenarios. 
We first consider a scenario where all contents have equal sizes, which corresponds to practical applications such as video-on-demand services using harmonic broadcasting that divide each video into segments of equal size \cite{juhn1997harmonic}. 
We then consider a scenario where contents have unequal sizes, which corresponds to practical applications such as video streaming services over HTTP (e.g., Netflix and Hulu) that split each video into segments of lengths from 2 to 10 seconds \cite{huang2012confused}.

Our primary contributions are:
\begin{itemize}
\item \emph{Selfish Caching Game}: To the best of our knowledge, this is the \emph{first} work that studies the selfish caching game on directed graphs with arbitrary topologies and heterogeneous content popularity. 
\item \emph{Pure Strategy Nash Equilibrium (PSNE)}: For selfish caching games with equal-sized content items, we first show that a PSNE does not always exist. We then show that a PSNE exists if the network does not have a mixed request loop, i.e., a directed loop in which each edge is traversed by at least one content request. Furthermore, we propose a polynomial-time algorithm to find a PSNE for the selfish caching game with no mixed request loop. 
\item \emph{Price of Anarchy}: We show that the PoA in general can be arbitrarily poor if we allow arbitrary content request patterns. 
Furthermore, adding extra cache nodes can make the PoA worse, a phenomenon which we call the \emph{cache paradox}.
However, when cache nodes have homogeneous request patterns, we show that the selfish caching game is an $\alpha$-scalable valid utility game and the PoA is bounded in arbitrary-topology caching networks. 
\item \emph{Approximate PSNE}: For selfish caching games with unequal-sized content items, each node's payoff maximization problem is NP-hard. When cache nodes have limited computational capability, we show that their selfish caching behaviors lead to an approximate PSNE with bounded PoA in certain cases of interest. 
\end{itemize}

The rest of the paper is organized as follows. 
In Section \ref{sec:literature}, we review related literature.
In Section \ref{sec:model}, we introduce our system model. 
In Section \ref{sec:SCG}, we model the selfish caching game and analyze the PSNE.
In Section \ref{sec:PoA}, we study the PoA. 
In Section \ref{sec:Approx}, we analyze selfish caching games with unequal-sized content items. 
In Section \ref{sec:simu}, we provide simulation results. 
We conclude in Section \ref{sec:conclusion}. 

\section{Related Work}\label{sec:literature}

There has been a rich body of previous work on caching, many of which are summarized in an excellent recent survey \cite{Survey}. 
In the following, we introduce related work regarding caching optimization and selfish caching game, respectively.

\textbf{Caching Optimization.} 
There is considerable recent literature on a variety of caching optimization problems, including proactive caching \cite{shukla2017hold, tadrous2016joint}, optimal caching under queuing models \cite{YuanyuanInfocom2020, KellyCache2019}, optimal caching under unknown content popularities \cite{garetto2015efficient, zhang2018coded}, distributed adaptive algorithms for optimal caching \cite{poularakis2018distributed, YehSigmetrics, ao2015distributed, DrCache2018}, caching at the edges \cite{zhao2018red, li2018hierarchical, zhao2018collaborative, kwak2018hybrid, cao2018optimal}, TTL (time-to-live) caches \cite{FerragutSIGMETRICS2016, DehghanTTLton2019}, optimal caching in evolving networks \cite{qin2018content}, joint caching and routing optimization \cite{dehghan2015complexity, amble2011content, StratisJSAC2018}, optimal cache partitioning \cite{chu2016allocating}, and collaborative caching \cite{gharaibeh2016provably, shin2017t, rahimzadeh2017svc, yu2016enhancing, Lui,  maille2015impact}. 
All the above work assumes that all cache nodes aim to maximize the social welfare.

\textbf{Selfish Caching Game.} 
There are several papers which study selfish caching behaviors in simple settings. 
In \cite{SelfishCaching}, Chun \emph{et al.} study the selfish caching game on undirected graphs with a single content item, assuming homogeneous content popularity across users. 
In \cite{MarketSharing}, Goemans \emph{et al.} study the content market sharing game, where users get rewards for caching content items. 
The paper assumes that any node which caches a requested item can serve the request with same cost, without considering network topology. 
The authors in \cite{DSR} and \cite{DSR2} study a distributed selfish replication game in an undirected complete graph, where the distance between any two nodes is the same. 
In \cite{CSR}, Gopalakrishnan \emph{et al.} study the capacitated selfish replication game in an undirected network, where users are equally interested in a set of content items.

The analysis in the above literature is applicable to undirected graphs, and some are restricted to homogeneous content popularity. 
In this work, we study the selfish caching game on directed graphs with arbitrary topologies and heterogeneous content popularity.

\textbf{\begin{figure}[t]
\centering
\includegraphics[width=0.5\textwidth]{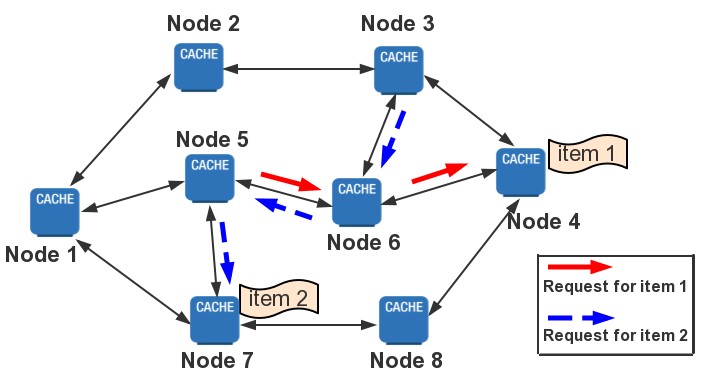}
\caption{A caching network with $|V|=8$ nodes and $|\I|=2$ content items, where node 4 (node 7, respectively) is the designated server of item 1 (item 2, respectively). The request forwarding paths are fixed in our model. For example, the path of node 5 requesting item 1 is $p^{(5,1)}=(5,6,4)$, and the path of node 3 requesting item 2 is $p^{(3,2)}=(3,6,5,7)$. }\label{fig:Model}
\end{figure}}

\section{System Model}\label{sec:model}

We consider a network of selfish caches, represented by a directed caching graph $G(V,E)$ with an \emph{arbitrary topology}, where $V$ is the set of cache nodes and $E$ is the set of bidirectional edges which enable ARQ with asymmetric edge costs (see an example in Figure \ref{fig:Model}). 
Each cache node requests one or more content items (e.g., movies) from the set $\I=\{1,\ldots,|\I|\}$. 
For each content item $i\in \I$, there is a fixed set of \emph{designated server nodes} $\mathcal{D}^i \subseteq V$, $|\mathcal{D}^i|>0$, that store $i$ in their permanent storage (outside of their caches).\footnote{For example, designated server nodes can be content providers' caches.} 
We consider equal-sized content items in Sections \ref{sec:model}--\ref{sec:PoA}, which correspond to applications such as video-on-demand services using harmonic broadcasting that divide each video into segments of equal size \cite{juhn1997harmonic}.\footnote{Without loss of generality, we normalize the size of each item to be one.} 
We will consider the case of unequal-sized items in Section \ref{sec:Approx}.

\subsection{Caching Strategies}
Each node $s\in V$ has a cache of capacity $c_s \in \mathbb{N}$, i.e., node $s$ can store exactly $c_s$ equal-sized content items. 
We denote the caching strategy of node $s\in V$ by $\boldsymbol{x}_s =\{x_{si}: \forall i\in \I\} \in\{0,1\}^{|\I|}$, where 
$$x_{si} \in\{0,1\}, \mbox{ for all } i\in\I,$$ 
indicates whether node $s$ stores content item $i$, and satisfies 
\begin{equation*}
\textstyle \sum_{i\in\I}  x_{si} \leq c_s, \mbox{ for all } s\in V.
\end{equation*}
We let $\boldsymbol{x}_{-s}=\{\boldsymbol{x}_1,\ldots, \boldsymbol{x}_{s-1},\boldsymbol{x}_{s+1},\ldots, \boldsymbol{x}_{|V|}\}$ denote the caching strategy of nodes other than node $s$, and let $\boldsymbol{x}=\{\boldsymbol{x}_s,\boldsymbol{x}_{-s} \}$ denote the global caching strategy. 
Given $\boldsymbol{x}_s$, we let $Z_s=\{i:x_{si}=1,i\in\I\}$ denote the set of items cached by node $s\in V$.

\subsection{Content Requests}
We describe each content request by a pair $(s,i)$, where the request source\footnote{We consider a request source to be a point of aggregation which combines many network users. While a single user may request a given content item only once over a time period, an aggregation point is likely to submit many requests for a given content item over a time period.} $s\in V$ requests content item $i\in\I$. 
We assume that each request $(s,i)$ arrives according to a stationary ergodic process \cite{jiang2018convergence, PanigraphyPoisson2018} with arrival rate $\lambda _{(s,i)} \geq 0$ for all $s \in V$ and $i \in \I$, which reflects \emph{heterogeneous content popularity} across items and request nodes.\footnote{We consider selfish caching behaviors under complete information, where cache nodes know all other nodes' content request patterns \cite{Milking,SelfishCaching,MarketSharing}. 
Specifically, cache nodes can estimate content request patterns through historical information or long-term learning \cite{shukla2017hold}.}

Request $(s,i)$ is forwarded over a pre-determined fixed \emph{request forwarding path}\footnote{Similar as in the named data networks, we assume that the request forwarding path is determined in a longer timescale compared with caching. And we consider selfish caching behaviors under complete information where cache nodes know the request forwarding paths \cite{Milking,SelfishCaching,MarketSharing}.} $p^{(s,i)}$, from request source $s$ to one of content item $i$'s designated server nodes in $\mathcal{D}^i$.
Specifically, the path $p^{(s,i)}$ of length $K \leq |V|$ is a sequence $(p_1,\ldots,p_K)$ of nodes $p_k\in V$ such that $p_1=s$, $p_K\in\mathcal{D}^i$, and $(p_k,p_{k+1})\in E$ for all $k\in\{1,\ldots,K-1\}$. 
We require that $p^{(s,i)}$ contains no loops ($p_k\neq p_l$ for all $1\leq k< l\leq K$) and no node other than the terminal node on $p^{(s,i)}$ is a designated server for content item $i$ ($p_k\not\in\mathcal{D}^i$ for all $1\leq k<K$). 
For request $(s,i)$, we let $V_{(s,i)}=\{v:v\in p^{(s,i)}, v\neq s, v\notin \S^i \}$ denote the set of intermediate nodes on path $p^{(s,i)}$. 
We denote $V_s=\cup_{i\in\I}V_{(s,i)}$ as the set of intermediate nodes on all the request forwarding paths of node $s$.\footnote{Note that each cache node can play some or all of the following roles: a designated server of content items, a source of requests, and an intermediate node on request forwarding paths.}

Request $(s,i)$ travels along path $p^{(s,i)}$ until either (i) the request reaches a node $v\in p^{(s,i)}$ such that node $v$ caches content item $i$, i.e., $x_{vi}=1$ or, (ii) if $x_{vi}=0$ for all $v\in p^{(s,i)}\setminus \{p_K\}$, the request reaches $p_K\in \mathcal{D}^i$.
Having found the closest copy of content item $i$, the network generates a \emph{response message} carrying the requested content item $i$.
The response message is propagated in the reverse direction along the request forwarding path, i.e., from the closest node with content item $i$ back to the request source node $s$.\footnote{In this paper, we assume that forwarding and transmission follow standard network protocols. In some settings, forwarding and transmission incur a service cost to the cache node due to the consumption of the transmit power and communication resource. We will consider such costs in the future work.}

\begin{table}[t]
\newcommand{\tabincell}[2]{\begin{tabular}{@{}#1@{}}#2\end{tabular}}
\centering
\caption{Key Notation}
\begin{tabular}{l l }
\hline
$G(V,E)$  & Caching graph, with nodes in $V$ and edges in $E$  \\
$c_s$ & Cache capacity of node $s\in V$   \\
$w_{uv}$ & Cost on edge $(u,v)\in E$   \\
$\I$ & Set of content items    \\ 
$\mathcal{D}^i$ & Set of designated servers for content item $i \in \I$ \\
$(s,i)$ & Request for item $i$ from node $s$  \\
$\lambda_{(s,i)}$ & Arrival rate of request $(s,i)$ \\
$p^{(s,i)}$ & Request forwarding path of request $(s,i)$ \\
$V_{(s,i)}$ & The set of intermediate nodes on path $p^{(s,i)}$\\
$V_{s}$ & The set of intermediate nodes on node $s$' paths\\
$x_{si}$ & Caching strategy of node $s\in V$ for item $i\in \I$\\
$\boldsymbol{x}_s$ & Caching strategy of node $s\in V$ \\
$Z_s$ & The set of content items cached by node $s\in V$\\
$\boldsymbol{x}$ & Global caching strategy of all nodes \\
$h_{(s,i)}$ & The routing cost to serve request $(s,i)$ \\
$h_s$ & The routing cost of node $s$ \\
$g_s$ & The caching gain of node $s$ \\
$G$ & The aggregate caching gain in the network \\
\hline
\end{tabular}
\label{table:Notation}
\end{table}

\subsection{Routing Costs}\label{sec:RoutingCosts}

Transferring a content item across edge $e=(u,v)\in E$ incurs a cost (e.g., delay or financial expense) denoted by $w_{uv}\geq 0$.\footnote{We do not model the congestion effect on each edge. How to jointly consider cost and throughput issues is an interesting open problem.}
Since the size of each request message is relatively small compared with the content item, we assume that costs are only due to content item transfers, and the costs of forwarding requests are negligible \cite{YehSigmetrics}. 
To serve the request $(s,i)$, the routing cost depends on the caching decision $x_{si}$ of the request source node $s$, as well as the caching decisions $x_{vi}, \forall v\in V_{(s,i)}$, of all the intermediate nodes on the request forwarding path $p^{(s,i)}$. 
Specifically, the routing cost of transferring item $i$ over the reverse direction of $p^{(s,i)}$ is 
\begin{equation*}
\begin{aligned}
&\textstyle h_{(s,i)}\left(x_{si},\{x_{vi}:v\in V_{(s,i)}\}  \right)\\
=& \textstyle \sum_{k=1}^{|p^{(s,i)}|-1}w_{p_{k+1}p_k} \prod_{k'=1}^{k}\left( 1-x_{p_{k'}i} \right)\\
=& \textstyle \sum_{k=1}^{|p^{(s,i)}|-1}w_{p_{k+1}p_k} \left( 1-x_{si} \right) \prod_{k'=2}^{k}\left( 1-x_{p_{k'}i} \right).
\end{aligned}
\end{equation*}
Note that $h_{(s,i)}(\cdot)$ includes the cost on edge $(p_{k+1},p_k)$, i.e., $w_{p_{k+1}p_k}$, if and only if none of the nodes from $p_1$ to $p_k$ on path $p^{(s,i)}$ has cached content item $i$. 
For example, in Figure \ref{fig:Model}, $p^{(3,2)}=(3,6,5,7)$ and the routing cost of request $(3,2)$ depends on $x_{32}$, $x_{62}$ and $x_{52}$. 
If $(x_{32},x_{62},x_{52})=(0,0,1)$, then $h_{(3,2)}(x_{32},x_{62},x_{52})=w_{63}+w_{56}$.

\subsection{Selfish Caching Behavior}

Each selfish cache node $s\in V$ seeks a caching strategy to optimize its own benefit, i.e., minimizing the aggregate expected cost for serving all its own requests, calculated as follows: 
\begin{equation}\label{eq:aggregatecost}
\begin{aligned}
&  h_s\left(\boldsymbol{x}_s, \{\boldsymbol{x}_v: v\in V_s\} \right)\\
=& \sum_{i\in\I}\lambda_{(s,i)} \cdot  h_{(s,i)}\left(x_{si},\{x_{vi}:v\in V_{(s,i)}\}  \right) .
\end{aligned}
\end{equation} 
For notation simplicity, we write $h_s(\cdot)$ as $h_s(\boldsymbol{x}_s,\boldsymbol{x}_{-s})$. 
In the absence of caching, i.e., $\boldsymbol{x}=\boldsymbol{0}$, the aggregate expected cost of node $s$ is: 
\begin{equation*}
h_s(\boldsymbol{0})=\sum_{i\in\I}\lambda_{(s,i)} \sum_{k=1}^{|p^{(s,i)}|-1}w_{p_{k+1}p_k}.
\end{equation*}
We define the \emph{caching gain} of node $s$ as
\begin{equation}\label{eq:CachingGain}
g_s(\boldsymbol{x}_s,\boldsymbol{x}_{-s})=h_s(\boldsymbol{0})-h_s(\boldsymbol{x}_s,\boldsymbol{x}_{-s}).
\end{equation}
Intuitively, the caching gain is the cost reduction enabled by caching.
Since $h_s(\boldsymbol{0})$ is a constant, minimizing the aggregate expected cost in~\eqref{eq:aggregatecost} is equivalent to maximizing the caching gain in~\eqref{eq:CachingGain}. 
Hence, the caching gain in~\eqref{eq:CachingGain} serves as node $s$' payoff function.

\textbf{\begin{figure}[t]
\centering
\includegraphics[width=0.33\textwidth]{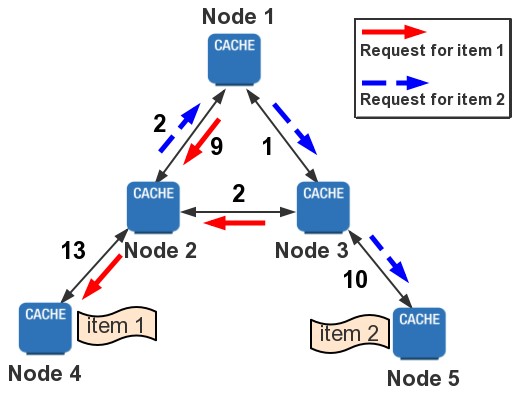}
\caption{An example where PSNE does not exist. The caching network has $|V|=5$ nodes and $|\I|=2$ content items, where node 4 (node 5, respectively) is the designated server of item 1 (item 2, respectively). The cache capacity is $1$ at each node. The request arrival rates satisfy $\lambda_{(v,i)}=\lambda_i, \forall v\in V, i\in \I$, where $\lambda_1=10$ and $\lambda_2=14$. The request forwarding paths are fixed, for example, $p^{(3,1)}=(3,2,4)$.}\label{fig:NoNE}
\end{figure}}

\section{Selfish Caching Game}\label{sec:SCG}

In this section, we model the interactions among selfish cache nodes by a selfish caching game on directed graphs. 
We construct an example where the pure strategy Nash equilibrium (PSNE) does not exist for such a game. 
We then identify the condition under which a PSNE exists, and propose a polynomial-time algorithm to find a PSNE under the condition.

\subsection{Game Modeling}

We define the selfish caching game as follows:

\begin{game}[Selfish Caching Game on Directed Graphs]
$ $
\begin{itemize}
\item Players: the set $V$ of cache nodes on the caching graph; 
\item Strategies: the caching strategy $\boldsymbol{x}_s=\{x_{si}: \forall i\in\I\}$ for each cache node $s\in V$, where $x_{si}\in\{0,1\}$ and $\sum_{i\in\I}x_{si} \leq c_s$;
\item Payoffs: the caching gain $g_s(\boldsymbol{x}_s,\boldsymbol{x}_{-s})$ for each $s\in V$.
\end{itemize}
\end{game}

Since the selfish caching game is a finite game, there exists at least one mixed strategy Nash equilibrium (including pure strategy Nash equilibrium as a special case). 
However, since it is difficult to implement random caching strategies in practical caching networks, we focus on analyzing pure strategy Nash equilibria in this paper, as defined below.

\begin{definition}[Pure Strategy Nash Equilibrium]
A pure strategy Nash equilibrium of the selfish caching game is a caching strategy profile $\boldsymbol{x}^{\rm NE}$ such that for every cache node $s\in V$,
\begin{equation}
g_s(\boldsymbol{x}_s^{\rm NE},\boldsymbol{x}_{-s}^{\rm NE}) \geq g_s(\boldsymbol{x}_s,\boldsymbol{x}_{-s}^{\rm NE}), \mbox{ for all feasible } \boldsymbol{x}_s.
\end{equation}
\end{definition}

\begin{figure}[t]
\centering
\includegraphics[width=0.35\textwidth]{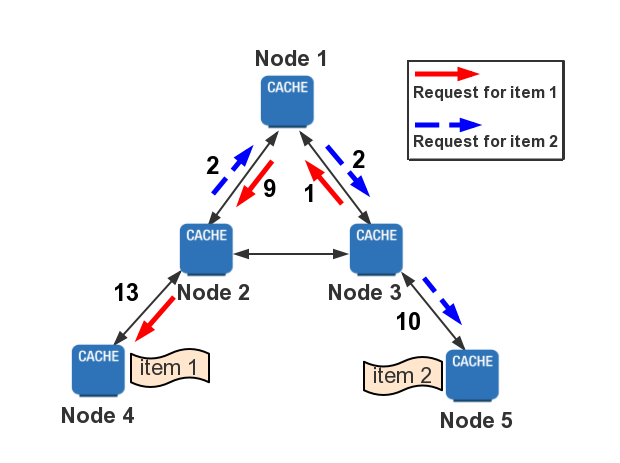}
\caption{An example where there is no mixed request loop. The request forwarding path is $p^{(3,1)}=(3,1,2,4)$. }\label{fig:NoNEAPP}
\end{figure}

\subsection{An Example with No PSNE}

In the following, we first show that the PSNE does not always exist.

\begin{theorem}\label{theo:NoNE}
There exists a selfish caching game for which the pure strategy Nash equilibrium does not exist. 
\end{theorem}

\begin{proof}
Figure \ref{fig:NoNE} is an example with no PSNE. 
For node 4 (node 5, respectively), caching item 2 (item 1, respectively) is its dominant strategy. 
Now we analyze the selfish behaviors of nodes 1, 2, and 3.
It is easy to verify that for all 8 feasible caching strategy profiles, there always exists one cache node that can improve its caching gain by changing its caching strategy unilaterally. 
For example, if all the three nodes cache item 1, then node 3 has the the incentive to cache item 2 to improve its caching gain assuming that the other two nodes do no change their caching strategies. 
Hence there is no strategy profile where everyone is achieving its maximum payoff assuming other nodes do not change their strategies. 
Hence, the PSNE does not exist. 
\end{proof}

\subsection{Existence of a PSNE}

Deciding the existence of a PSNE for games on graphs is NP-hard in general \cite{wang2014belief}. 
However, we identify the condition under which a PSNE of the selfish caching game exists and can be found in polynomial time. 
To proceed, we first introduce the definition below.

\begin{definition}[Mixed Request Loop]
A mixed request loop on a directed graph is a directed loop $(p_1,p_2,\ldots,p_K,p_{K+1}=p_1)$ involving $3 \leq K \leq |V|$ nodes, where $p_k\in V$ for $1 \leq k \leq K$, $p_k\neq p_l$ for all $1\leq k< l\leq K$, and at least one content request traverses edge $(p_k,p_{k+1})\in E$ for all $1 \leq k \leq K$. 
\end{definition}

In Figure \ref{fig:NoNE}, $(1,3,2,1)$ forms a mixed request loop, where requests for item 1 traverse edge $(3,2)$, and requests for item 2 traverse edges $(2,1)$ and $(1,3)$.

Note that a loop on graph is not always a mixed request loop. 
For example, $(1,3,2,1)$ in Figure \ref{fig:NoNEAPP} is a loop.
However, the request forwarding path is $p^{(3,1)}=(3,1,2,4)$ rather than $(3,2,4)$, meaning no request traverses edge $(3,2)$.  Hence, loop $(1,3,2,1)$ is not a mixed request loop. 
In other words, we can avoid mixed request loops by properly choosing the request forwarding paths.

Next, we will show that a PSNE exists in the selfish caching game on caching graphs with no mixed request loop.\footnote{Note that no mixed request loop is a sufficient (but not necessary) condition for a PSNE to exist.}

\begin{theorem}\label{theo:Existence}
A PSNE always exists in the selfish caching game on caching graphs with no mixed request loop. 
\end{theorem}

\begin{proof}
See Appendix A. 
\end{proof}

Theorem \ref{theo:Existence} holds in caching networks with arbitrary topologies and heterogeneous content popularity. 
We prove the existence\footnote{The selfish caching game generally admits multiple PSNEs, depending on system parameters such as edge weights and request arrival rates.} of the PSNE by finding a PSNE in polynomial time.

\subsection{Polynomial-Time Algorithm to Find a PSNE}

In this section, we present a polynomial-time algorithm to find a PSNE for the selfish caching game. 
Specifically, for each selfish caching game, we can define a \emph{state graph} \cite{MarketSharing} as follows.

Recall that given node $s$' caching strategy $\boldsymbol{x}_s$,  the set $Z_s=\{i:x_{si}=1, i\in\I\}$ is the set of content items cached by node $s\in V$. 
Hence we can use $\boldsymbol{x}=\{\boldsymbol{x}_s: \forall s\in V\}$ and $Z=\{Z_s: \forall s\in V\}$ interchangeably to represent the caching strategy profile.

\begin{definition}[State Graph \cite{MarketSharing}]
A state graph is a directed graph where each vertex corresponds to a strategy profile $Z$. 
There is a directed arc from vertex $Z$ to vertex $Z'$ with label $v$ if the only difference between $Z$ and $Z'$ is the strategy of player $v$ and the payoff of player $v$ in $Z$ is strictly less than its payoff in $Z'$. 
\end{definition}

A PSNE corresponds to a vertex on the state graph without any outgoing arc, i.e., a sink. 
Hence identifying a PSNE of the selfish caching game is equivalent to identifying a sink on the corresponding state graph.

We propose a polynomial-time algorithm (Algorithm \ref{algo:FindNESG} \cite{MarketSharing}) to find a sink on the state graph. 
The algorithm proceeds in rounds. 
The first round starts at the vertex $Z=\emptyset$, corresponding to the strategy profile where none of the cache nodes cache any content item (Line 1 of Algorithm \ref{algo:FindNESG}). 
In each round, the first arc traversed on the state graph corresponds to an \emph{add arc} where a player, say $s$, changes from $Z_s$ to $Z_s \cup \{i^\ast\}$. 
Intuitively, player $s$ adds only one content item $i^\ast$ to its cache, where we select $i^\ast$ among all possible content items not currently in $Z_s$ to maximize player $s$' caching gain (Lines 3-4 of Algorithm \ref{algo:FindNESG}). 
After the first arc, subsequent arcs in the same round correspond to \emph{change arcs}. 
Specifically, a change arc corresponds to a player, say $v$, replacing $Z_v$ by $Z_v\cup \{j\} \setminus \{t\}$, where $j \notin Z_v$ and $t\in Z_v$. 
Intuitively, player $v$ replaces content item $t$ for content item $j$ if $g_v(Z_v\cup \{j\} \setminus \{t\},Z_{-v})> g_v(Z_v,Z_{-v})$ (Lines 5-7 of Algorithm \ref{algo:FindNESG}). 
When the current vertex on the state graph has no change arcs, one round ends. 
For the vertex where a round ends, if there is an add arc outgoing from it, a new round starts; otherwise, it is a sink and the algorithm terminates. 
Such a sink corresponds to the PSNE.

\begin{algorithm}[t]
\LinesNumbered
\SetAlgoLined
\begin{small}
\KwIn{$G(V,E), \I, w_{uv}, \forall (u,v)\in E, \lambda_{(s,i)}$ and $p^{(s,i)}$, for all $s\in V, i\in \I$}
\KwOut{$Z^{\rm NE}$}
Set $Z=\emptyset$\;
\Repeat{$\forall s\in V$ satisfies $|Z_s|=c_s$}{
Randomly pick a node $s\in V$ where $|Z_s|<c_s$\;
Add item $i^\ast$ where $i^\ast \in \arg \max_{i\in\I\setminus Z_s} g_s(Z_s \cup \{i\},Z_{-s})$ to node $s$, i.e., $Z_s \gets Z_s \cup \{i^\ast\}$\;
\While{$\exists v\in V, j\notin Z_v, t\in Z_v$, such that $g_v(Z_v\cup \{j\} \setminus \{t\},Z_{-v})> g_v(Z_v,Z_{-v})$}
{
Set $Z_v \gets Z_v\cup \{j\} \setminus \{t\}$\;
}
}
Set $Z^{\rm NE} = Z$\;
\end{small}
\caption{Find PSNE on State Graph \cite{MarketSharing}}
\label{algo:FindNESG}
\end{algorithm}

In the following theorem, we show that Algorithm \ref{algo:FindNESG} can find a sink on the state graph in polynomial time.

\begin{theorem}\label{theo:ExistencePolytime}
For the selfish caching game on arbitrary-topology caching graphs with no mixed request loop, Algorithm \ref{algo:FindNESG} computes a PSNE in polynomial time by traversing a path of length at most $|V||\I|^2(|V|-2)^2$ on the corresponding state graph.
\end{theorem}

\begin{proof}
See Appendix B. 
\end{proof}

Note that for any given selfish caching game, Algorithm \ref{algo:FindNESG} does not require the construction of the whole state graph. 
At any given vertex of the state graph, Algorithm \ref{algo:FindNESG} only requires one to find the next arc to traverse, which takes $\mathcal{O}(|V|)$ time. 
Hence, the total maximum running time of Algorithm \ref{algo:FindNESG} is $\mathcal{O}(|V|^2|\I|^2(|V|-2)^2)$. 
Furthermore, different random choices of the next arc to traverse in Algorithm \ref{algo:FindNESG} correspond to different outcomes if there is more than one PSNE in the selfish caching game. 

Since each cache node maximizes its own benefit, a PSNE of the selfish caching game does not in general optimize the social welfare. 
We will quantify the efficiency of the Nash equilibria in terms of social welfare next.

\section{Price of Anarchy}\label{sec:PoA}

To evaluate the efficiency of Nash equilibria, we analyze the price of anarchy (PoA) \cite{CSgame}, i.e., the ratio of the social welfare achieved by the worst Nash equilibrium to that achieved by a socially optimal strategy. 

In this paper, we define the social welfare as the aggregate caching gain in the network.
Specifically, the social welfare maximization problem is
\begin{equation}\label{prob:maxCG}
\begin{aligned}
\displaystyle \mbox{max}~ & ~~  \textstyle G(\boldsymbol{x}) \triangleq \sum_{s\in V} g_s(\boldsymbol{x}_s,\boldsymbol{x}_{-s}) \\
\mbox{s.t.}~ & ~~  \textstyle  \sum_{i\in\I}  x_{si} \leq c_s,  ~x_{si} \in\{0,1\}, ~ \forall  s\in V, i\in\I .
\end{aligned}
\end{equation}
Problem \eqref{prob:maxCG} is NP-hard \cite{YehSigmetrics}. 
It is challenging to calculate the socially optimal solution and analyze the PoA in general.

In the following, we first show that the PoA can be arbitrarily poor if we allow any content request patterns. 
We then identify the cache paradox where adding extra cache nodes can make the PoA worse. 
Under reasonable constraints of request patterns and paths, however, we can show that the PoA is bounded in general caching networks. 
Furthermore, for given caching networks with known network topology and parameters, we can derive a better bound for PoA.

\subsection{An Example with an Arbitrarily Poor PoA}\label{sec:BadPoA}

Next, we show that the PoA can be arbitrarily close to $0$, indicating that the selfish caching behaviors can lead to unboundedly poor performance in terms of social welfare.

\begin{figure}[t]
\centering
\includegraphics[width=0.3\textwidth]{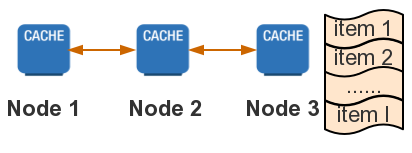}
\caption{An example where PoA approaches $0$. Consider a caching network with $|V|=3$ nodes where node 3 is the designated server of content items in set $\I=\{1,2,\ldots,I\}$. 
The request arrival rates at node 1 satisfy $ \lambda_{(1,1)} >0, \lambda_{(1,i)}= \lambda_{(1,1)}-\epsilon>0$, where $ \epsilon > 0$, for $i\in I\setminus \{1\}$. 
Node 2 does not generate request, i.e., $\lambda_{(2,i)} = 0, \forall i\in I$. 
The cache capacities are $c_1=1$ and $c_2=I-1$.}\label{fig:PoA2}
\end{figure}

\begin{lemma}
There exists a selfish caching game for which the PoA is arbitrarily close to $0$.
\end{lemma}

\begin{proof}
We construct an example where PoA approaches $0$, as shown in Figure \ref{fig:PoA2}. 
In this example, the socially optimal caching strategy is for node 1 to cache content item 1 and for node 2 to cache content items $2$ to $I$. 
The optimal social welfare, i.e., aggregate caching gain, is\footnote{The superscript ``SO'' represents socially optimal.} 
\begin{equation*}
\textstyle G^{\rm SO}=\lambda_{(1,1)}(w_{21}+w_{32})+\sum_{i=2}^I \lambda_{(1,i)}w_{32}.
\end{equation*} 
There may exist more than one PSNE. 
One is that node 1 caches item 1 and node 2 caches none of the content items (since node $2$ has no request of its own). 
The social welfare achieved by this PSNE is $\textstyle G^{\rm NE}=\lambda_{(1,1)}(w_{21}+w_{32}).$ 
We have
\begin{equation*}
\textstyle \frac{G^{\rm NE}}{G^{\rm SO}} =\frac{1}{1+\sum_{i=2}^I\frac{\lambda_{(1,i)}w_{32}}{\lambda_{(1,1)}(w_{21}+w_{32})}  }.
\end{equation*}
When $w_{32} \gg w_{21}$ and $\epsilon \to 0$, we have $\frac{\lambda_{(1,i)}w_{32}}{\lambda_{(1,1)}(w_{21}+w_{32})} \to 1$ and $\frac{G^{\rm NE}}{G^{\rm SO}} \to \frac{1}{I}$, 
which goes to $0$ as $I$ becomes very large. 
Since PoA measures the worst case ratio between any PSNE and the social optimal solution, the PoA will be no larger than ${G^{\rm NE}}/{G^{\rm SO}}$ and hence can be arbitrarily close to $0$.
\end{proof}

\subsection{Cache Paradox}

In practice, one way to improve the aggregate caching gain in the network is to add extra cache nodes. 
However, we identify the following cache paradox.

\begin{lemma}
In the selfish caching game, adding extra cache nodes can make the PoA worse. 
\end{lemma}

\begin{proof}
Consider a caching network with two nodes in Figure \ref{fig:Paradox} (left subfigure), where node 2 is the designated server for two content items. 
Assume $c_1=1$ and $\lambda_{(1,1)}>\lambda_{(1,2)}>0$. 
At the equilibrium, node 1 caches item 1, which is also socially optimal. 
Hence, $PoA = 1$. 

Now we add an extra cache node, i.e., node 3 (see the right subfigure in Figure \ref{fig:Paradox}), where $c_3=1$ and $\lambda_{(3,1)}=\lambda_{(3,2)}=0$. 
Assume $w_{21}=w_{31}+w_{23}$, $w_{31}>0$, and $w_{23}>0$. 
Then one equilibrium is that node 1 caches item 1 and node 3 caches nothing. 
However, the socially optimal strategy is that node 1 caches item 1 and node 3 caches item 2. 
Hence, the PoA with node 3 satisfies
$$PoA'=\frac{\lambda_{(1,1)} (w_{31}+w_{23}) }{ \lambda_{(1,1)} (w_{31}+w_{23})+\lambda_{(1,2)} w_{23} } < 1.$$
Intuitively, adding node 3 does not change the social welfare achieved at the equilibrium, but increases the optimal social welfare, and hence makes the PoA worse. 
\end{proof}

\begin{figure}[t]
\centering
\includegraphics[width=0.45\textwidth]{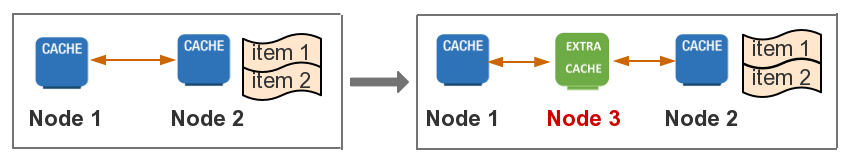}
\caption{An example where adding an extra cache node makes the PoA worse.}\label{fig:Paradox}
\end{figure}

\subsection{Bound on PoA}

In this section, we show that under reasonable constraints of request patterns and paths, the selfish caching game belongs to a class of games that we call \emph{$\alpha$-scalable valid utility games,} and the PoA is bounded by the length of the longest request forwarding path in the network.

Recall that $Z_s=\{i:x_{si}=1,i\in\I\}$ represents the set of content items cached by node $s\in V$. 
For convenience, we express the caching gain of node $s$ as $g_s(Z_s,Z_{-s})$, and the aggregate caching gain of the network as $G(Z)=\sum_{s\in V}g_s(Z_s,Z_{-s})$.

We first define the \emph{valid utility games} (for general games not restricted to selfish caching games), introduced by Vetta in \cite{ValidUtilityGame}.

\begin{definition}[Valid Utility Game \cite{ValidUtilityGame}]
A game (with social function $\gamma(\cdot)$ and individual payoff functions $f_s(\cdot), \forall s\in V$)\footnote{Note that the social function can be any objective that the network aims to optimize, and may not be the summation of individual players' payoff functions.} is a valid utility game if the following three properties are satisfied:
\begin{enumerate}
\item The social function $\gamma(\cdot)$ is non-decreasing and submodular. 
Mathematically, for every content item $i\in \I$ and for any subsets $Z, Z'$ such that $Z \subseteq Z'$,
\begin{align}
& \gamma(Z) \leq \gamma(Z'), \label{eq:prop21} \\
& \gamma(Z \cup \{i\}) - \gamma(Z) \geq \gamma(Z' \cup \{i\}) - \gamma(Z'). \label{eq:prop22}
\end{align}
\item The sum of players' payoff functions $f_s(\cdot)$ for any strategy profile $\boldsymbol{x}$ should be no larger than the social function $\gamma(\cdot)$: 
\begin{equation}\label{eq:prop1}
\textstyle \sum_{s\in V} f_s(\boldsymbol{x}_s,\boldsymbol{x}_{-s}) \leq \gamma(\boldsymbol{x}).
\end{equation}
\item The payoff of a player is no less than the difference between the social function when the player participates and that when it does not participate
\begin{equation}\label{eq:prop3}
f_s(\boldsymbol{x}_s,\boldsymbol{x}_{-s}) \geq \gamma(\boldsymbol{x}_s,\boldsymbol{x}_{-s}) - \gamma(\boldsymbol{0},\boldsymbol{x}_{-s}).
\end{equation}
\end{enumerate}
\end{definition}

Vetta in \cite{ValidUtilityGame} proved that the PoA of a valid utility game is bounded by $2$. 
In the following, we define a new class of games called \emph{$\alpha$-scalable valid utility games,} which generalizes the notation of valid utility games.

\begin{definition}[$\alpha$-Scalable Valid Utility Game]
A game is an $\alpha$-scalable valid utility game if it satisfies the two properties in \eqref{eq:prop21}, \eqref{eq:prop22}, and \eqref{eq:prop1}, and 
the payoff of a player is no less than the product of a positive constant $\alpha$ and the difference between the social function when the player participates and that when it does not participate:
\begin{equation}\label{eq:svug}
f_s(\boldsymbol{x}_s,\boldsymbol{x}_{-s}) \geq \frac{1}{\alpha} \cdot \Big[ \gamma(\boldsymbol{x}_s,\boldsymbol{x}_{-s}) - \gamma(\boldsymbol{0},\boldsymbol{x}_{-s}) \Big].
\end{equation}
\end{definition}

Note that the valid utility game is a special case of the $\alpha$-scalable valid utility games with $\alpha=1$.
We show that the selfish caching game is an $\alpha$-scalable valid utility game with $\alpha=\max_{v\in V,i\in\I}|p^{(v,i)}|-1$ when the following two properties are satisfied.

\begin{definition}[Homogeneous Request Pattern Property]\label{assum:arrival}
The request arrival processes for content item $i\in\I$ at different nodes are the same, i.e.,
\begin{equation}\label{eq:homolambda}
\lambda_{(s,i)}=\lambda_i, \forall s \in V, i\in \I.
\end{equation}
\end{definition}

The homogeneous request pattern property implies that each content item has a global popularity. 
Note that even under the homogeneous request pattern property, the popularity of different content items can be different, i.e., $\lambda_i \neq \lambda_j, i\neq j, i,j \in \I$.

\begin{figure}[t]
\centering
\includegraphics[width=0.26\textwidth]{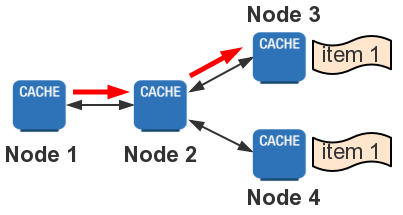}
\caption{An example that satisfies the path overlap property. Here, $p^{(2,1)}=(2,3)$ and $p^{(1,1)}=(1,2,3)$.}\label{fig:PathOverlap}
\end{figure}

\begin{definition}[Path Overlap Property]\label{assum:path}
If node $s$ is on path $p^{(v,i)}$, then starting from node $s$, path $p^{(v,i)}$ overlaps with path $p^{(s,i)}$, i.e., 
\begin{equation}\label{eq:pathoverlap}
s \in p^{(v,i)} \Rightarrow p^{(s,i)} \subseteq p^{(v,i)}.
\end{equation}
\end{definition}

Figure \ref{fig:PathOverlap} shows an example that satisfies the path overlap property. 
Note that the path overlap property is naturally satisfied when each node chooses a unique shortest path to fetch content items.

\begin{theorem}\label{theo:scalablevalid}
The selfish caching game with the homogeneous request pattern and path overlap properties on caching graphs with no mixed request loop is an $\alpha$-scalable valid utility game where 
\begin{equation}
\textstyle \alpha=\max_{v\in V,i\in\I}|p^{(v,i)}|-1.
\end{equation}
\end{theorem}

\begin{proof}
See Appendix C. 
\end{proof}

In the following theorem, we show that when the selfish caching game is an $\alpha$-scalable valid utility game, the PoA is bounded by the length of the longest request forwarding path in the network.

\begin{theorem}\label{theo:PoAalpha}
When the selfish caching game is an $\alpha$-scalable valid utility game, the PoA satisfies
\begin{equation}\label{eq:PoAalpha}
PoA \geq \frac{1}{1+\alpha} = \frac{1}{\max_{v\in V,i\in\I}|p^{(v,i)}|}.
\end{equation}
\end{theorem}

\begin{proof}
See Appendix D. 
\end{proof}

The PoA bound decreases with $\alpha$. 
The intuition is that as the length of the request forwarding path increases, the selfish behaviors of the intermediate nodes on a request forwarding path affect more succeeding nodes.
The above performance guarantee is true for general caching networks with an arbitrary topology. 
However, given a caching network with a known topology and network parameters, we can further explore the network structure and derive a better bound for PoA. 
This is achieved by characterizing the discrete curvature of the social function, as discussed next.

\begin{figure}[t]
\centering
\includegraphics[width=0.36\textwidth]{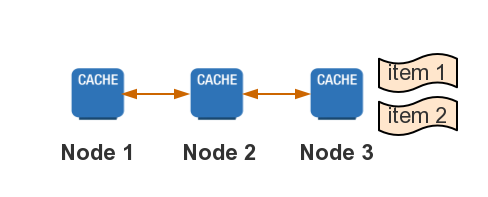}
\caption{An example to calculate the value of $\delta(G)$. We assume $\lambda_{(v,i)}=\lambda_i, \forall v\in V, i\in\I$. According to \eqref{eq:deltaGcalculation}, we have $\delta(G)=\frac{1}{1+w_{21}/w_{32}} \in [0,1].$ When $w_{21}/w_{32} \to 0$, we have $\delta(G) \to 1$; when $w_{21}/w_{32} \to \infty$, we have $\delta(G) \to 0$. }\label{fig:delta}
\end{figure}

\subsection{PoA and the Discrete Curvature of the Social Function}

To understand how the discrete curvature \cite{ValidUtilityGame} of the social function will affect our PoA analysis, we first introduce the discrete derivative. 
For a set function $G(\cdot)$, we define the discrete derivative at $Y$ in the direction $Z$ as $$ G'_Z(Y)=G(Y \cup Z) - G(Y) .$$
We define the \emph{discrete curvature} of a non-decreasing, submodular social function $G(\cdot)$ to be
\begin{equation}\label{eq:deltaGcalculation}
\delta(G)=\max_{\forall s\in V:G'_{Z_s}(\emptyset) > 0} \frac{G'_{Z_s}(\emptyset)-G'_{Z_s}(\mathcal{I}^{|V|}-Z_s)}{G'_{Z_s}(\emptyset)} \in [0,1],
\end{equation}
where $\mathcal{I}^{|V|}-Z_s$ represents the caching strategy profile (which can be infeasible) under which node $s$ caches content items in set $\I\setminus Z_s$ while all other nodes cache all content items in set $\I$. 
Figure \ref{fig:delta} shows an example to calculate the value of $\delta(G)$.

Given the discrete curvature of the social function, we can obtain a better bound for the PoA of the selfish caching game.

\begin{theorem}\label{theo:deltaG}
Given the discrete curvature $\delta(G)$ of the social function, for any selfish caching game with the homogeneous request pattern and path overlap properties on caching graphs with no mixed request loop, the PoA satisfies
\begin{equation}\label{eq:PoAdeltaG}
PoA \geq \frac{1}{\alpha+\delta(G)}.
\end{equation}
\end{theorem}

\begin{proof}
See Appendix E. 
\end{proof}

The PoA bound decreases with $\delta(G)$. 
The intuition is that under a larger $\delta(G)$, the selfish behavior of a cache node has a greater impact on the achieved social welfare. 
Since $\delta(G)\in [0,1]$ exploits the curvature property of the given network structure, the performance guarantee in \eqref{eq:PoAdeltaG} is better than the one in \eqref{eq:PoAalpha}.

\section{Selfish Caching Games with Unequal-Sized Items}\label{sec:Approx}

In this section, we analyze more general selfish caching games with unequal-sized items, which correspond to the practical applications such as video streaming services over HTTP (e.g., Netflix and Hulu) that split each video into segments of lengths from 2 to 10 seconds \cite{huang2012confused}.  
We show that the caching gain maximization problem for each cache node is NP-hard. 
We further generalize the model by considering that each cache node has limited computational capability to solve its caching gain maximization problem. 
This may lead to an approximate PSNE. 
We analyze the existence and efficiency of an approximate PSNE under these two generalizations.

\subsection{Game Modeling}

Let $L_i$ denote the size of content item $i$, for all $i\in\I$. 
We define the selfish caching game with unequal-sized items as follows:

\begin{game}[Selfish Caching Game with Unequal-Sized Items]\label{gameapp}
$ $
\begin{itemize}
\item Players: the set $V$ of cache nodes on the caching graph G(V,E); 
\item Strategies: the caching strategy $\boldsymbol{x}_s=\{x_{si}: \forall i\in\I\}$ for each cache node $s\in V$, where $x_{si}\in\{0,1\}$ and $\sum_{i\in\I}L_ix_{si} \leq c_s$;
\item Payoffs: the caching gain $g_s(\boldsymbol{x}_s,\boldsymbol{x}_{-s})$ for each $s\in V$.
\end{itemize}
\end{game}

In the following, we will show that for each cache node $s\in V$, given fixed $\boldsymbol{x}_{-s}$, its caching gain maximization problem is equivalent to a knapsack problem. 
For node $s \in V$, the caching gain in \eqref{eq:CachingGain} can be equivalently written as
\begin{equation}\label{eq:gsapp}
\begin{aligned}
g_s(\boldsymbol{x}_s,\boldsymbol{x}_{-s}) &= \sum_{i\in \I} \lambda_{(s,i)} \sum_{k=1}^{|p^{(s,i)}|-1}w_{p_{k+1}p_k} ( 1-\prod_{k'=2}^k(1-x_{p_{k'}i}) ) \\
& + \sum_{i\in \I } x_{si} \cdot \lambda_{(s,i)} \sum_{k=1}^{|p^{(s,i)}|-1}w_{p_{k+1}p_k} \prod_{k'=2}^k(1-x_{p_{k'}i}).
\end{aligned}
\end{equation}
Given fixed $\boldsymbol{x}_{-s}$, the first term in \eqref{eq:gsapp} is a constant, while the second term in \eqref{eq:gsapp} depends on $\boldsymbol{x}_s$. 
Define weight 
\begin{equation}\label{eq:weightqsi}
\textstyle q_{si}(\boldsymbol{x}_{-s})=\lambda_{(s,i)} \sum_{k=1}^{|p^{(s,i)}|-1}w_{p_{k+1}p_k} \prod_{k'=2}^k(1-x_{p_{k'}i}).
\end{equation}
Intuitively, $q_{si}(\boldsymbol{x}_{-s})$ represents the routing cost for request $(s,i)$ under $\boldsymbol{x}_{-s}$ if $x_{si}=0$.
Given fixed $\boldsymbol{x}_{-s}$, the caching gain maximization problem of node $s$ is equivalent to the following knapsack problem:
\begin{equation}\label{prob:knapsack}
\begin{aligned}
& \max_{\boldsymbol{x}_s} ~\sum_{i\in \I } x_{si} \cdot q_{si}(\boldsymbol{x}_{-s}) \\
& ~~\mbox{s.t.} ~~~~~ \textstyle  ~\sum_{i\in\I}L_ix_{si} \leq c_s,  x_{si}\in\{0,1\}, \forall i\in \I.
\end{aligned}
\end{equation}

Solving the knapsack problem \eqref{prob:knapsack} is NP-hard.\footnote{When $L_i=L_j, \forall i,j\in \I$, problem \eqref{prob:knapsack} is a max-weight knapsack problem, which is easy to solve and corresponds to the scenario with equal-sized items in Sections \ref{sec:model}--\ref{sec:PoA}.} 
In practice, each cache node has limited computational capability in a short time period (e.g., minutes or hours for which the request patterns remain unchanged \cite{garetto2015efficient}), and can only solve the knapsack problem \eqref{prob:knapsack} to an approximate solution $\hat{\boldsymbol{x}}_s=\{\hat{x}_{si}:\forall i\in \I\}$. 
There is extensive literature on the polynomial-time approximation algorithms for the knapsack problem \cite{KnapsackBook}. 
We present one such algorithm in Lines 4-12 of Algorithm \ref{algo:FindNECloud}, which achieves a $1/2$ approximation ratio (see Section 9.4.2 of \cite{KnapsackBook}).

Now we consider the general case where cache nodes obtain only a ${1}/{\beta}$ approximate solution with $\beta >1$ for problem \eqref{prob:knapsack}. 
This leads to the $\beta$-approximate PSNE of Game \ref{gameapp}.

\subsection{Existence of an Approximate PSNE}

A $\beta$-approximate PSNE is a strategy profile for which no player can improve its caching gain by a factor more than $\beta$ of its current caching gain by unilaterally changing its strategy.\footnote{An alternative notion of approximate PSNE (see, e.g., \cite{ApproximateNEadditive}) is based on an additive error, rather than the multiplicative error. Our definition is equally natural, and indeed more in line with the notion of price of anarchy in game theory \cite{CSgame, ValidUtilityGame, ApproximateNE}.} 
\begin{definition}[$\beta$-Approximate PSNE \cite{ApproximateNE}]
A pure strategy profile $\boldsymbol{x}^{\beta-NE}$ is a $\beta$-approximate PSNE if no player can find an alternative pure strategy with a payoff which is more that $\beta$ times its current payoff. That is for any player $s\in V$, 
\begin{equation}
g_s(\boldsymbol{x}_s',\boldsymbol{x}_{-s}^{\beta-NE}) \leq \beta \cdot g_s(\boldsymbol{x}_s^{\beta-NE},\boldsymbol{x}_{-s}^{\beta-NE}), \mbox{ for all feasible } \boldsymbol{x}_s'.
\end{equation}
\end{definition}

Next, we show that a $\beta$-approximate PSNE exists when the following property is satisfied:

\begin{definition}[Cloud Property]\label{assum:cloud}
All content items are stored in the same designated server node, i.e.,
\begin{equation}\label{eq:cloud}
|\S^i|=1 \mbox{ and } \S^i=\S^j, \forall i\neq j, i,j\in\I. 
\end{equation}
Furthermore, for each cache node $s\in V$, its request forwarding path for different content items is the same, i.e., 
\begin{equation}\label{eq:samepath}
p^{(s,i)}=p^s, \forall i\in\I .
\end{equation}
\end{definition}
In practice, a network in which all content items are stored in the cloud server satisfies \eqref{eq:cloud}. 
Note that \eqref{eq:samepath} is naturally satisfied when each node chooses the unique shortest path to fetch content items. 
Furthermore, \eqref{eq:samepath} is naturally satisfied in a tree topology.

In the following, we show that a $\beta$-approximate PSNE exists in Game \ref{gameapp} with the cloud property and the path overlap property in \eqref{eq:pathoverlap}.\footnote{The cloud property and the path overlap property are sufficient conditions for existence. Analyzing the sufficient and necessary conditions for existence of (approximate) PSNE on graphs is an open problem, and we will consider it in the future work.} 
Note that in the caching graph satisfying the cloud property and the path overlap property, there is no mixed request loop.

\begin{theorem}\label{theo:ExistenceBeta}
A $\beta$-approximate PSNE always exists in Game \ref{gameapp} with the cloud property and the path overlap property. 
\end{theorem}

\begin{proof}
See Appendix F. 
\end{proof}

The result holds in arbitrary-topology networks with heterogeneous content popularity. 
However, the complexity for finding an approximate PSNE may grow exponentially with the number of nodes and their strategies in general.

\subsection{Polynomial-Time Algorithm to Find an Approximate PSNE}

In this section, we propose a polynomial-time algorithm to find an approximate PSNE of Game \ref{gameapp}.

With the cloud property in \eqref{eq:cloud} and \eqref{eq:samepath}, and given designated server node $u$, the caching gain of each cache node $s\in V$ depends on not only $\boldsymbol{x}_s$ but also $\{\boldsymbol{x}_v: v\in V_s\}$, where $V_s=\{v: v\in p^s, v\neq s, v\neq u\}$ is the set of intermediate nodes on node $s$' request forwarding path $p^s$. 
We group nodes with the same number of intermediate nodes into one set, i.e., denote the set of nodes with $|V_s|=m$ by $\mathcal{V}^m=\{s\in V: |V_s|=m\}$ where $0 \leq m \leq |V|-2$. 
Note that since node $u$ stores all content items, its caching strategy does not affect other nodes.

If the path overlap property in \eqref{eq:pathoverlap} is satisfied, we know that if node $v\in V_s$, then $s\notin V_v$. 
That is, if node $v$ in on path $p^s$, then node $s$ is not on path $p^v$. 
Hence, for each node $s\in \mathcal{V}^m$ with $|V_s|=m$ intermediate nodes in set $V_s=\{v: v\in p^s, v\neq s, v\neq u\}$, every intermediate node $v\in V_s$ has a smaller number of intermediate nodes, i.e., $|V_v|<m$. 
This motivates us to find the equilibrium strategies for nodes in sets $\mathcal{V}^m, 0 \leq m \leq |V|-2,$  according to the increasing order of $m$.

\begin{algorithm}[t]
\LinesNumbered
\SetAlgoLined
\begin{small}
\KwIn{$G(V,E), \I, w_{uv}, \forall (u,v)\in E, \lambda_{(s,i)}$ and $p^{(s,i)}$, for all $s\in V, i\in \I$}
\KwOut{$\boldsymbol{x}^{\beta-NE}$}
Classify nodes into sets $\mathcal{V}^m$ for $0 \leq m \leq |V|-2$\;
\For{$m=0:|V|-2$}{
\For{$s\in \mathcal{V}^m$}{
Set $\hat{\boldsymbol{x}}_s=\boldsymbol{0}$, i.e., $\hat{x}_{si}=0, \forall i\in \I$\;
Relax problem \eqref{prob:knapsack} to a linear programming problem by relaxing $\boldsymbol{x}_s\in \{0,1\}^{|\I|}$ to $\widetilde{\boldsymbol{x}}_s\in [0,1]^{|\I|}$\; 
Compute an optimal solution $\widetilde{\boldsymbol{x}}_s^{\ast}$ of the LP-relaxation\;
Set $I_s=\{i:\widetilde{x}_{si}^{\ast}=1\}$ and $F_s=\{i:0<\widetilde{x}_{si}^{\ast}<1\}$\;
\eIf{$\sum_{i\in I_s}q_{(s,i)}(\boldsymbol{x}_{-s})>\max_{i\in F_s}q_{(s,i)}(\boldsymbol{x}_{-s})$}{
Set $\hat{x}_{si}=1, \forall i\in I_s$\;
}{
Set $\hat{x}_{sj}=1$ for $j=\arg \max_{i\in F_s}q_{(s,i)}(\boldsymbol{x}_{-s})$\;
}
}
}
Set $\boldsymbol{x}^{\beta-NE}=\hat{\boldsymbol{x}}$\;
\end{small}
\caption{Find $\beta$-Approximate PSNE}
\label{algo:FindNECloud}
\end{algorithm}

We propose a polynomial-time algorithm (Algorithm \ref{algo:FindNECloud}) to find a $\beta$-approximate PSNE. 
Specifically, We find the equilibrium strategies of nodes in sets $\mathcal{V}^m$ for $0 \leq m \leq |V|-2$ sequentially (Lines 1-2 of Algorithm \ref{algo:FindNECloud}). 
For example, for node $s\in \mathcal{V}^0$ such that $V_s=\emptyset$ (Line 3 of Algorithm \ref{algo:FindNECloud}), its $\beta$-approximate equilibrium strategy $\boldsymbol{x}_s^{\beta-NE}$ is the $\beta$-approximate solution to problem \eqref{prob:knapsack}, calculated by Lines 4-12 of Algorithm \ref{algo:FindNECloud}. 
Note that for a node $s\in\mathcal{V}^0$, $q_{si}(\boldsymbol{x}_{-s})=q_{si}$ is a constant independent of other nodes' strategies. 
For node $s\in \mathcal{V}^m$ with $1 \leq m \leq |V|-2$ (Line 3 of Algorithm \ref{algo:FindNECloud}), its equilibrium strategy is the $\beta$-approximate solution to problem \eqref{prob:knapsack}, calculated by Lines 4-12 of Algorithm \ref{algo:FindNECloud}.\footnote{Lines 4-12 of Algorithm \ref{algo:FindNECloud} achieve a $1/2$ approximation ratio of problem \eqref{prob:knapsack}, and hence $\beta=2$. Note that $\beta$ is identical across all nodes.} 
Note that its equilibrium strategy depends only on the caching strategies of nodes $v\in V_s$ in its intermediate node set with $|V_v| < m$, and hence $q_{si}(\boldsymbol{x}_{-s})=q_{si}(\{\boldsymbol{x}_{v}^{\beta-NE}:v\in V_s\})$. 
We continue this sequential process until all nodes decide their $\beta$-approximate equilibrium caching strategies. 
The resulting caching strategy profile is a $\beta$-approximate PSNE of Game \ref{gameapp}.

In the following theorem, we show that Algorithm \ref{algo:FindNECloud} can find a $\beta$-approximate PSNE of Game \ref{gameapp} in polynomial time.

\begin{theorem}\label{theo:ExistenceAPP}
For Game \ref{gameapp} with the cloud property and the path overlap property, Algorithm \ref{algo:FindNECloud} computes a $\beta$-approximate PSNE in $\mathcal{O}(|V||\I|)$ time. 
\end{theorem}

\begin{proof}
See Appendix G. 
\end{proof}

We next analyze the PoA of Game \ref{gameapp}.

\subsection{Price of Anarchy}

we show that the PoA for the $\beta$-approximate PSNE is bounded under the homogeneous request pattern property in \eqref{eq:homolambda}. 

\begin{theorem}\label{theo:PoAapp}
For Game \ref{gameapp} with the cloud property, the path overlap property, and the homogeneous request pattern property, the PoA for the $\beta$-approximate PSNE satisfies
\begin{equation}\label{eq:PoAbetaapp}
PoA^{\beta} \geq \frac{1}{1+\alpha \cdot \beta} = \frac{1}{1+ \beta \cdot \left( \max_{v\in V,i\in\I}|p^{(v,i)}| -1\right)}. 
\end{equation}
\end{theorem}

\begin{proof}
See Appendix H. 
\end{proof}

The performance guarantee holds for arbitrary caching networks. 
However, due to cache nodes' limited computational capabilities, the guarantee for the approximate PSNE in \eqref{eq:PoAbetaapp} is worse than the one in \eqref{eq:PoAalpha} for the equal-sized item case.

\section{Simulation Results}\label{sec:simu}

\begin{figure*}[t] 
\centering 
\begin{minipage}[t]{0.49 \linewidth}
\centering
\includegraphics[width=1\textwidth]{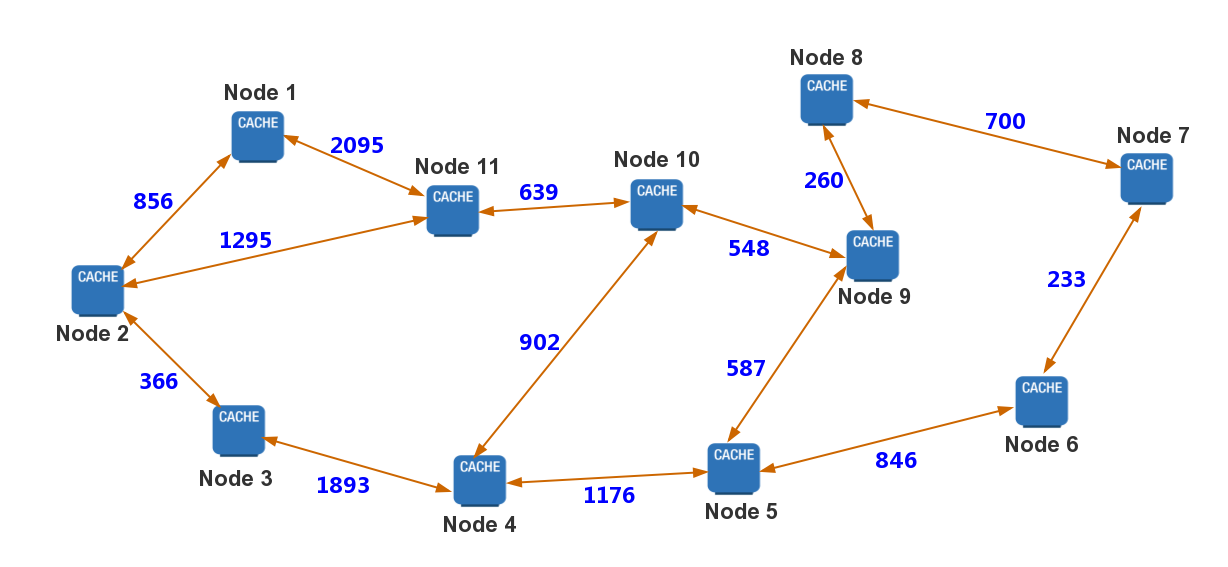}
\caption{Abilene network.}\label{fig:Abilene}
\end{minipage}
\begin{minipage}[t]{0.005 \linewidth}
~
\end{minipage}
\begin{minipage}[t]{0.49 \linewidth}
\centering
\includegraphics[width=1\textwidth]{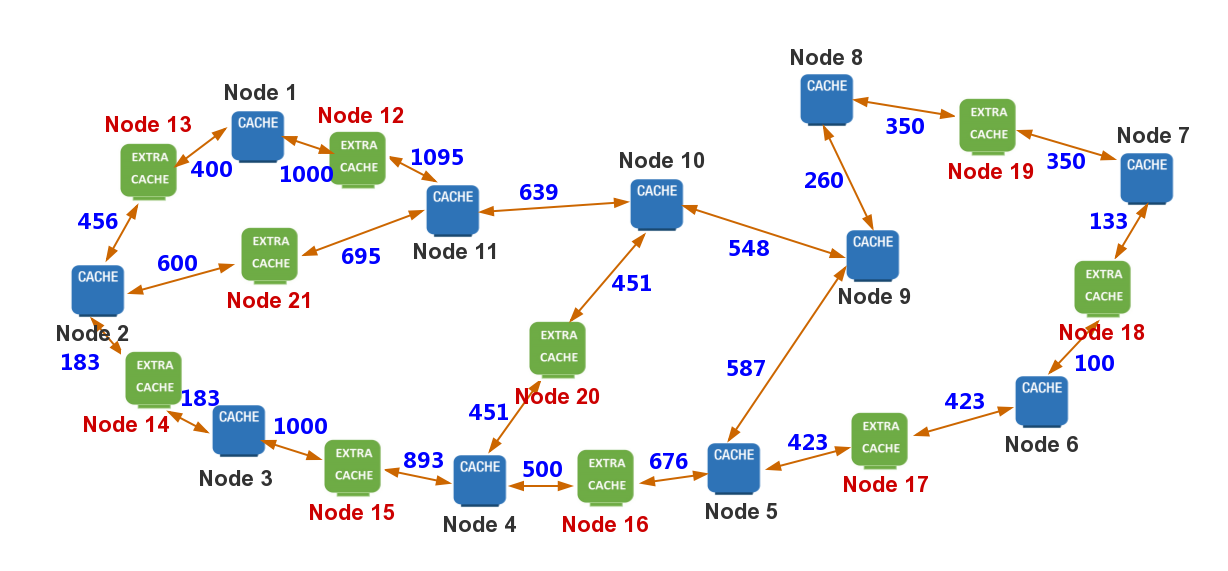}
\caption{Abilene network with extra nodes.}\label{fig:AbileneAdd}
\end{minipage}
\end{figure*}

\begin{figure*}[t] 
\centering 
\begin{minipage}[t]{0.23 \linewidth}
\centering
\includegraphics[width=1\textwidth]{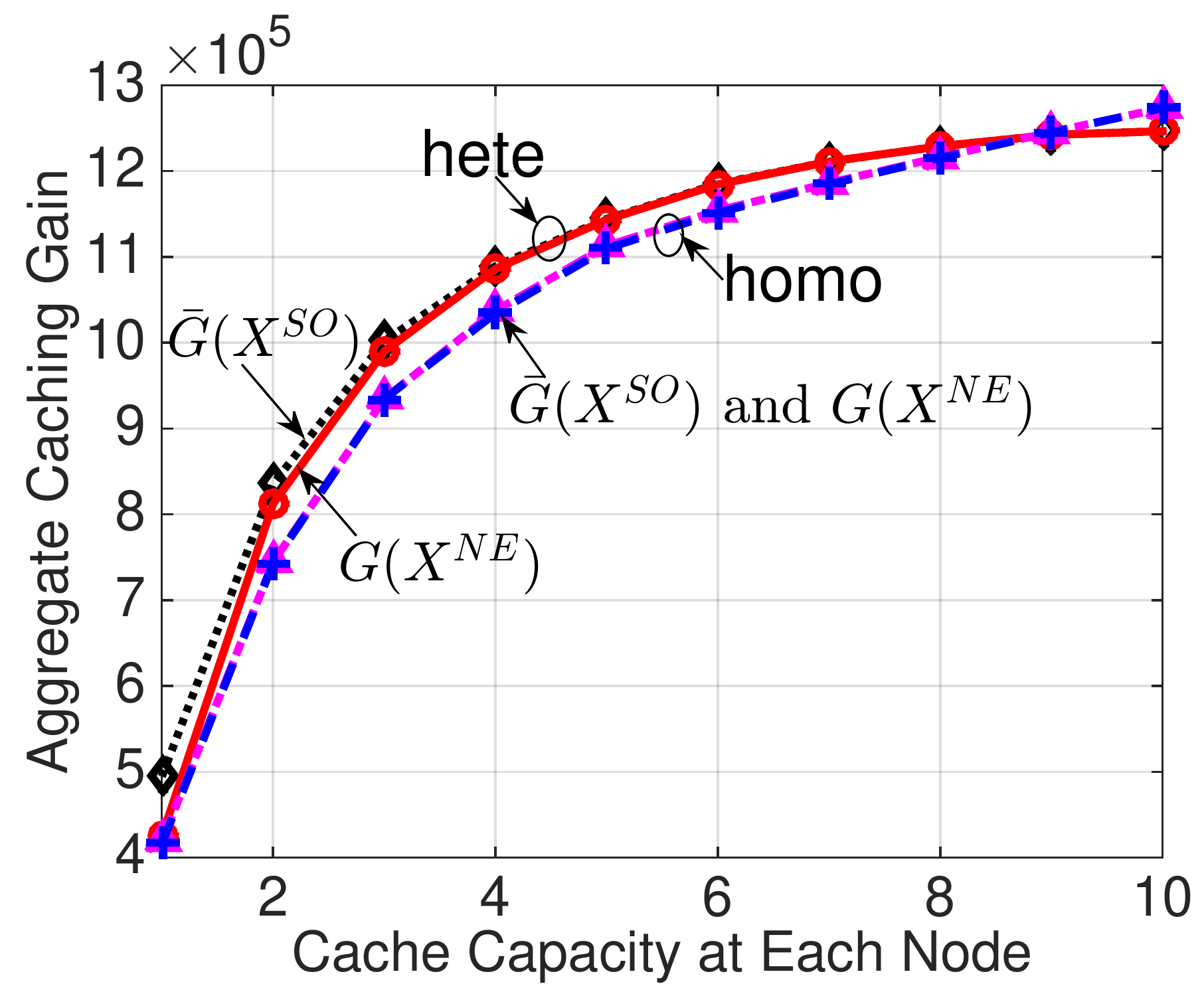}
\caption{$G(\cdot)$ vs. $c_v$, under both heterogeneous and homogeneous request patterns.}\label{fig:HomoLambda_Abilene10}
\end{minipage}
\begin{minipage}[t]{0.005 \linewidth}
~
\end{minipage}
\begin{minipage}[t]{0.23 \linewidth}
\centering
\includegraphics[width=1\textwidth]{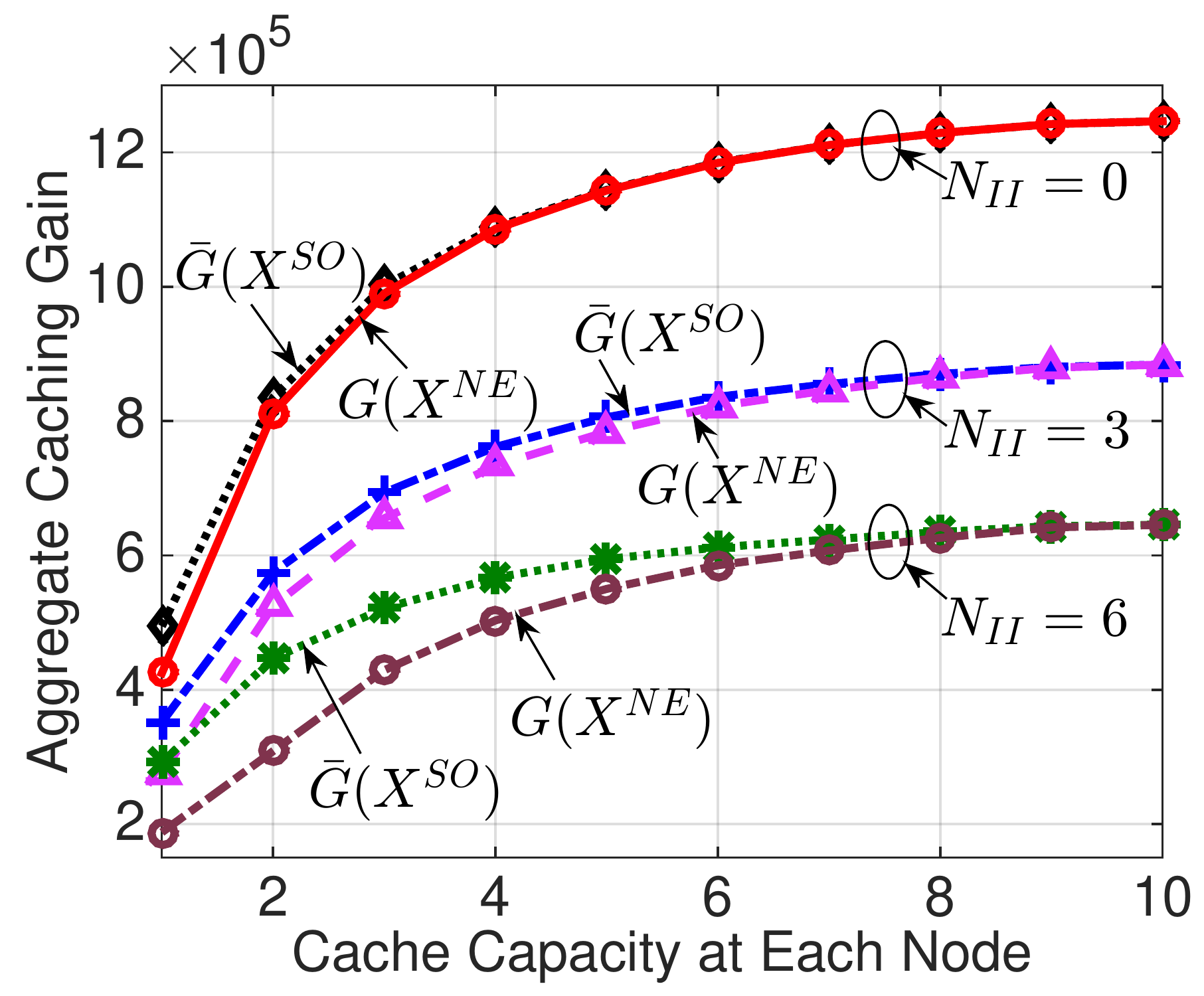}
\caption{$G(\cdot)$ vs. $c_v$, under different no. of Type-II nodes $N_{II}$.}\label{fig:IdleNodes_Abilene10}
\end{minipage}
\begin{minipage}[t]{0.005 \linewidth}
~
\end{minipage}
\begin{minipage}[t]{0.23 \linewidth}
\centering
\includegraphics[width=1.024\textwidth]{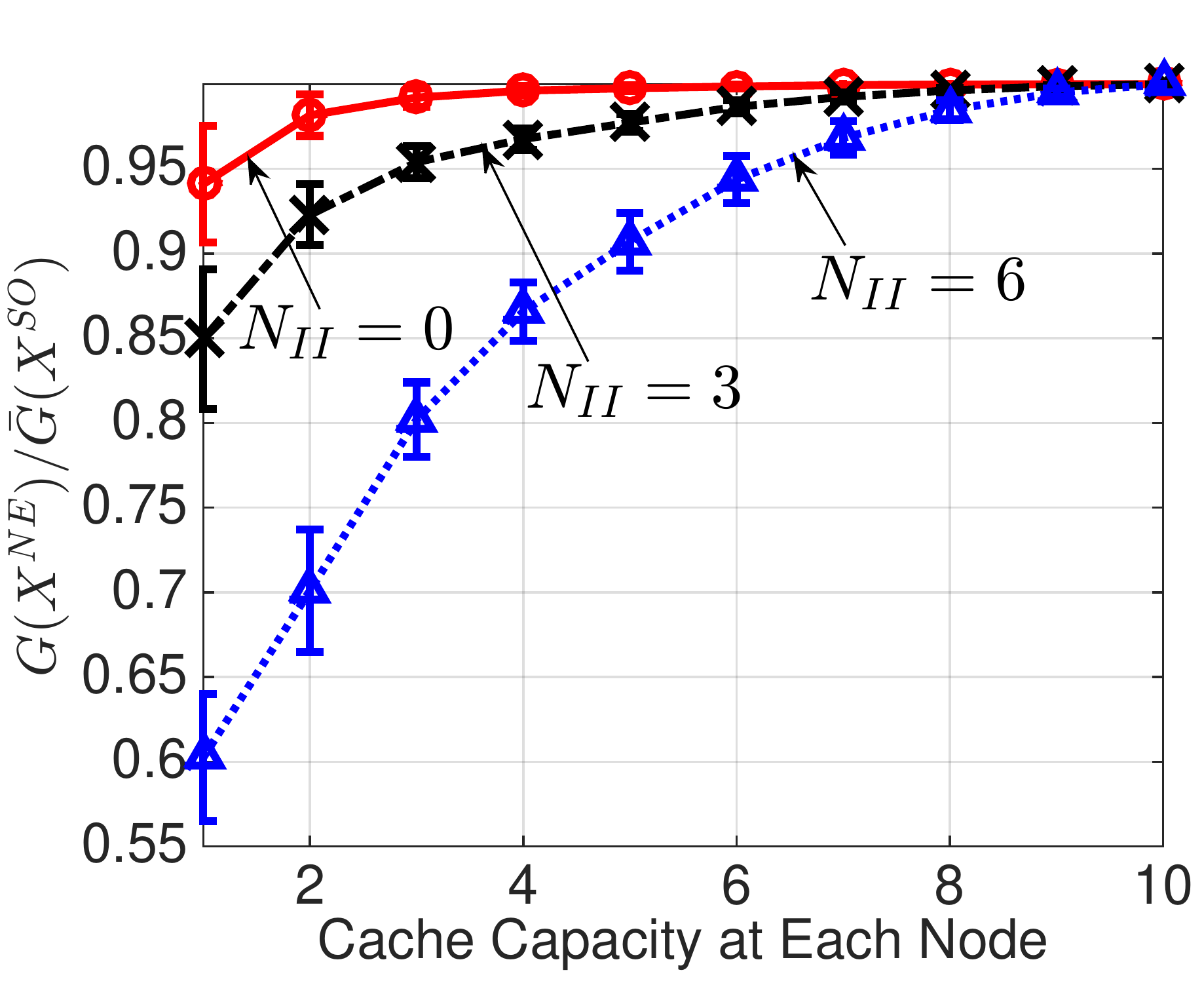}
\caption{$G(\boldsymbol{x}^{\rm NE})/\bar{G}(\boldsymbol{x}^{\rm SO})$ vs. $c_v$, under different $N_{II}$, for 100 trials.}\label{fig:100Trials}
\end{minipage}
\begin{minipage}[t]{0.005 \linewidth}
~
\end{minipage}
\begin{minipage}[t]{0.25 \linewidth}
\centering
\includegraphics[width=1.1\textwidth]{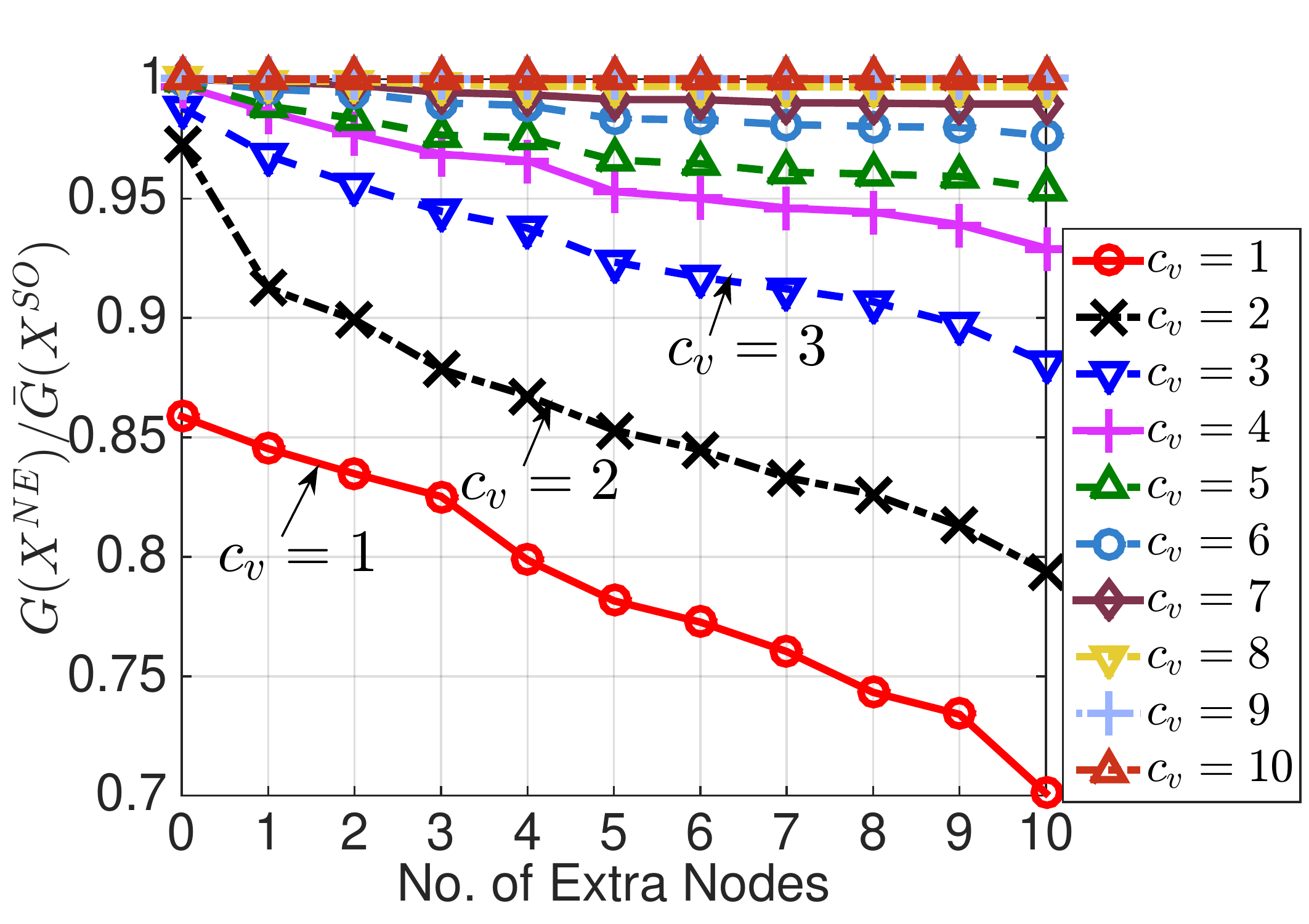}
\caption{$G(\boldsymbol{x}^{\rm NE})/\bar{G}(\boldsymbol{x}^{\rm SO})$ vs. no. of extra cache nodes, under different $c_v$.}\label{fig:ImpactExtraNodesTwoServer}
\end{minipage}
\end{figure*}

We perform simulations on networks including the Abilene network shown in Figure \ref{fig:Abilene}, the GEANT network shown in Figure \ref{fig:GEANT}, and the Grid topology shown in Figure \ref{fig:Grid}. 
Simulation results show that the performance of Nash equilibria improves with the cache capacity at each node, while it degrades with the number of nodes that do not generate content requests. 
Furthermore, adding extra cache nodes to the existing network can make the performance of Nash equilibria worse.

\textbf{Upper Bound of the Optimal Social Welfare.}
Since the social welfare maximization problem \eqref{prob:maxCG} is NP-hard, we calculate an upper bound for the optimal social welfare. 
Specifically, we relax problem \eqref{prob:maxCG} by relaxing the binary caching strategy $\boldsymbol{x}=\{x_{si}\in\{0,1\}: \forall s\in V,i\in\I\}$ to be a continuous caching probability strategy $\boldsymbol{\phi}=\{\phi_{si}\in[0,1]: \forall s\in V,i\in\I\}$ where $\sum_{i\in\I}\phi_{si}\leq c_s,\forall s\in V$, while keeping the objective function unchanged. 
The relaxed problem is 
\begin{equation}\label{prob:relaxedG}
 \textstyle \mbox{max}  ~ G(\boldsymbol{\phi})~\mbox{s.t.}  ~  \textstyle  \sum_{i\in\I}  \phi_{si} \leq c_s,  ~\phi_{si} \in [0,1], ~  \forall  s\in V, i\in\I .
\end{equation}
The relaxation objective function $G(\boldsymbol{\phi})$ is not concave, so \eqref{prob:relaxedG} is not a convex optimization problem. 
We approximate $G(\boldsymbol{\phi})$ by $L(\boldsymbol{\phi})$ below \cite{YehSigmetrics}:
\begin{equation*}
L(\boldsymbol{\phi})= \sum_{s\in V, i\in\I}\lambda_{(s,i)} \sum_{k=1}^{|p^{(s,i)}|-1}w_{p_{k+1}p_k}  \min\left\{ 1,\sum_{k'=1}^k \phi_{p_{k'}i} \right\}. 
\end{equation*}
Note that $L(\boldsymbol{\phi})$ is concave, and we can solve the following convex optimization problem in polynomial time. 
\begin{equation}\label{prob:maxL}
 \textstyle \mbox{max}  ~ L(\boldsymbol{\phi}) ~ \mbox{s.t.} ~  \textstyle  \sum_{i\in\I}  \phi_{si} \leq c_s,  ~\phi_{si} \in [0,1], ~  \forall  s\in V, i\in\I .
\end{equation}

We have the following result:
\begin{lemma}\label{lemma:SocialWelfareUB}
Let $\boldsymbol{x}^\ast$, $\boldsymbol{\phi}^\ast$, and $\boldsymbol{\phi}^{\ast\ast}$ be the optimal solutions to problems \eqref{prob:maxCG}, \eqref{prob:relaxedG}, and \eqref{prob:maxL}, respective. 
Then: 
\begin{equation}
G(\boldsymbol{x}^\ast) \leq G(\boldsymbol{\phi}^\ast) \leq L(\boldsymbol{\phi}^{\ast}) \leq  L(\boldsymbol{\phi}^{\ast\ast}) .
\end{equation}
\end{lemma}

\begin{proof}
See Appendix I. 
\end{proof}

Hence, $L(\boldsymbol{\phi}^{\ast\ast})$ serves as an upper bound for the optimal social welfare $G(\boldsymbol{x}^\ast)$. 
We define $\bar{G}(\boldsymbol{x}^{\rm SO})=L(\boldsymbol{\phi}^{\ast\ast})$.\footnote{The superscript ``SO'' represents socially optimal.}

In the following, we first perform simulations for the case with equal-sized content items and show the results in Figures \ref{fig:HomoLambda_Abilene10}--\ref{fig:IdleNodes_Ratios_Grid}, which validate the existence of a PSNE in Theorem \ref{theo:Existence} and the PoA analysis in Theorem \ref{theo:PoAalpha}. 
We then perform simulations for the case with unequal-sized content items and show the results in Figures \ref{fig:HomoLambda_AbileneApp} -- \ref{fig:ImpactExtraNodesApp}, which validate the existence of an approximate PSNE in Theorem \ref{theo:ExistenceBeta} and the PoA analysis of the approximate PSNE in Theorem \ref{theo:PoAapp}.

\textbf{Experiment Setup for the Abilene Network.} 
For the Abilene network shown in Figure \ref{fig:Abilene}, we take all edge costs from the Abilene network configuration \cite{AbileneTopology}.\footnote{We assume that the edge costs are symmetric.} 
We consider a set $\I=\{1,\ldots,10\}$ of content items \cite{YehSigmetrics}, where node 1 is the designated server of the first 6 content items and node 2 is the designated server of the remaining 4 content items. 
Each node chooses the shortest path to fetch every content item, following which there is no mixed request loop on the graph. 
We generate the arrival rates $\lambda_{(s,i)}, \forall s\in V, i\in \I$ uniformly at random in the interval $[0,10]$.

\begin{figure*}[t] 
 \centering 
 \begin{minipage}[t]{0.25 \linewidth}
 \centering
 \includegraphics[width=1\textwidth]{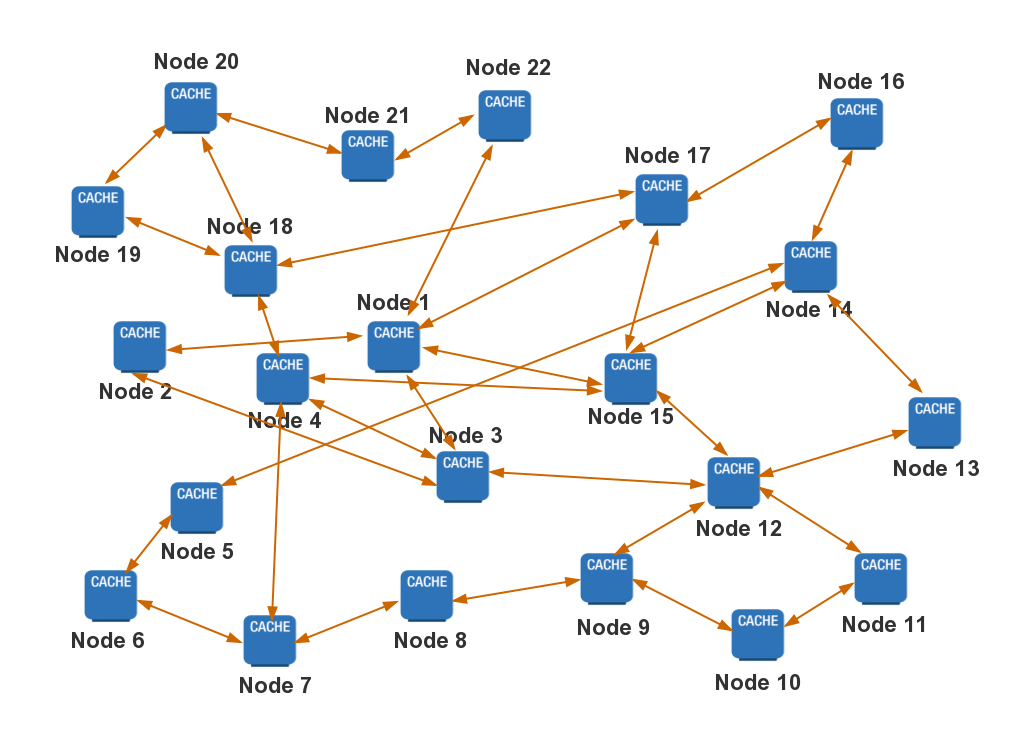}
 \caption{GEANT network.}\label{fig:GEANT}
 \end{minipage}
 \begin{minipage}[t]{0.005 \linewidth}
~
\end{minipage}
 \begin{minipage}[t]{0.23 \linewidth}
 \centering
 \includegraphics[width=1\textwidth]{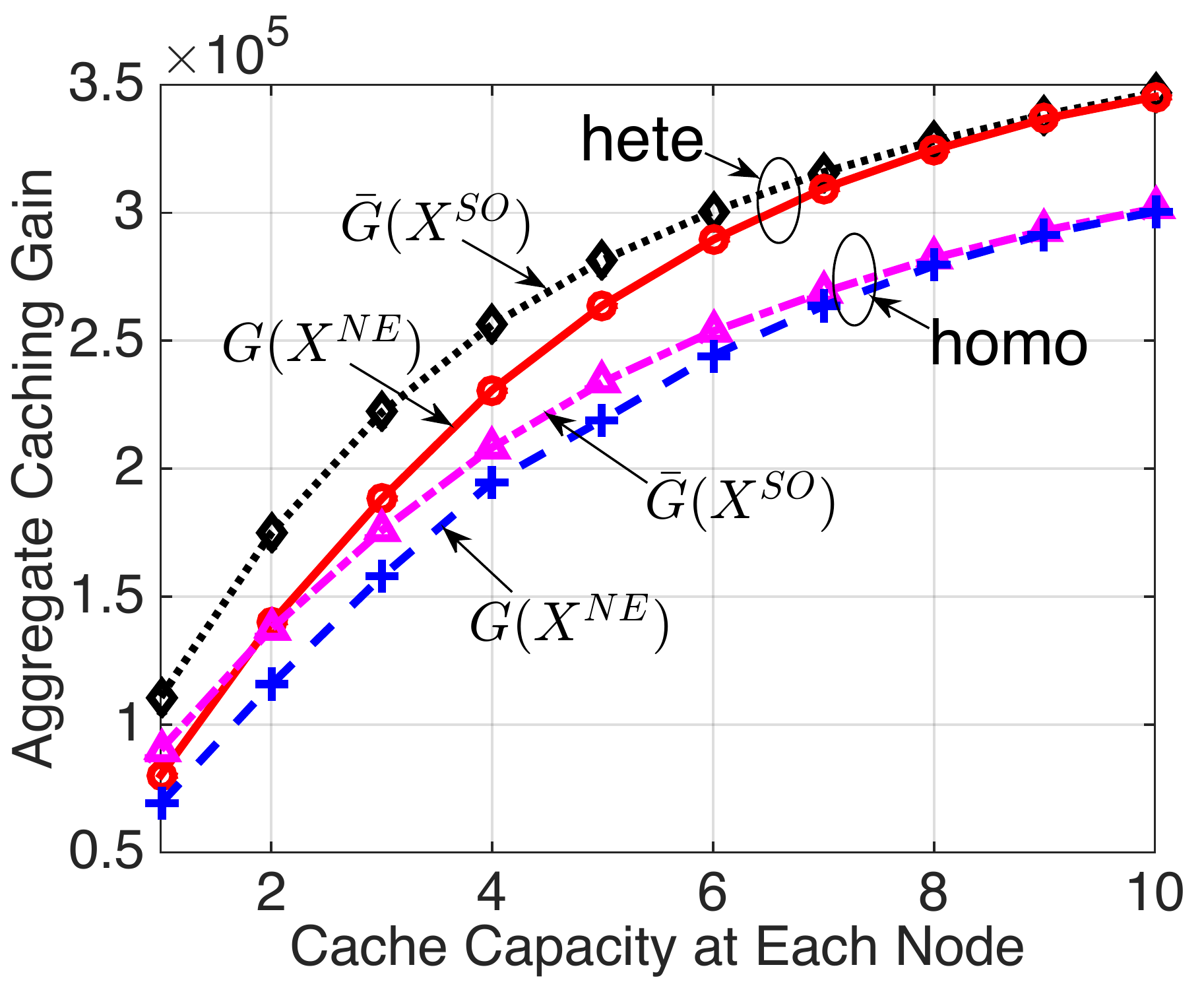}
 \caption{$G(\cdot)$ vs. $c_v$, under both heterogeneous and homogeneous request patterns.}\label{fig:HomoLambda_GEANT}
 \end{minipage}
 \begin{minipage}[t]{0.005 \linewidth}
~
\end{minipage}
 \begin{minipage}[t]{0.23 \linewidth}
 \centering
 \includegraphics[width=1\textwidth]{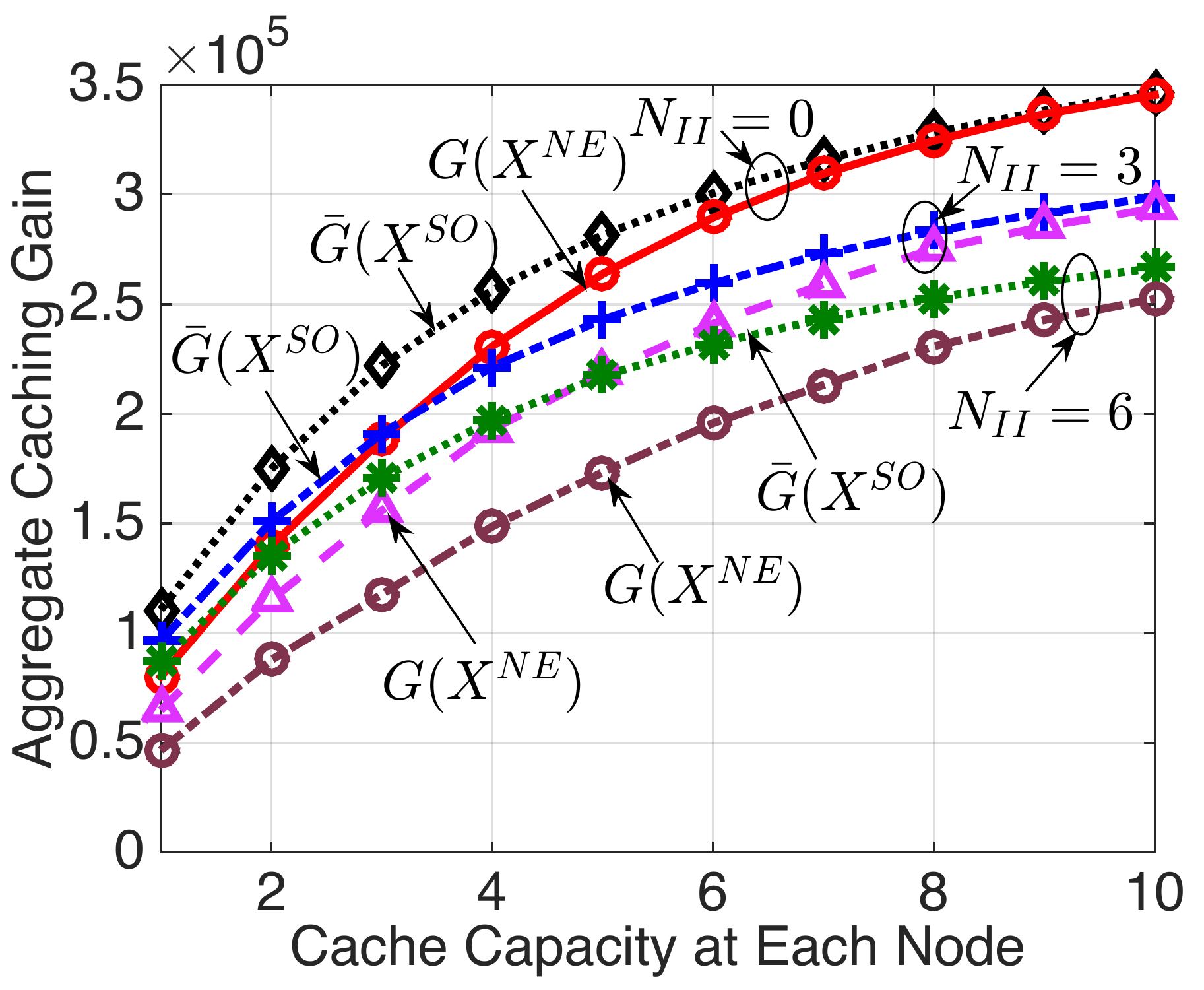}
 \caption{$G(\cdot)$ vs. $c_v$, under different no. of Type-II nodes $N_{II}$.}\label{fig:IdleNodes_GEANT}
 \end{minipage}
 \begin{minipage}[t]{0.005 \linewidth}
~
\end{minipage}
 \begin{minipage}[t]{0.23 \linewidth}
 \centering
 \includegraphics[width=1\textwidth]{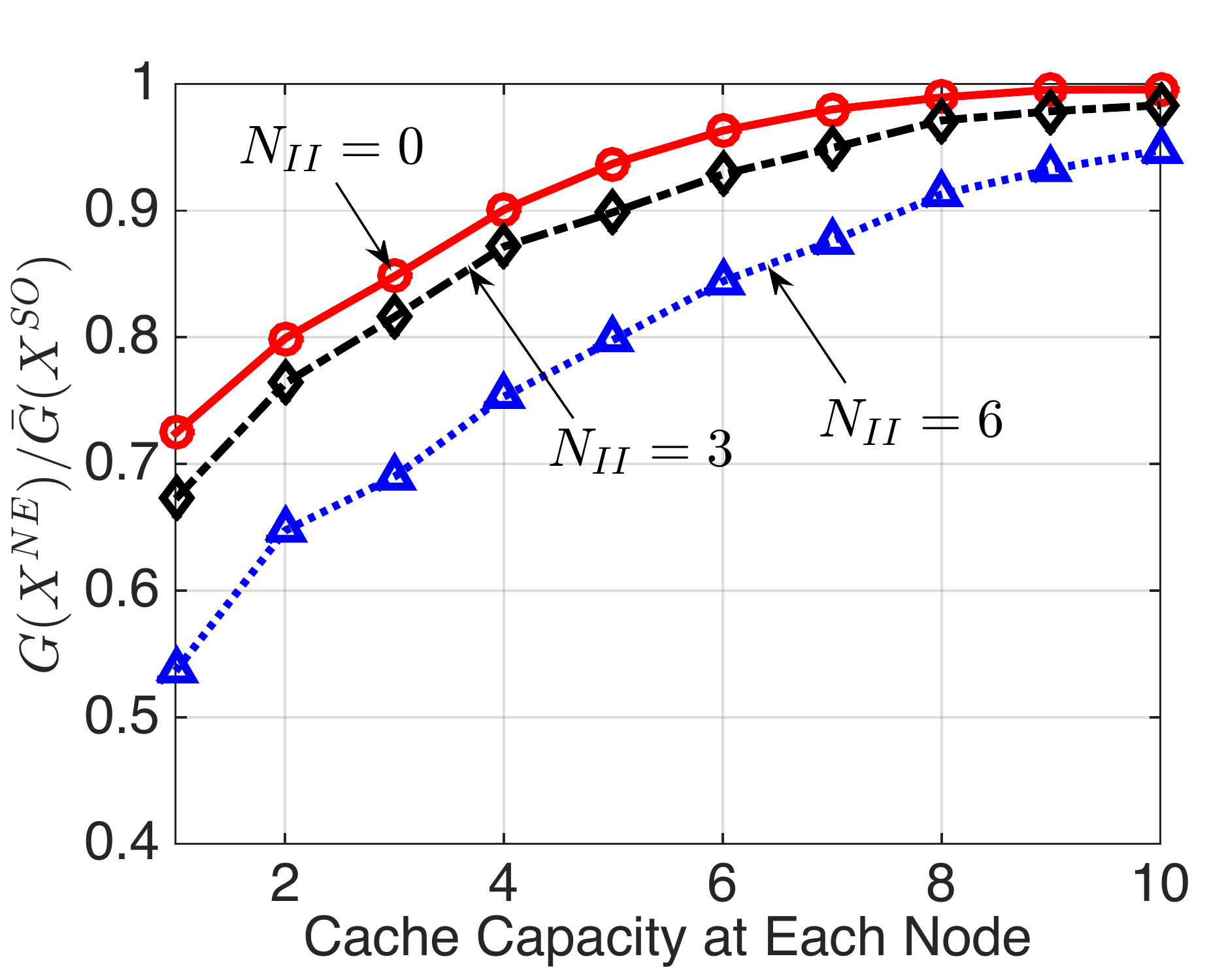}
 \caption{$G(\boldsymbol{x}^{\rm NE})/\bar{G}(\boldsymbol{x}^{\rm SO})$ vs. $c_v$, under different $N_{II}$.}\label{fig:IdleNodes_Ratios_GEANT}
 \end{minipage}
 \end{figure*}

 \begin{figure*}[t] 
 \centering 
 \begin{minipage}[t]{0.2 \linewidth}
 \centering
 \includegraphics[width=1\textwidth]{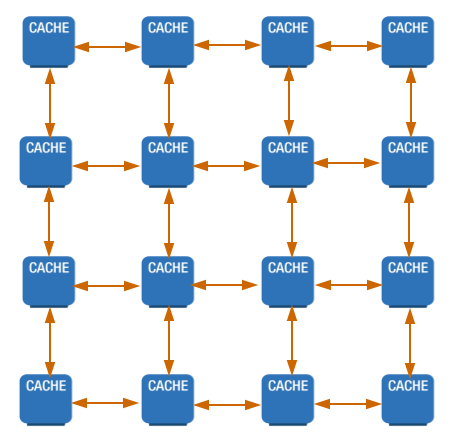}
 \caption{Grid topology.}\label{fig:Grid}
 \end{minipage}
 \begin{minipage}[t]{0.005 \linewidth}
~
\end{minipage}
 \begin{minipage}[t]{0.245 \linewidth}
 \centering
 \includegraphics[width=1\textwidth]{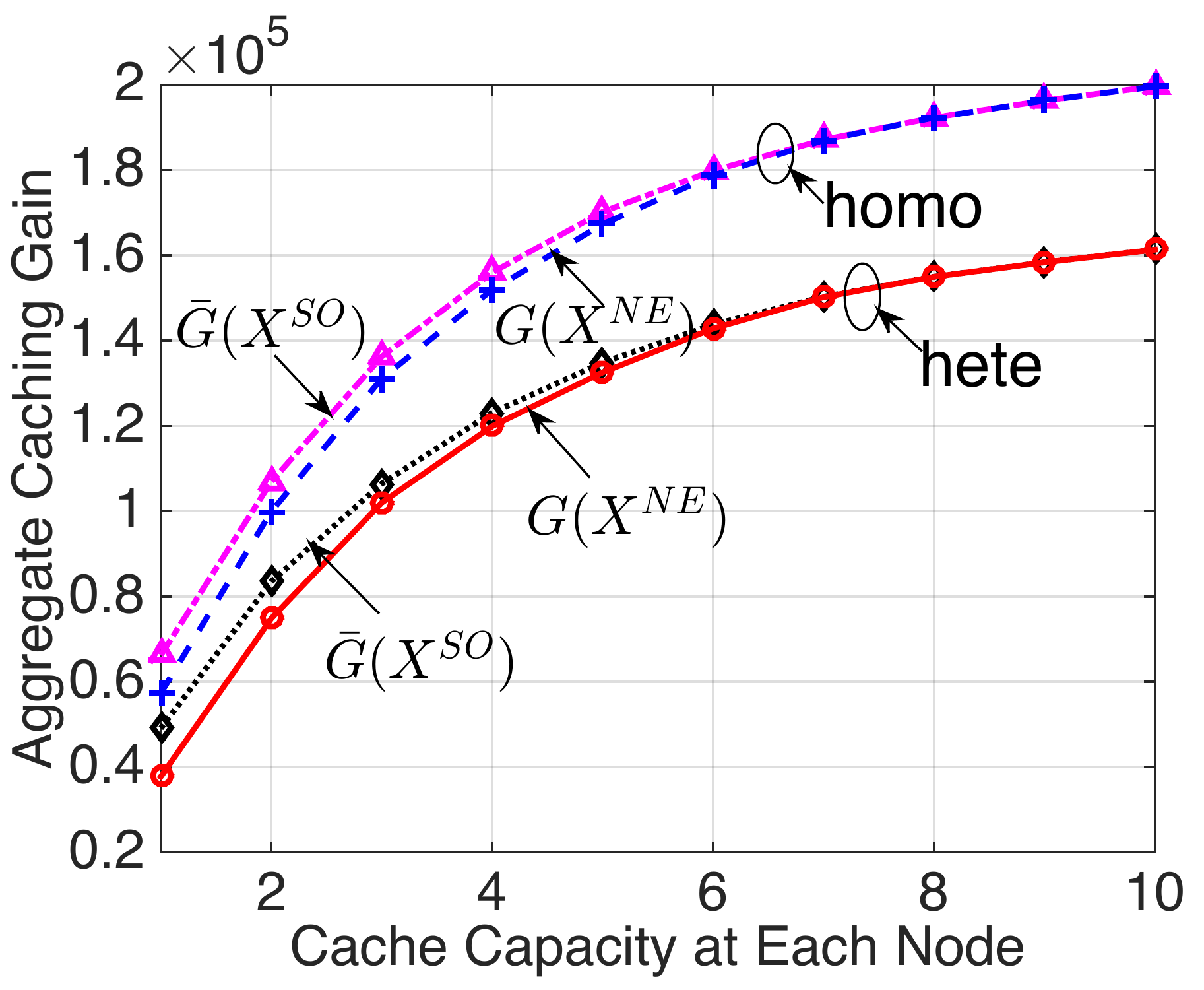}
 \caption{$G(\cdot)$ vs. $c_v$, under both heterogeneous and homogeneous request patterns.}\label{fig:HomoLambda_Grid}
 \end{minipage}
 \begin{minipage}[t]{0.005 \linewidth}
~
\end{minipage}
 \begin{minipage}[t]{0.245 \linewidth}
 \centering
 \includegraphics[width=1\textwidth]{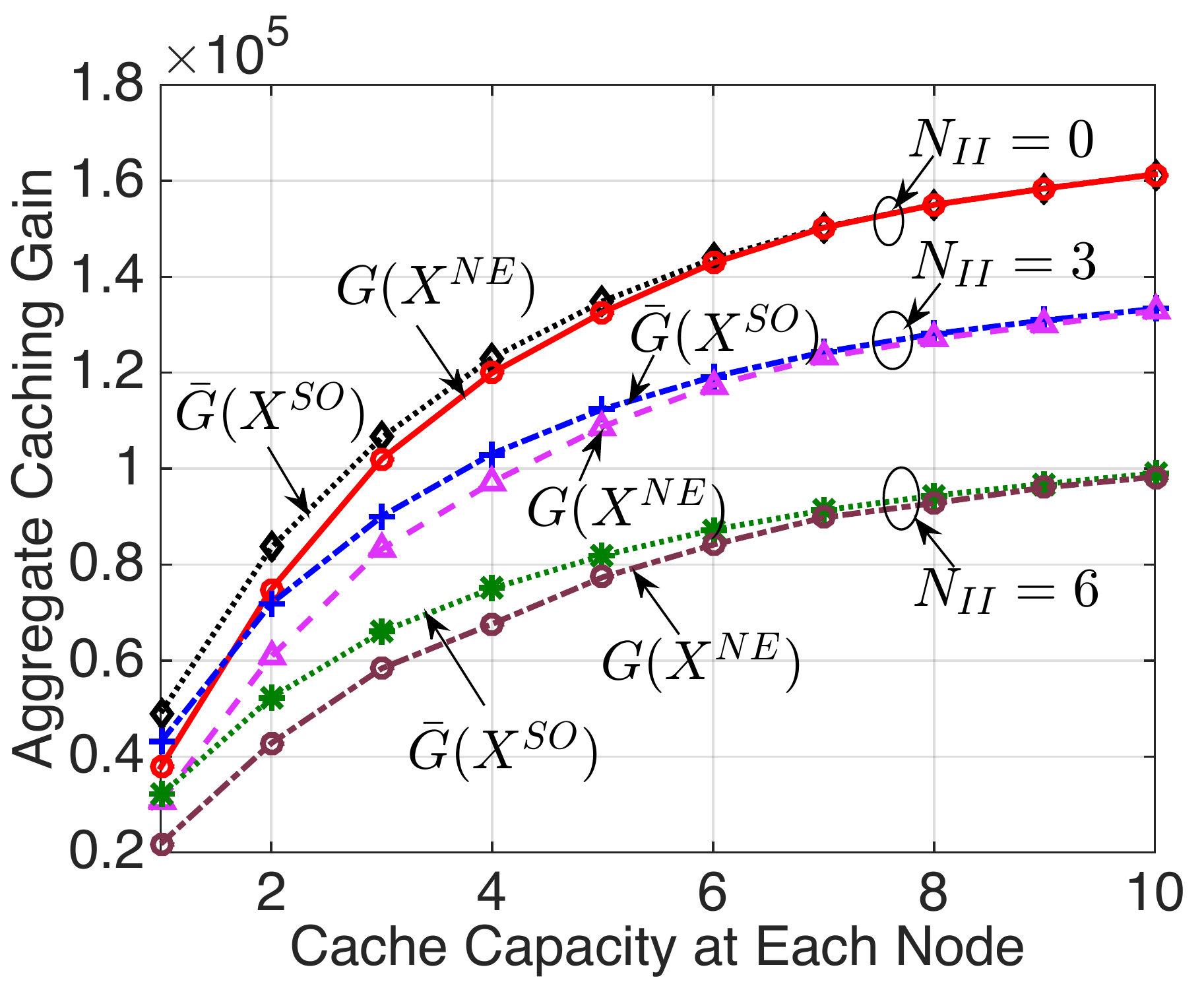}
 \caption{$G(\cdot)$ vs. $c_v$, under different no. of Type-II nodes $N_{II}$.}\label{fig:IdleNodes_Grid}
 \end{minipage}
 \begin{minipage}[t]{0.005 \linewidth}
~
\end{minipage}
 \begin{minipage}[t]{0.245 \linewidth}
 \centering
 \includegraphics[width=1\textwidth]{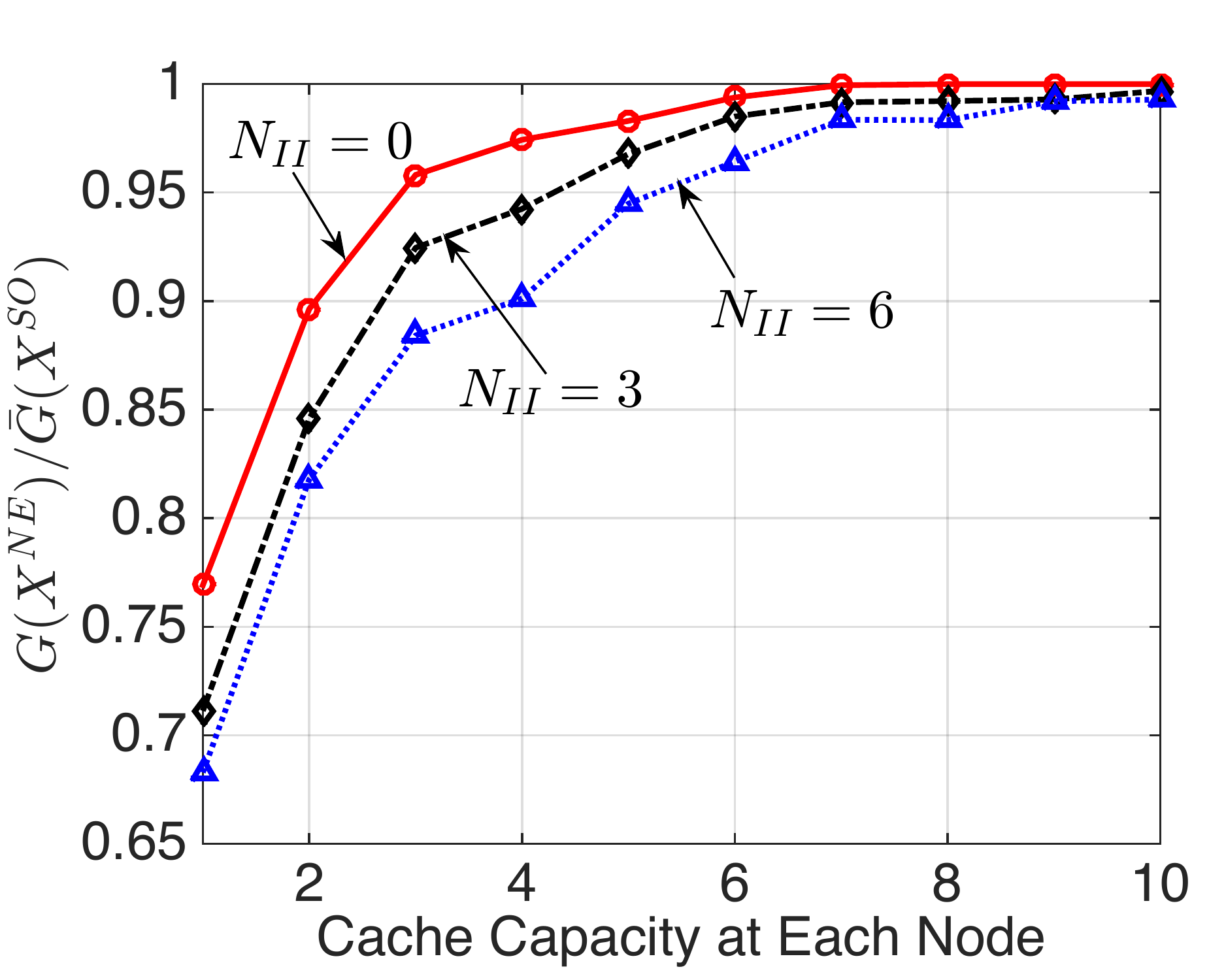}
 \caption{$G(\boldsymbol{x}^{\rm NE})/\bar{G}(\boldsymbol{x}^{\rm SO})$ vs. $c_v$, under different $N_{II}$.}\label{fig:IdleNodes_Ratios_Grid}
 \end{minipage}
 \end{figure*}

\textbf{Results in the Abilene Network.} 
Figure \ref{fig:HomoLambda_Abilene10} shows the aggregate caching gain $G(\boldsymbol{x}^{\rm NE})$ and $\bar{G}(\boldsymbol{x}^{\rm SO})$ under different cache capacities at each node,\footnote{We show the results for the case where cache nodes may have different cache capacities in Appendix J.} for the case with heterogeneous request patterns $\lambda_{(s,i)}$ (the upper two curves) and for the case with homogeneous request patterns $\lambda_{(s,i)}=\lambda_i, \forall s\in V, i\in \I$ (the lower two curves)\footnote{We take $\lambda_i=\sum_{s\in V}\lambda_{(s,i)}/V$ given the heterogeneous $\lambda_{(s,i)}, \forall s\in v,i\in\I$.}, respectively. 
We can see that the gap between $G(\boldsymbol{x}^{\rm NE})$ and $\bar{G}(\boldsymbol{x}^{\rm SO})$ under homogeneous $\lambda_i$ is smaller than the gap under heterogeneous $\lambda_{(s,i)}$. 
Thus, the homogeneous request pattern leads to better performance achieved by selfish caching behaviors in the Abilene network.

In practice, some cache nodes are intermediate routers which do not request for any content items. 
We define nodes with positive request rates as Type-I nodes (with a total number $N_{I}$), and nodes with no request as Type-II nodes (with a total number $N_{II}$). 
We show the impact of $N_{II}$ in Figure \ref{fig:IdleNodes_Abilene10}. 
We can see that the gap between $G(\boldsymbol{x}^{\rm NE})$ and $\bar{G}(\boldsymbol{x}^{\rm SO})$ decreases with the cache capacity at each node, while the gap increases with $N_{II}$. 
This implies that the impact of the selfish behaviors is mitigated when the cache resource increases, and the selfish behaviors of Type-II nodes degrade the (relative) performance of Nash equilibria (since the selfish Type-II nodes will not cache content items at equilibrium).

To understand the impact of the randomness of request arrival rates $\lambda_{(s,i)}$, we perform simulations on $100$ sets of randomly generated $\{\lambda_{(s,i)}: \forall s\in V, i\in \I\}$, and show the average ratios $G(\boldsymbol{x}^{\rm NE})/\bar{G}(\boldsymbol{x}^{\rm SO})$ of the $100$ trials in Figure \ref{fig:100Trials}, where the error bars represent the standard deviations. 
As is consistent with our observation from Figure \ref{fig:IdleNodes_Abilene10}, the performance of the Nash equilibria increases with the cache capacity, while decreases with $N_{II}$.

In practice, one direct way to improve the aggregate caching gain in the network is to add extra cache nodes.
To check the impact of extra caches on the performance of Nash equilibria, we sequentially add node 12, node 13, until node 21, shown in Figure \ref{fig:AbileneAdd}. 
We show the ratio $G(\boldsymbol{x}^{\rm NE})/\bar{G}(\boldsymbol{x}^{\rm SO})$ with different number of extra nodes in Figure \ref{fig:ImpactExtraNodesTwoServer}. 
We can see that adding more extra caches makes PoA worse. 
The reason is that adding extra cache nodes can improve the optimal social welfare, while it cannot improve the social welfare achieved by Nash equilibria due to the selfish nature of cache nodes. 
Hence the ``relative'' performance of the Nash equilibria (measured in terms of PoA) reduces.

\textbf{Results in the GEANT Network.} 
We perform simulations on the GEANT network shown in Figure \ref{fig:GEANT}. 
We consider a set $\I=\{1,\ldots,20\}$ of content items. 
We generate the cost on each edge uniformly at random from the interval $[1, 100]$. 
We show the performances corresponding to selfish behaviors in Figures \ref{fig:HomoLambda_GEANT}--\ref{fig:IdleNodes_Ratios_GEANT}. 
As in the Abilene network, the homogeneous request pattern leads to better (relative) performance achieved by selfish caching behaviors, and the ratio $G(X^{NE})/\bar{G}(X^{SO})$ increases with $c_v$ and decreases with $N_{II}$.

\begin{figure*}[t] 
\centering 
\begin{minipage}[t]{0.23 \linewidth}
\centering
\includegraphics[width=1\textwidth]{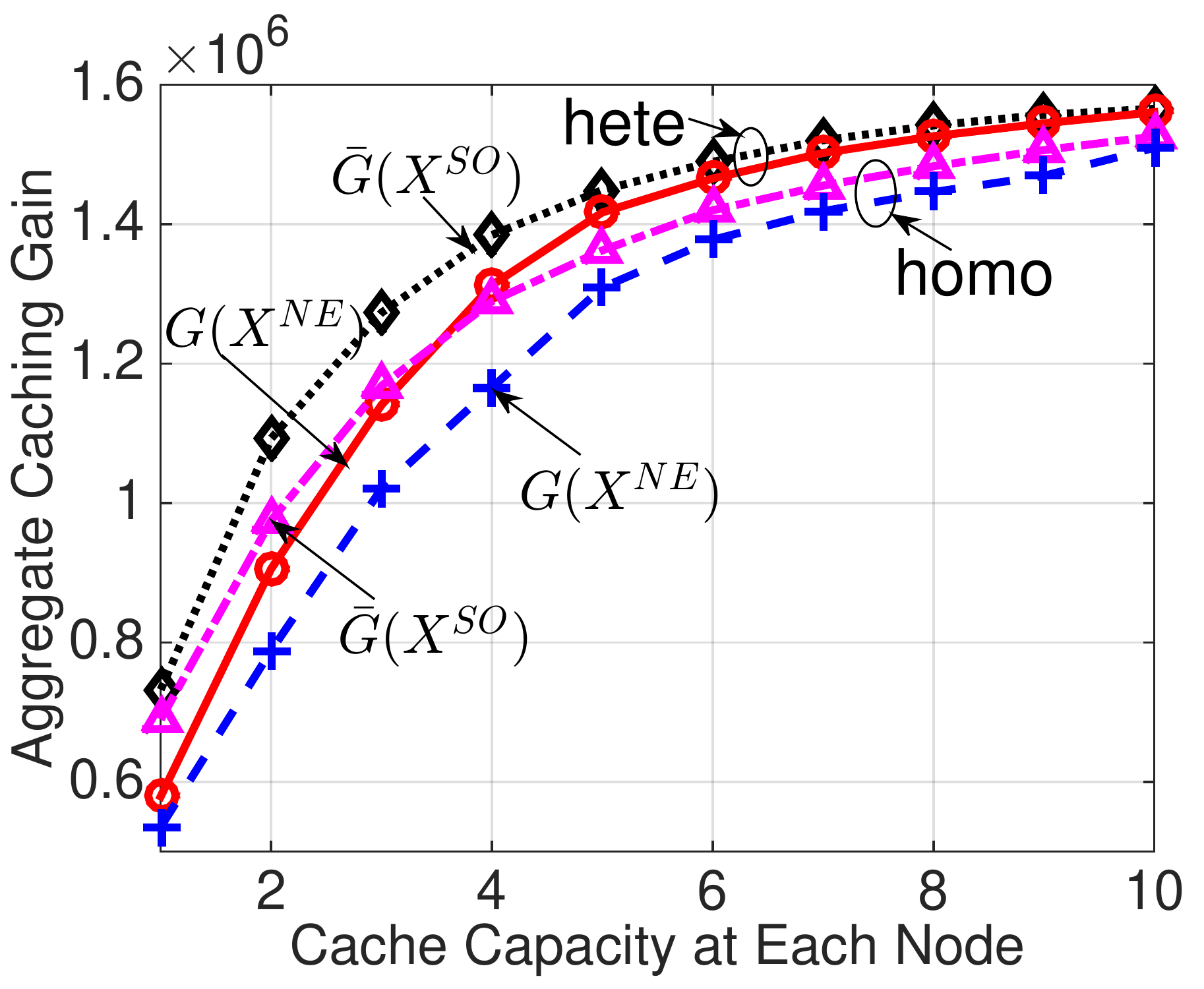}
\caption{$G(\cdot)$ vs. $c_v$, under both heterogeneous and homogeneous request patterns.}\label{fig:HomoLambda_AbileneApp}
\end{minipage}
\begin{minipage}[t]{0.005 \linewidth}
~
\end{minipage}
\begin{minipage}[t]{0.23 \linewidth}
\centering
\includegraphics[width=1\textwidth]{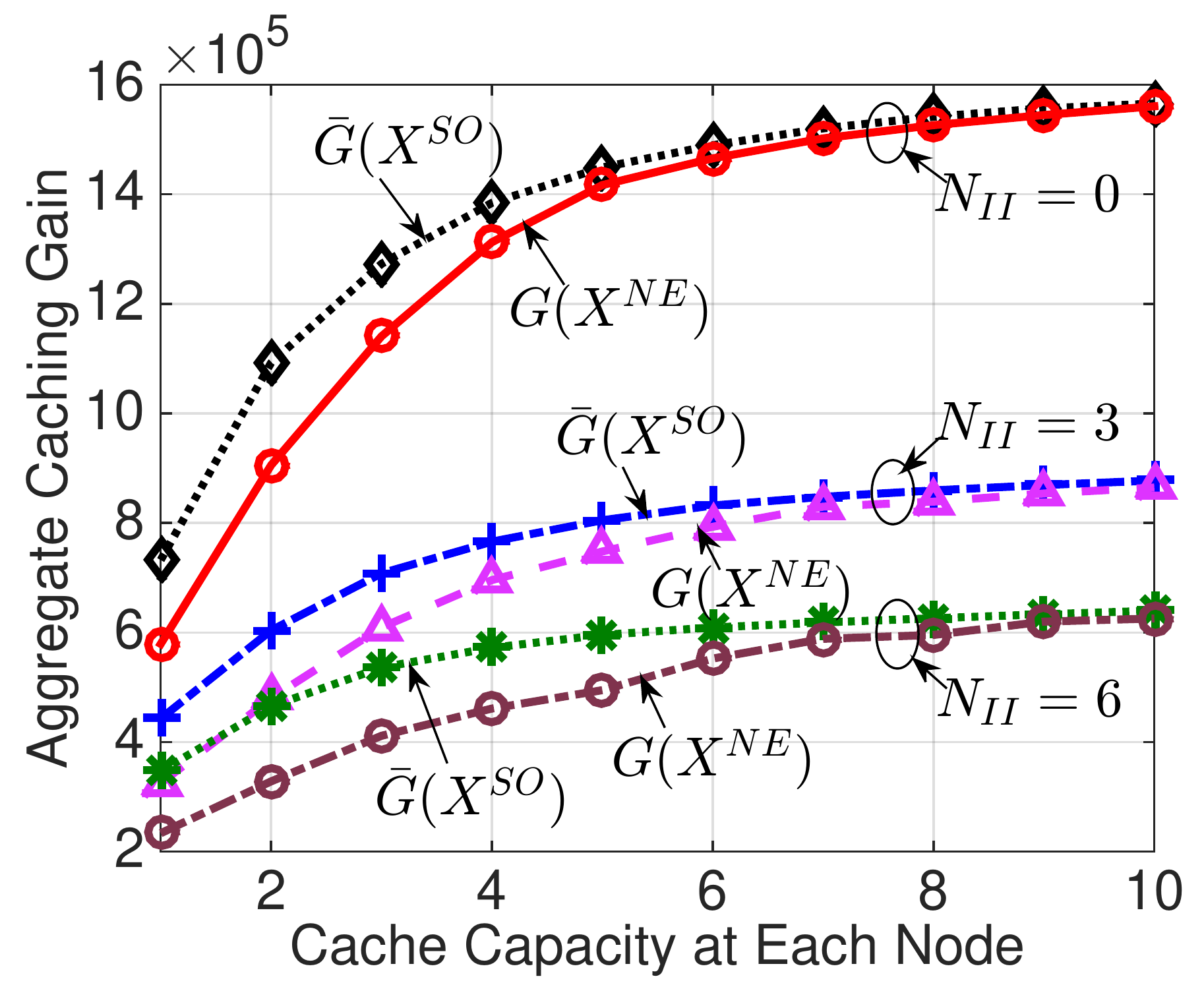}
\caption{$G(\cdot)$ vs. $c_v$, under different no. of Type-II nodes $N_{II}$.}\label{fig:IdleNodes_AbileneApp}
\end{minipage}
\begin{minipage}[t]{0.005 \linewidth}
~
\end{minipage}
\begin{minipage}[t]{0.23 \linewidth}
\centering
\includegraphics[width=1.024\textwidth]{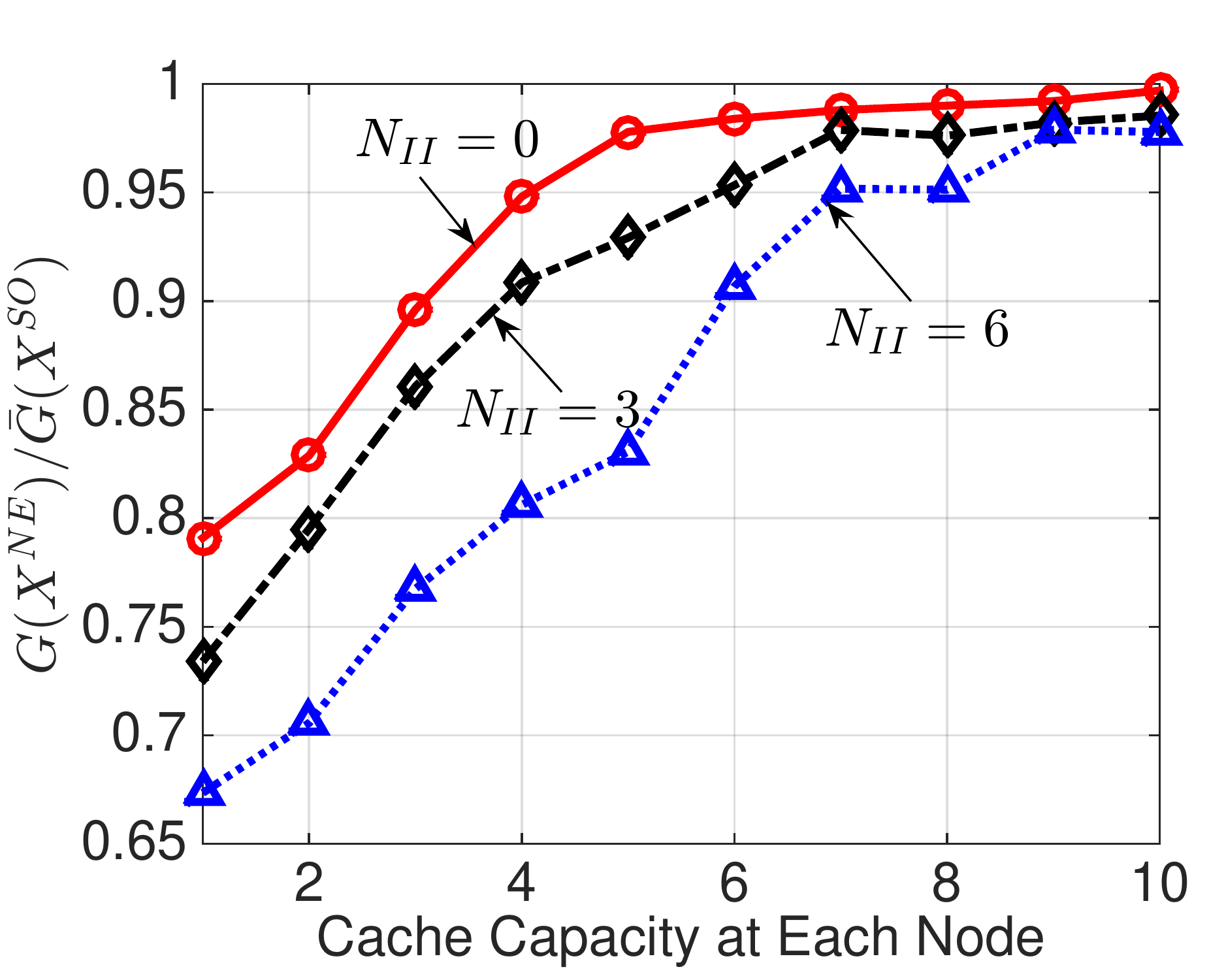}
\caption{$G(\boldsymbol{x}^{\rm NE})/\bar{G}(\boldsymbol{x}^{\rm SO})$ vs. $c_v$, under different $N_{II}$.}\label{fig:IdleNodes_Abilene_Ratios_App}
\end{minipage}
\begin{minipage}[t]{0.005 \linewidth}
~
\end{minipage}
\begin{minipage}[t]{0.25 \linewidth}
\centering
\includegraphics[width=1.08\textwidth]{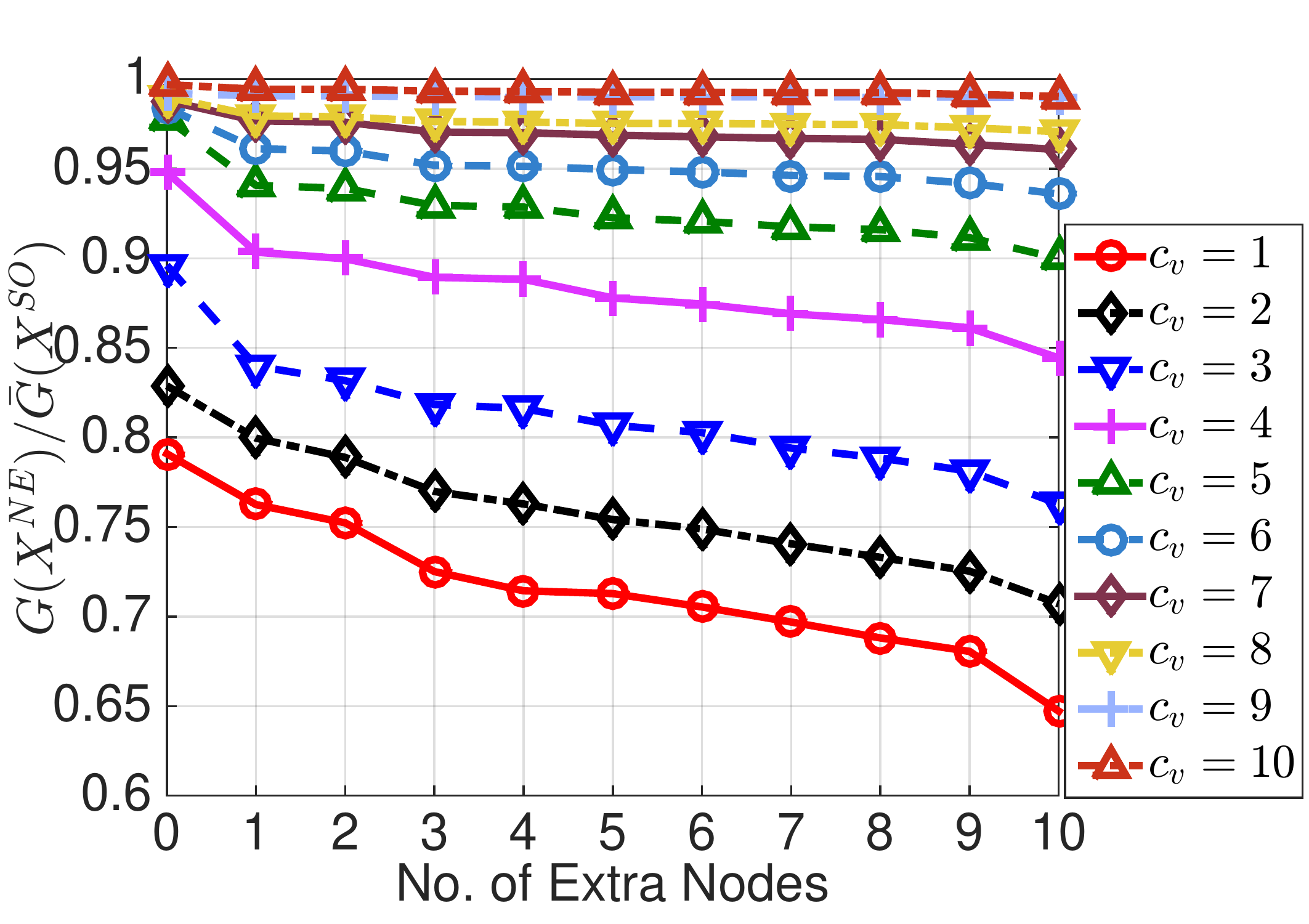}
\caption{$G(\boldsymbol{x}^{\rm NE})/\bar{G}(\boldsymbol{x}^{\rm SO})$ vs. no. of extra cache nodes, under different $c_v$.}\label{fig:ImpactExtraNodesApp}
\end{minipage}
\end{figure*}

\textbf{Results in the Grid Topology.} 
We perform simulations on the Grid topology shown in Figure \ref{fig:Grid}. 
We consider a set $\I=\{1,\ldots,16\}$ of content items, and generate the cost on each edge uniformly at random from the interval $[1, 100]$. 
We show the performance corresponding to selfish behaviors in Figures \ref{fig:HomoLambda_Grid}--\ref{fig:IdleNodes_Ratios_Grid}. 
Different from the Abilene and GEANT networks, we observe in Figure \ref{fig:HomoLambda_Grid} that the homogeneous request pattern leads to a larger aggregate caching gain but a smaller ratio $G(X^{NE})/\bar{G}(X^{SO})$ than that under the heterogeneous request pattern.  
As in the Abilene and GEANT networks, $G(X^{NE})/\bar{G}(X^{SO})$ increases with $c_v$, and decreases with $N_{II}$.

\textbf{Results for the Scenario with Unequal-Sized Items.} 
We perform simulations in the Abilene network for the case where different content items have different sizes in Figures \ref{fig:HomoLambda_AbileneApp} -- \ref{fig:ImpactExtraNodesApp}. 
We assume that the sizes of the $|\I|=10$ content items are 
$\boldsymbol{L}=\{0.2,0.4,0.6,0.8,1.0,1.2,1.4,1.6,1.8,2.0\},$
and node 1 is the designated server of the 10 content items.  
We compare the performance achieved by the approximate Nash equilibria of Game 2 and that by the socially optimal solution. 
Figure \ref{fig:HomoLambda_AbileneApp} shows that the homogeneous request pattern leads to larger gaps between $G(\boldsymbol{x}^{\rm NE})$ and $\bar{G}(\boldsymbol{x}^{\rm SO})$, and hence a worse performance achieved by selfish caching behaviors at the approximate Nash equilibrium. 
Figures \ref{fig:IdleNodes_AbileneApp} -- \ref{fig:ImpactExtraNodesApp} show that the gap between $G(\boldsymbol{x}^{\rm NE})$ and $\bar{G}(\boldsymbol{x}^{\rm SO})$ decreases with the cache capacity at each node, while the gap increases with $N_{II}$. 
Furthermore, adding more extra caches makes the PoA worse.

\section{Conclusion}\label{sec:conclusion}

In this paper, we analyze selfish caching games on directed graphs, which can yield arbitrary bad performance.
We show that a PSNE exists and can be found in polynomial time if there is no mixed request loop, and we can avoid mixed request loops by properly choosing request forwarding paths.
We then show that although cache paradox happens, i.e., adding extra cache nodes does not improve the performance of PSNE, with the homogeneous request pattern property and the path overlap property, the PoA is bounded in arbitrary-topology networks. 
We further show that the selfish caching game with unequal-sized items admits an approximate PSNE with bounded PoA in special cases. 

There are several interesting directions to explore in the future, such as analyzing the impact of the congestion effect on each edge, analyzing the joint caching and routing decisions of selfish nodes,  analyzing the privacy issue, analyzing the dynamic selfish caching game under incomplete information, and analyzing the coalitional game for the caching network with multiple cache providers where each provider owns several cache nodes.

\bibliographystyle{IEEEtran}
\bibliography{TNET-2019-00477}

\begin{thebibliography}{10}
\providecommand{\url}[1]{#1}
\csname url@samestyle\endcsname
\providecommand{\newblock}{\relax}
\providecommand{\bibinfo}[2]{#2}
\providecommand{\BIBentrySTDinterwordspacing}{\spaceskip=0pt\relax}
\providecommand{\BIBentryALTinterwordstretchfactor}{4}
\providecommand{\BIBentryALTinterwordspacing}{\spaceskip=\fontdimen2\font plus
\BIBentryALTinterwordstretchfactor\fontdimen3\font minus
  \fontdimen4\font\relax}
\providecommand{\BIBforeignlanguage}[2]{{%
\expandafter\ifx\csname l@#1\endcsname\relax
\typeout{** WARNING: IEEEtran.bst: No hyphenation pattern has been}%
\typeout{** loaded for the language `#1'. Using the pattern for}%
\typeout{** the default language instead.}%
\else
\language=\csname l@#1\endcsname
\fi
#2}}
\providecommand{\BIBdecl}{\relax}
\BIBdecl

\bibitem{QianMobiHoc2019}
Q.~Ma, E.~Yeh, and J.~Huang, ``How bad is selfish caching?'' in \emph{The
  Twentieth ACM International Symposium on Mobile Ad Hoc Networking and
  Computing (Mobihoc '19)}.\hskip 1em plus 0.5em minus 0.4em\relax ACM, 2019.

\bibitem{CDN}
S.~Borst, V.~Gupta, and A.~Walid, ``Distributed caching algorithms for content
  distribution networks,'' in \emph{IEEE INFOCOM}, 2010, pp. 1--9.

\bibitem{CDN2}
M.~Dehghan, A.~Seetharam, B.~Jiang, T.~He, T.~Salonidis, J.~Kurose, D.~Towsley,
  and R.~Sitaraman, ``On the complexity of optimal routing and content caching
  in heterogeneous networks,'' in \emph{IEEE INFOCOM}, 2015, pp. 936--944.

\bibitem{YehICN2014}
E.~Yeh, T.~Ho, Y.~Cui, M.~Burd, R.~Liu, and D.~Leong, ``Vip: A framework for
  joint dynamic forwarding and caching in named data networks,'' in \emph{ACM
  ICN}, 2014, pp. 117--126.

\bibitem{femtocell}
K.~Shanmugam, N.~Golrezaei, A.~G. Dimakis, A.~F. Molisch, and G.~Caire,
  ``Femtocaching: Wireless content delivery through distributed caching
  helpers,'' \emph{IEEE TIT}, vol.~59, no.~12, pp. 8402--8413, 2013.

\bibitem{DSR}
N.~Laoutaris, O.~Telelis, V.~Zissimopoulos, and I.~Stavrakakis, ``Distributed
  selfish replication,'' \emph{IEEE TPDS}, vol.~17, no.~12, pp. 1401--1413,
  2006.

\bibitem{p2p}
E.~Cohen and S.~Shenker, ``Replication strategies in unstructured peer-to-peer
  networks,'' in \emph{ACM SIGCOMM}, vol.~32, no.~4, 2002, pp. 177--190.

\bibitem{poularakis2018distributed}
K.~Poularakis, G.~Iosifidis, A.~Argyriou, I.~Koutsopoulos, and L.~Tassiulas,
  ``Distributed caching algorithms in the realm of layered video streaming,''
  \emph{IEEE TMC}, 2018.

\bibitem{YehSigmetrics}
S.~Ioannidis and E.~Yeh, ``Adaptive caching networks with optimality
  guarantees,'' in \emph{ACM SIGMETRICS}, vol.~44, no.~1, 2016, pp. 113--124.

\bibitem{ao2015distributed}
W.~C. Ao and K.~Psounis, ``Distributed caching and small cell cooperation for
  fast content delivery,'' in \emph{ACM MobiHoc}, 2015, pp. 127--136.

\bibitem{Milking}
L.~Wang, G.~Tyson, J.~Kangasharju, and J.~Crowcroft, ``Milking the cache cow
  with fairness in mind,'' \emph{IEEE/ACM TON}, vol.~25, no.~5, pp. 2686--2700,
  2017.

\bibitem{afanasyev2010usage}
M.~Afanasyev, T.~Chen, G.~M. Voelker, and A.~C. Snoeren, ``Usage patterns in an
  urban wifi network,'' \emph{IEEE/ACM TON}, vol.~18, no.~5, pp. 1359--1372,
  2010.

\bibitem{vega2012topology}
D.~Vega, L.~Cerda-Alabern, L.~Navarro, and R.~Meseguer, ``Topology patterns of
  a community network: Guifi. net,'' in \emph{IEEE WiMob}, 2012, pp. 612--619.

\bibitem{draves2004routing}
R.~Draves, J.~Padhye, and B.~Zill, ``Routing in multi-radio, multi-hop wireless
  mesh networks,'' in \emph{ACM MobiCom}, 2004, pp. 114--128.

\bibitem{ValidUtilityGame}
A.~Vetta, ``Nash equilibria in competitive societies, with applications to
  facility location, traffic routing and auctions,'' in \emph{IEEE FOCS}, 2002,
  pp. 416--425.

\bibitem{CSgame}
C.~Papadimitriou, ``Algorithms, games, and the internet,'' in \emph{ACM STOC},
  2001, pp. 749--753.

\bibitem{roughgarden2002bad}
T.~Roughgarden and {\'E}.~Tardos, ``How bad is selfish routing?'' \emph{Journal
  of the ACM}, vol.~49, no.~2, pp. 236--259, 2002.

\bibitem{juhn1997harmonic}
L.-S. Juhn and L.-M. Tseng, ``Harmonic broadcasting for video-on-demand
  service,'' \emph{IEEE transactions on broadcasting}, vol.~43, no.~3, pp.
  268--271, 1997.

\bibitem{huang2012confused}
T.-Y. Huang, N.~Handigol, B.~Heller, N.~McKeown, and R.~Johari, ``Confused,
  timid, and unstable: picking a video streaming rate is hard,'' in
  \emph{Proceedings of the 2012 Internet Measurement Conference}.\hskip 1em
  plus 0.5em minus 0.4em\relax ACM, 2012, pp. 225--238.

\bibitem{Survey}
G.~S. Paschos, G.~Iosifidis, M.~Tao, D.~Towsley, and G.~Caire, ``The role of
  caching in future communication systems and networks,'' \emph{arXiv preprint
  arXiv:1805.11721}, 2018.

\bibitem{shukla2017hold}
S.~Shukla, O.~Bhardwaj, A.~A. Abouzeid, T.~Salonidis, and T.~He, ``Hold'em
  caching: Proactive retention-aware caching with multi-path routing for
  wireless edge networks,'' in \emph{ACM MobiHoc}, 2017, p.~24.

\bibitem{tadrous2016joint}
J.~Tadrous, A.~Eryilmaz, and H.~El~Gamal, ``Joint smart pricing and proactive
  content caching for mobile services,'' \emph{IEEE/ACM TON}, vol.~24, no.~4,
  pp. 2357--2371, 2016.

\bibitem{YuanyuanInfocom2020}
Y.~Li and S.~Ioannidis, ``Universally stable cache networks,'' in \emph{IEEE
  INFOCOM}, 2020.

\bibitem{KellyCache2019}
M.~{Mahdian}, A.~{Moharrer}, S.~{Ioannidis}, and E.~{Yeh}, ``Kelly cache
  networks,'' in \emph{IEEE INFOCOM}, 2019, pp. 217--225.

\bibitem{garetto2015efficient}
M.~Garetto, E.~Leonardi, and S.~Traverso, ``Efficient analysis of caching
  strategies under dynamic content popularity,'' in \emph{IEEE INFOCOM}, 2015,
  pp. 2263--2271.

\bibitem{zhang2018coded}
J.~Zhang, X.~Lin, and X.~Wang, ``Coded caching under arbitrary popularity
  distributions,'' \emph{IEEE TIT}, vol.~64, no.~1, pp. 349--366, 2018.

\bibitem{DrCache2018}
J.~{Li}, T.~{Khoa Phan}, W.~{Koong Chai}, D.~{Tuncer}, G.~{Pavlou},
  D.~{Griffin}, and M.~{Rio}, ``Dr-cache: Distributed resilient caching with
  latency guarantees,'' in \emph{IEEE INFOCOM}, 2018, pp. 441--449.

\bibitem{zhao2018red}
T.~Zhao, I.-H. Hou, S.~Wang, and K.~Chan, ``Red/led: An asymptotically optimal
  and scalable online algorithm for service caching at the edge,'' \emph{IEEE
  JSAC}, vol.~36, no.~8, pp. 1857--1870, 2018.

\bibitem{li2018hierarchical}
X.~Li, X.~Wang, P.-J. Wan, Z.~Han, and V.~C. Leung, ``Hierarchical edge caching
  in device-to-device aided mobile networks: Modeling, optimization, and
  design,'' \emph{IEEE JSAC}, 2018.

\bibitem{zhao2018collaborative}
X.~Zhao, P.~Yuan, and S.~Tang, ``Collaborative edge caching in context-aware
  device-to-device networks,'' \emph{IEEE TVT}, vol.~67, no.~10, pp.
  9583--9596, 2018.

\bibitem{kwak2018hybrid}
J.~Kwak, Y.~Kim, L.~B. Le, and S.~Chong, ``Hybrid content caching in 5g
  wireless networks: Cloud versus edge caching,'' \emph{IEEE TWC}, vol.~17,
  no.~5, pp. 3030--3045, 2018.

\bibitem{cao2018optimal}
X.~Cao, J.~Zhang, and H.~V. Poor, ``An optimal auction mechanism for mobile
  edge caching,'' in \emph{IEEE ICDCS}, 2018, pp. 388--399.

\bibitem{FerragutSIGMETRICS2016}
A.~Ferragut, I.~Rodriguez, and F.~Paganini, ``Optimizing ttl caches under
  heavy-tailed demands,'' in \emph{ACM SIGMETRICS}, 2016, p. 101–112.

\bibitem{DehghanTTLton2019}
M.~{Dehghan}, L.~{Massoulié}, D.~{Towsley}, D.~S. {Menasché}, and Y.~C.
  {Tay}, ``A utility optimization approach to network cache design,''
  \emph{IEEE/ACM TON}, vol.~27, no.~3, pp. 1013--1027, 2019.

\bibitem{qin2018content}
Z.~Qin, X.~Gan, L.~Fu, X.~Di, J.~Tian, and X.~Wang, ``Content delivery in
  cache-enabled wireless evolving social networks,'' \emph{IEEE TWC}, 2018.

\bibitem{dehghan2015complexity}
M.~Dehghan, A.~Seetharam, B.~Jiang, T.~He, T.~Salonidis, J.~Kurose, D.~Towsley,
  and R.~Sitaraman, ``On the complexity of optimal routing and content caching
  in heterogeneous networks,'' in \emph{IEEE INFOCOM}, 2015, pp. 936--944.

\bibitem{amble2011content}
M.~M. Amble, P.~Parag, S.~Shakkottai, and L.~Ying, ``Content-aware caching and
  traffic management in content distribution networks,'' in \emph{IEEE
  INFOCOM}, 2011.

\bibitem{StratisJSAC2018}
S.~{Ioannidis} and E.~{Yeh}, ``Jointly optimal routing and caching for
  arbitrary network topologies,'' \emph{IEEE JSAC}, vol.~36, no.~6, pp.
  1258--1275, 2018.

\bibitem{chu2016allocating}
W.~Chu, M.~Dehghan, D.~Towsley, and Z.-L. Zhang, ``On allocating cache
  resources to content providers,'' in \emph{ACM ICN}, 2016, pp. 154--159.

\bibitem{gharaibeh2016provably}
A.~Gharaibeh, A.~Khreishah, B.~Ji, and M.~Ayyash, ``A provably efficient online
  collaborative caching algorithm for multicell-coordinated systems,''
  \emph{IEEE TMC}, vol.~15, no.~8, pp. 1863--1876, 2016.

\bibitem{shin2017t}
K.~Shin, C.~Joe-Wong, S.~Ha, Y.~Yi, I.~Rhee, and D.~S. Reeves, ``T-chain: A
  general incentive scheme for cooperative computing,'' \emph{IEEE/ACM TON},
  vol.~25, no.~4, pp. 2122--2137, 2017.

\bibitem{rahimzadeh2017svc}
P.~Rahimzadeh, C.~Joe-Wong, K.~Shin, Y.~Im, J.~Lee, and S.~Ha, ``Svc-tchain:
  Incentivizing good behavior in layered p2p video streaming,'' in \emph{IEEE
  INFOCOM}, 2017, pp. 1--9.

\bibitem{yu2016enhancing}
R.~Yu, S.~Qin, M.~Bennis, X.~Chen, G.~Feng, Z.~Han, and G.~Xue, ``Enhancing
  software-defined ran with collaborative caching and scalable video coding,''
  in \emph{IEEE ICC}, 2016, pp. 1--6.

\bibitem{Lui}
A.~T. Ip, J.~Lui, and J.~Liu, ``A revenue-rewarding scheme of providing
  incentive for cooperative proxy caching for media streaming systems,''
  \emph{ACM TOMM}, vol.~4, no.~1, p.~5, 2008.

\bibitem{maille2015impact}
P.~Maill{\'e}, G.~Simon, and B.~Tuffin, ``Impact of revenue-driven cdn on the
  competition among network operators,'' in \emph{IEEE CNSM}, 2015, pp.
  163--167.

\bibitem{SelfishCaching}
B.~G. Chun, K.~Chaudhuri, H.~Wee, M.~Barreno, C.~H. Papadimitriou, and
  J.~Kubiatowicz, ``Selfish caching in distributed systems: a game-theoretic
  analysis,'' in \emph{ACM PODC}, 2004, pp. 21--30.

\bibitem{MarketSharing}
M.~Goemans, L.~E. Li, V.~S. Mirrokni, and M.~Thottan, ``Market sharing games
  applied to content distribution in ad-hoc networks,'' in \emph{ACM MobiHoc},
  2004, pp. 55--66.

\bibitem{DSR2}
G.~G. Pollatos, O.~A. Telelis, and V.~Zissimopoulos, ``On the social cost of
  distributed selfish content replication,'' in \emph{International Conference
  on Research in Networking}.\hskip 1em plus 0.5em minus 0.4em\relax Springer,
  2008, pp. 195--206.

\bibitem{CSR}
R.~Gopalakrishnan, D.~Kanoulas, N.~N. Karuturi, C.~P. Rangan, R.~Rajaraman, and
  R.~Sundaram, ``Cache me if you can: capacitated selfish replication games,''
  in \emph{Latin American Symposium on Theoretical Informatics}.\hskip 1em plus
  0.5em minus 0.4em\relax Springer, 2012, pp. 420--432.

\bibitem{jiang2018convergence}
B.~Jiang, P.~Nain, and D.~Towsley, ``On the convergence of the ttl
  approximation for an lru cache under independent stationary request
  processes,'' \emph{ACM TOMPECS}, vol.~3, no.~4, p.~20, 2018.

\bibitem{PanigraphyPoisson2018}
N.~K. {Panigrahy}, J.~{Li}, and D.~{Towsley}, ``Network cache design under
  stationary requests: Exact analysis and poisson approximation,'' in
  \emph{IEEE MASCOTS}, 2018, pp. 251--263.

\bibitem{wang2014belief}
Y.-K. Wang, Y.~Yin, and S.~Zhong, ``Belief propagation for spatial spectrum
  access games,'' in \emph{Proceedings of the 15th ACM international symposium
  on Mobile ad hoc networking and computing}.\hskip 1em plus 0.5em minus
  0.4em\relax ACM, 2014, pp. 225--234.

\bibitem{KnapsackBook}
H.~Kellerer, U.~Pferschy, and D.~Pisinger, ``Introduction to np-completeness of
  knapsack problems,'' in \emph{Knapsack problems}.\hskip 1em plus 0.5em minus
  0.4em\relax Springer, 2004, pp. 483--493.

\bibitem{ApproximateNEadditive}
C.~Daskalakis, A.~Mehta, and C.~Papadimitriou, ``Progress in approximate nash
  equilibria,'' in \emph{ACM EC}, 2007, pp. 355--358.

\bibitem{ApproximateNE}
S.~Chien and A.~Sinclair, ``Convergence to approximate nash equilibria in
  congestion games,'' \emph{Games and Economic Behavior}, vol.~71, no.~2, pp.
  315--327, 2011.

\bibitem{AbileneTopology}
A.~Li, X.~Yang, and D.~Wetherall, ``Safeguard: safe forwarding during route
  changes,'' in \emph{PACM CoNEXT}, 2009, pp. 301--312.

\end{thebibliography}

\begin{IEEEbiography}[{\includegraphics[width=1in,height=1.25in,clip,keepaspectratio]{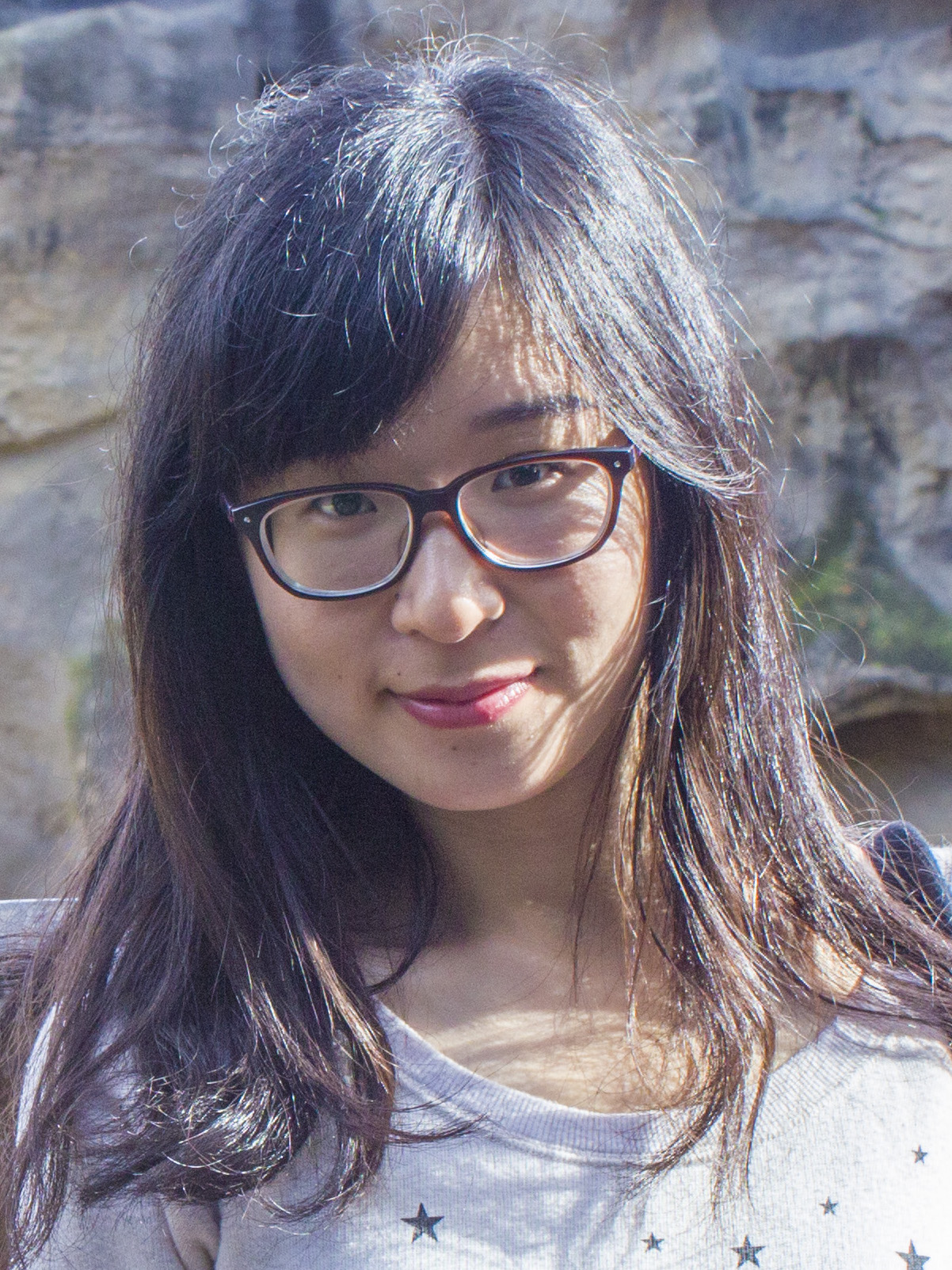}}]{Qian Ma}
is an Associate Professor of School of Intelligent Systems Engineering, Sun Yat-sen University. 
She worked as a Postdoc Research Associate at Northeastern University during 2018-2019. 
She received the Ph.D. degree in the Department of Information Engineering from the Chinese University of Hong Kong in 2017, and the B.S. degree from Beijing University of Posts and Telecommunications (China) in 2012.
Her research interests lie in the field of network optimziation and economics.
She is the recipient of the Best Student Paper Award from the IEEE International Symposium on Modeling and Optimization in Mobile, Ad Hoc and Wireless Networks (WiOpt) in 2015.
\end{IEEEbiography}

\begin{IEEEbiography}[{\includegraphics[width=1in,height=1.25in,clip,keepaspectratio]{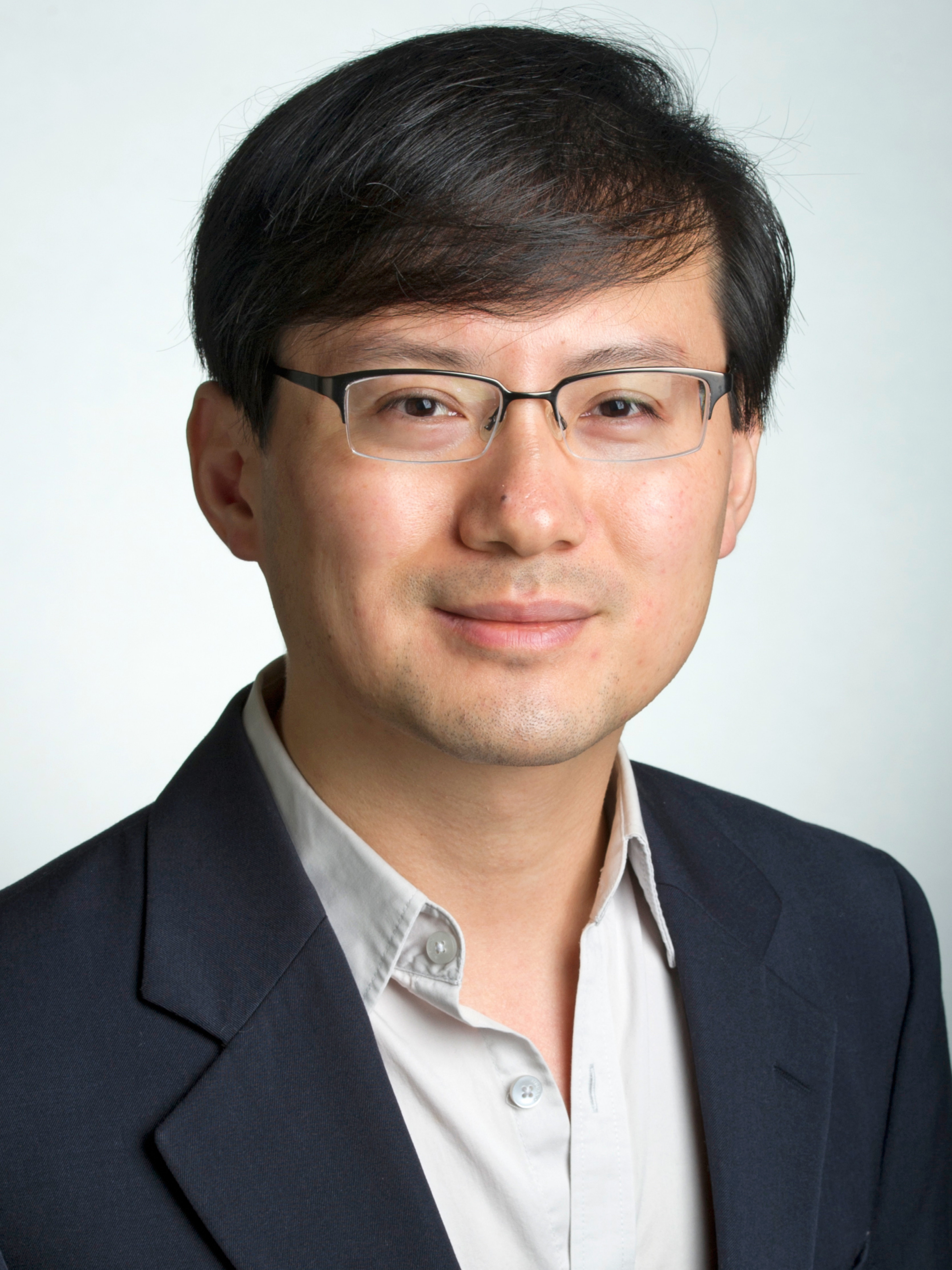}}]{Edmund Yeh}
received the B.S. degree in Electrical Engineering with Distinction and Phi Beta Kappa from Stanford University in 1994, the M. Phil degree in Engineering from Cambridge University on the Winston Churchill Scholarship in 1995, and the Ph.D. in Electrical Engineering and Computer Science from MIT under Professor Robert Gallager in 2001.  He is a Professor of Electrical and Computer Engineering at Northeastern University with a courtesy appointment in Khoury School of Computer Sciences.  He was previously Assistant and Associate Professor of Electrical Engineering, Computer Science, and Statistics at Yale University.  

Professor Yeh is an IEEE Communications Society Distinguished Lecturer.  He serves as TPC Co-Chair for ACM MobiHoc 2021 and served as General Chair for ACM SIGMETRICS 2020.  He is the recipient of the Alexander von Humboldt Research Fellowship, the Army Research Office Young Investigator Award, the Winston Churchill Scholarship, the National Science Foundation and Office of Naval Research Graduate Fellowships, the Barry M. Goldwater Scholarship, the Frederick Emmons Terman Engineering Scholastic Award, and the President's Award for Academic Excellence (Stanford University).  Professor Yeh serves as Treasurer of the Board of Governors of the IEEE Information Theory Society.  He has served as an Associate Editor for IEEE Transactions on Networking, IEEE Transactions on Mobile Computing, and IEEE Transactions on Network Science and Engineering, as Guest Editor-in-Chief of the Special Issue on Wireless Networks for Internet Mathematics, and Guest Editor for IEEE Journal on Selected Areas in Communications - Special Series on Smart Grid Communications. He has received three Best Paper Awards, including awards at the 2017 ACM Conference on Information-Centric Networking (ICN), and at the 2015 IEEE International Conference on Communications (ICC) Communication Theory Symposium.
\end{IEEEbiography}

\begin{IEEEbiography}[{\includegraphics[width=1in,height=1.25in,clip,keepaspectratio]{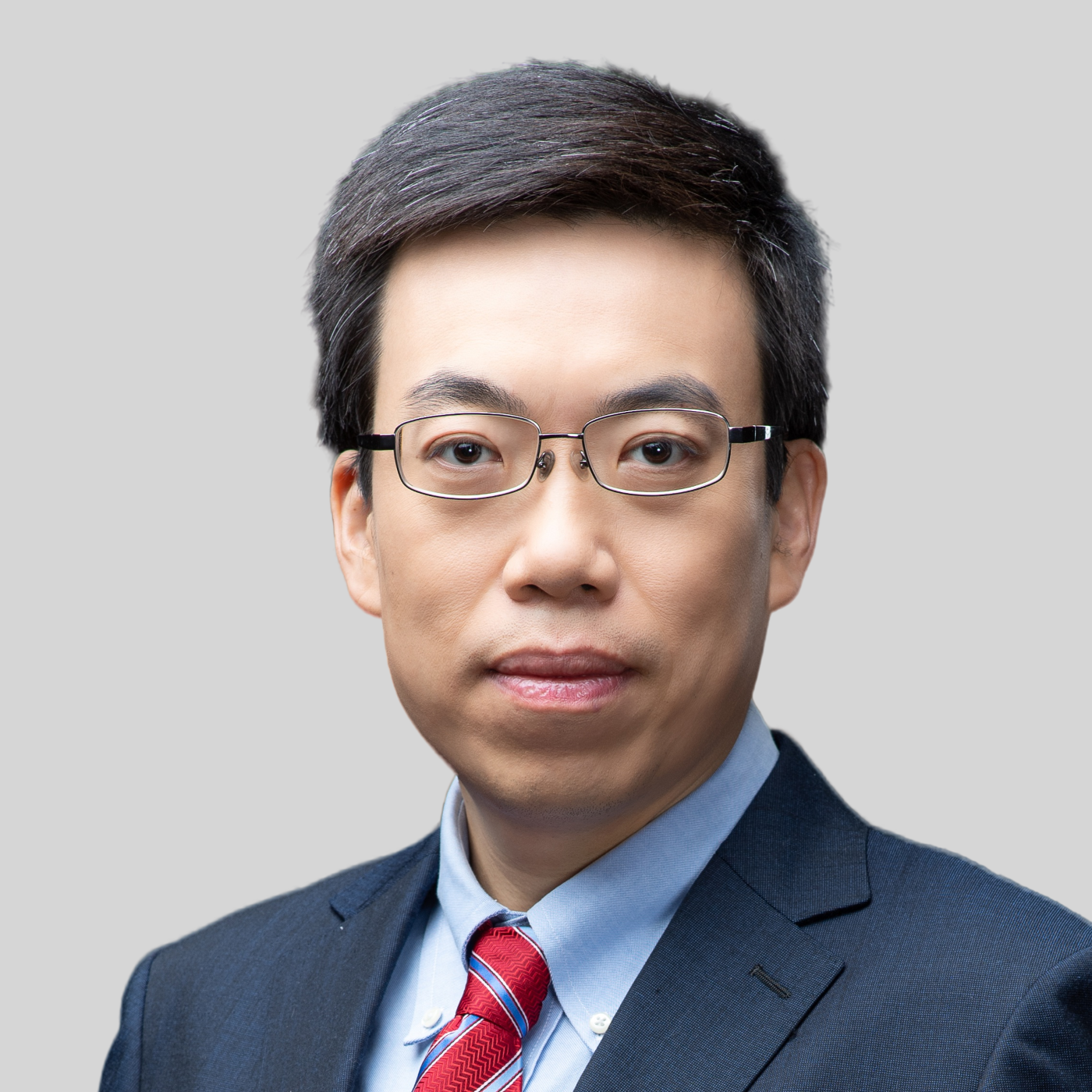}}]{Jianwei Huang}
is a Presidential Chair Professor and the Associate Dean of the School of Science and Engineering, The Chinese University of Hong Kong, Shenzhen. He is also the Vice President of Shenzhen Institute of Artificial Intelligence and Robotics for Society. He received the Ph.D. degree from Northwestern University in 2005, and worked as a Postdoc Research Associate at Princeton University during 2005-2007. He has been an IEEE Fellow, a Distinguished Lecturer of IEEE Communications Society, a Clarivate Analytics Highly Cited Researcher in Computer Science, and Associate Editor-in-Chief of IEEE Open Journal of the Communications Society. He is the incoming Editor-in-Chief of IEEE Transactions on Network Science and Engineering, effective 01/2021. Dr. Huang has published 290 papers in leading international journals and conferences in the area of network optimization and economics, with a total Google Scholar citations of 12,700+ and an H-index of 58. He is the co-author of 9 Best Paper Awards, including IEEE Marconi Prize Paper Award in Wireless Communications in 2011. He has co-authored seven books, including the textbook on "Wireless Network Pricing." He received the CUHK Young Researcher Award in 2014 and IEEE ComSoc Asia-Pacific Outstanding Young Researcher Award in 2009. He has served as an Associate Editor of IEEE Transactions on Mobile Computing, IEEE/ACM Transactions on Networking, IEEE Transactions on Network Science and Engineering, IEEE Transactions on Wireless Communications, IEEE Journal on Selected Areas in Communications - Cognitive Radio Series, and IEEE Transactions on Cognitive Communications and Networking. He has served as the Chair of IEEE ComSoc Cognitive Network Technical Committee and Multimedia Communications Technical Committee. He is the recipient of IEEE ComSoc Multimedia Communications Technical Committee Distinguished Service Award in 2015 and IEEE GLOBECOM Outstanding Service Award in 2010. More detailed information can be found at https://sse.cuhk.edu.cn/en/faculty/huangjianwei.
\end{IEEEbiography}

\appendices

\section{Proof of Theorem \ref{theo:Existence}}\label{app:theo:Existence}

A PSNE of the selfish caching game corresponds to a vertex on the state graph without any outgoing arc, i.e., a sink. 
By showing the existence of a sink on the state graph, we can show the existence of the PSNE of the selfish caching game. 
Since Algorithm \ref{algo:FindNESG} can find a sink on the corresponding state graph, the PSNE of the selfish caching game exists.

\section{Proof of Theorem \ref{theo:ExistencePolytime}}\label{app:theo:ExistencePolytime}

As one player adds a content item to its cache at the first arc of each round, the number of rounds is no greater than $|V||\I|$. 

We next show that each round ends after traversing at most $|\I|(|V|-2)^2$ arcs. 
We focus on one round where player $s$ adds content item $i$. 
The add step will only affect nodes $v\in V$ where $s\in p^{(v,i)}$. 
If $x_{vi}=0$, i.e., $i \notin Z_v$, the add step will not change node $v$'s behavior, as the add step will only decrease the value of content item $i$ to node $v$. 
If $x_{vi}=1$, i.e., $i\in Z_v$, since the add step will decrease the value of content item $i$ to node $v$, node $v$ may replace item $i$ with item $j$. 
Subsequently, the change step where node $v$ replaces item $i$ with item $j$ will lead to the change step where node $u$ (such that $v\in p^{(u,j)}$ and $x_{uj}=1$) replaces item $j$ with item $n$. 
The sequence of change steps caused by node $s$ adding item $i$ ends in at most $|V|-2$ steps, and node $s$ will keep item $i$ in its cache during the sequence of change steps, due to the assumption on no mixed request loop. 
Note that the number of such sequences of change steps is at most $|\I|(|V|-2)$. 
 Hence, each round ends in at most $|\I|(|V|-2)^2$ arcs. 

This proves our bound of $|V||\I|^2(|V|-2)^2$ on the length of the traversal to reach a PSNE.

\section{Proof of Theorem \ref{theo:scalablevalid}}\label{app:theo:scalablevalid}

The first property \eqref{eq:prop21} and \eqref{eq:prop22}, i.e., the social function $G(\cdot)$ is nondecreasing and submodular, is proved in \cite{YehSigmetrics}.

By the definition of our social function, we have 
$$G(\boldsymbol{x})=\sum_{s\in V}g_s(\boldsymbol{x}_s,\boldsymbol{x}_{-s}),$$ 
and therefore the second property \eqref{eq:prop1} is satisfied.

In the following, we prove that the third property \eqref{eq:svug} is satisfied. 
For the social function $G(\cdot)$, we have 
\begin{equation*}
\begin{aligned}
&G(\boldsymbol{x}_s,\boldsymbol{x}_{-s}) - G(\boldsymbol{0},\boldsymbol{x}_{-s}) \\
\overset{(a)}{=}&  \sum_{i\in Z_s} \lambda_{(s,i)} \sum_{k=1}^{|p^{(s,i)}|-1}w_{p_{k+1}p_k} \prod_{k'=2}^k(1-x_{p_{k'}i}) \\
+& \sum_{i\in Z_s} \sum_{\substack{v\in V\setminus \{\S^i, s\}\\s\in p^{(v,i)}}  } \lambda_{(v,i)} \sum_{k=k_{p^{(v,i)}}(s)}^{|p^{(v,i)}|-1}w_{p_{k+1}p_k} \prod_{k'=1}^k(1-x_{p_{k'}i}) 
\end{aligned}
\end{equation*}
\begin{equation*}
\begin{aligned}
\overset{(b)}{=}&  \sum_{i\in Z_s} \lambda_{i}  \sum_{k=1}^{|p^{(s,i)}|-1}w_{p_{k+1}p_k} \prod_{k'=2}^k(1-x_{p_{k'}i}) \cdot \\
& \sum_{v\in V\setminus \S^i} \mathbbm{1}_{ \left\{  v=s \mbox{ or } \left[ s\in p^{(v,i)} \mbox{ and }  \prod_{k'=1}^{k_{p^{(v,i)}}(s)-1}(1-x_{p_{k'}i})=1 \right]  \right\} } \\
\overset{(c)}{\leq} &  \sum_{i\in Z_s} \lambda_{i}  \sum_{k=1}^{|p^{(s,i)}|-1}w_{p_{k+1}p_k} \prod_{k'=2}^k(1-x_{p_{k'}i})  \cdot \left( \max_{v\in V,i\in\I} |p^{(v,i)}| -1 \right)
\end{aligned}
\end{equation*}
Step $(a)$ is from the definition of $G(\cdot)$. 
Step $(b)$ holds under the homogeneous request pattern and path overlap properties. 
Step $(c)$ is due to $\prod_{k'=1}^{k_{p^{(v,i)}}(s)-1}(1-x_{p_{k'}i}) \leq 1$ and $|\{v\in V\setminus \S^i: s\in p^{(v,i)}\}| \leq \max_{v\in V,i\in\I} |p^{(v,i)}| -1$. 
Hence for $g_s(\cdot)$, we have 
\begin{equation*}
\begin{aligned}
& g_s(\boldsymbol{x}_s,\boldsymbol{x}_{-s}) = \sum_{i\in Z_s} \lambda_{(s,i)} \sum_{k=1}^{|p^{(s,i)}|-1}w_{p_{k+1}p_k} \\
&~~~+  \sum_{i\in \I \setminus Z_s} \lambda_{(s,i)} \sum_{k=1}^{|p^{(s,i)}|-1}w_{p_{k+1}p_k} \left( 1-\prod_{k'=2}^k(1-x'_{p_{k'}i}) \right) \\
&  \overset{(d)}{\geq}  \sum_{i\in Z_s} \lambda_{(s,i)} \sum_{k=1}^{|p^{(s,i)}|-1}w_{p_{k+1}p_k} \\
& \overset{(e)}{\geq}   \sum_{i\in Z_s} \lambda_{(s,i)} \sum_{k=1}^{|p^{(s,i)}|-1}w_{p_{k+1}p_k} \prod_{k'=2}^k(1-x_{p_{k'}i}) \\
& \overset{(f)}{\geq}  \frac{1}{\alpha} \cdot \left( G(\boldsymbol{x}_s,\boldsymbol{x}_{-s}) - G(\boldsymbol{0},\boldsymbol{x}_{-s})  \right)
\end{aligned}
\end{equation*}
Step $(d)$ is due to $1-\prod_{k'=2}^k(1-x'_{p_{k'}i}) \geq 0$. 
Step $(e)$ is due to $\prod_{k'=2}^k(1-x_{p_{k'}i}) \leq 1$. 
Step $(f)$ is due to step $(c)$. 
This completes our proof.

\section{Proof of Theorem \ref{theo:PoAalpha}}\label{app:theo:PoAalpha}

To characterize the PoA, the key step is to find the relationship between the Nash equilibria and the socially optimal solution.
Let $\boldsymbol{x}^{\rm NE}$ and $\boldsymbol{x}^\ast$ denote the caching strategy under any Nash equilibria and socially optimal solution, respectively. 
For any Nash equilibria $\boldsymbol{x}^{\rm NE}$, we have
\vspace{-2mm}
\begin{equation*}
\begin{aligned}
G(\boldsymbol{x}^{\rm NE}) & \textstyle \overset{(a)}{=}\sum_{s\in V}g_s(\boldsymbol{x}_s^{\rm NE},\boldsymbol{x}_{-s}^{\rm NE}) \\
& \textstyle  \overset{(b)}{\geq} \sum_{s\in V}g_s(\boldsymbol{x}^\ast_s,\boldsymbol{x}_{-s}^{\rm NE}) \\
&  \textstyle \overset{(c)}{\geq} \sum_{s\in V} \frac{1}{\alpha} \left( G(\boldsymbol{x}^\ast_s,\boldsymbol{x}_{-s}^{\rm NE}) -G(\boldsymbol{0},\boldsymbol{x}_{-s}^{\rm NE}) \right) \\
&  \textstyle \overset{(d)}{=} \frac{1}{\alpha} \sum_{s\in V}\left( G(Z^\ast_s,Z_{-s}^{\rm NE}) -G(\emptyset,Z_{-s}^{\rm NE})  \right) \\
& \textstyle \overset{(e)}{\geq} \frac{1}{\alpha} \sum_{s\in V} ( G(Z^\ast_1,\ldots,Z^\ast_s,Z^{\rm NE}) \\
 &~~~~~~~~~~~~~~~-G(Z^\ast_1,\ldots,Z^\ast_{s-1},Z^{\rm NE}) ) \\
& \textstyle \overset{(f)}{=} \frac{1}{\alpha} \left(  G(Z^\ast_1,\ldots,Z^\ast_{|V|},Z^{\rm NE}) -G(Z^{\rm NE}) \right) \\
& \textstyle \overset{(g)}{\geq} \frac{1}{\alpha} \left(  G(Z^\ast)-G(Z^{\rm NE}) \right) \\
& \textstyle \overset{(h)}{=} \frac{1}{\alpha} \left(  G(\boldsymbol{x}^\ast)-G(\boldsymbol{x}^{\rm NE}) \right)
\end{aligned}
\end{equation*}
Step $(a)$ is from the definition of $G(\cdot)$. 
Step $(b)$ is from the definition of Nash equilibrium. 
Step $(c)$ is due to \eqref{eq:svug}. 
Step $(d)$ is due to the one-to-one correspondence between $\boldsymbol{x}_s$ and $Z_s$. 
Step $(e)$ is due to the submodularity of $G(\cdot)$. 
Step $(f)$ is derived by calculating the summation over all $s\in V$. 
Step $(g)$ is because $G(\cdot)$ is non-decreasing. 
Step $(h)$ is due to the one-to-one correspondence between $\boldsymbol{x}$ and $Z$. 
The relationship between $G(\boldsymbol{x}^{\rm NE})$ and $G(\boldsymbol{x}^\ast)$ leads to 
\begin{equation*}
\frac{G(\boldsymbol{x}^{\rm NE})}{G(\boldsymbol{x}^\ast)} \geq \frac{1}{1+\alpha} = \frac{1}{\max_{v\in V,i\in\I}|p^{(v,i)}|}.
\end{equation*}
This completes our proof.

\section{Proof of Theorem \ref{theo:deltaG}}\label{app:theo:deltaG}

We first introduce two notations. 
Given a strategy profile $Z=\{Z_1,\ldots, Z_{|V|}\}$, let $Z \oplus Z_s'$ denote the new strategy profile where player $s$ changes its strategy from $Z_s$ to $Z_s'$. 
Mathematically, $Z \oplus Z_s'=\{Z_1,\ldots,Z_{s-1},Z_s',Z_{s+1},\ldots, Z_{|V|}\}$. 
Given $Z=\{Z_1,\ldots, Z_{|V|}\}$, we define $Z^s=\{Z_1,\ldots,Z_s, \emptyset_{s+1},\ldots, \emptyset_{|V|}\}$. 

The social function $G(\cdot)$ satisfies the following two lemmas.

\begin{lemma}\label{lemma:Ginequality}\cite{ValidUtilityGame}
For any feasible strategy profile $Z$, the socially optimal solution $Z^\ast$ satisfies:
\begin{equation*}
\begin{aligned}
G(Z^\ast) \leq & G(Z) + \sum_{s:Z_s^\ast \in Z^\ast - Z} G'_{Z_s^\ast}(Z\oplus \emptyset_s) \\
& - \sum_{s:Z_s \in Z-Z^\ast}G'_{Z_s}(Z^\ast \cup Z^{s-1}).
\end{aligned}
\end{equation*}
\end{lemma}

\begin{proof}
By the submodularity of $G$, we have
\begin{equation*}
\begin{aligned}
&G(Z^\ast \cup Z)-G(Z) \\
=& \left[ G(Z\cup Z^\ast_1)-G(Z) \right] + \cdots \\
& + \left[ G(Z\cup Z^\ast_1 \cup \cdots \cup Z^\ast_j)-G(Z\cup \cdots \cup Z^\ast_{j-1}) \right] \\
\leq& \left[ G(Z\cup Z^\ast_1)-G(Z) \right] + \cdots + \left[ G(Z\cup Z^\ast_j)-G(Z) \right] \\
=& \sum_{s:Z_s^\ast \in Z^\ast - Z} G'_{Z_s^\ast}(Z) \\
\leq& \sum_{s:Z_s^\ast \in Z^\ast - Z} G'_{Z_s^\ast}(Z \oplus \emptyset_s)
\end{aligned}
\end{equation*}
Furthermore, 
\begin{equation*}
G(Z^\ast \cup Z)=G(Z^\ast) + \sum_{s:Z_s \in Z-Z^\ast}G'_{Z_s}(Z^\ast \cup Z^{s-1}).
\end{equation*}
This completes our proof.
\end{proof}

\begin{lemma}\label{lemma:Galpha}
For the selfish caching game with the homogeneous request pattern and path overlap properties on caching graphs with no mixed request loop, any Nash equilibrium $Z$ and the socially optimal solution $Z^\ast$ satisfy
\begin{equation*}
\textstyle G(Z^\ast) \leq (1+\alpha) G(Z) - \sum_{s\in V} G'_{Z_s}(Z^\ast \cup Z - Z_s).
\end{equation*}
\end{lemma}

\begin{proof}
From Lemma \ref{lemma:Ginequality}, we know that
\begin{align*}
&G(Z^\ast) \\
&\leq G(Z) + \sum_{s:Z_s^\ast \in Z^\ast - Z} G'_{Z_s^\ast}(Z\oplus \emptyset_s)\\
& \quad \quad \quad \quad - \sum_{s:Z_s \in Z-Z^\ast}G'_{Z_s}(Z^\ast \cup Z^{s-1}) \\
& \overset{(a)}{\leq} G(Z) + \alpha \cdot \sum_{s:Z_s^\ast \in Z^\ast - Z} g_s(Z)\\
& \quad \quad  \quad \quad - \sum_{s:Z_s \in Z-Z^\ast}G'_{Z_s}(Z^\ast \cup Z^{s-1}) \\
&\overset{(b)}{=} G(Z) + \alpha \cdot \sum_{s:Z_s\in Z-Z^\ast} g_s(Z)\\
& \quad \quad \quad \quad - \sum_{s:Z_s \in Z-Z^\ast}G'_{Z_s}(Z^\ast \cup Z^{s-1}) \\
&\overset{(c)}{=} \textstyle G(Z) + \alpha \cdot \left[ G(Z)- \sum_{s:Z_s\in Z\cap Z^\ast} g_s(Z) \right] \\
&\quad  \textstyle  - \sum_{s:Z_s \in Z-Z^\ast}G'_{Z_s}(Z^\ast \cup Z^{s-1}) \\
&\overset{(d)}{\leq}  \textstyle  G(Z) + \alpha \cdot \left[ G(Z)- \frac{1}{\alpha} \cdot \sum_{s:Z_s\in Z\cap Z^\ast} G'_{Z_s}(Z\oplus \emptyset_s) \right]\\
& \quad  \textstyle - \sum_{s:Z_s \in Z-Z^\ast}G'_{Z_s}(Z^\ast \cup Z^{s-1}) \\
& \overset{(e)}{\leq} \textstyle  (1+ \alpha) G(Z) - \sum_{s:Z_s\in Z\cap Z^\ast} G'_{Z_s}(Z^\ast\cup Z - Z_s)  \\
&\quad  \textstyle - \sum_{s:Z_s \in Z-Z^\ast}G'_{Z_s}(Z^\ast\cup Z - Z_s) \\
& =  \textstyle  (1+ \alpha) G(Z) -\sum_{s \in V} G'_{Z_s}(Z^\ast\cup Z - Z_s)
\end{align*}
Step $(a)$ is due to \eqref{eq:svug}. 
Step $(b)$ is due to the fact that the set of indices where $Z^\ast_s \neq Z_s$ is the same as the set of indices where $Z_s \neq Z^\ast_s$.
Step $(c)$ is due to $G(Z)=\sum_{s\in V}g_s(Z)$. 
Step $(d)$ is due to \eqref{eq:svug}. 
Step $(e)$ is due to the submodularity of $G(\cdot)$. 
This completes our proof.
\end{proof}

Now we are ready to prove Theorem \ref{theo:deltaG}. 
By the submodularity of $G$, any Nash equilibrium $Z$ and the socially optimal solution $Z^\ast$ satisfy
\begin{equation*}
\begin{aligned}
& \textstyle \sum_{s\in V} G'_{Z_s}(Z^\ast \cup Z - Z_s)  \geq \sum_{s\in V} G'_{Z_s}(\mathcal{I}^{|V|} - Z_s) \\
& \textstyle  \geq \sum_{s\in V} G'_{Z_s}(\mathcal{I}^{|V|} - Z_s)\frac{G'_{Z_s}(Z^{s-1})}{G'_{Z_s}(\emptyset)} \\
& \geq \min_{s\in V}\frac{G'_{Z_s}(\mathcal{I}^{|V|} - Z_s)}{G'_{Z_s}(\emptyset)} \sum_{s\in V} G'_{Z_s}(Z^{s-1}) = (1-\delta(G)) G(Z).
\end{aligned}
\end{equation*}
From Lemma \ref{lemma:Galpha}, we know that
\begin{equation*}
\begin{aligned}
G(Z^\ast) &\leq (1+\alpha) G(Z) - \sum_{s\in V} G'_{Z_s}(Z^\ast \cup Z - Z_s)\\
& \leq (1+\alpha) G(Z) - (1-\delta(G)) G(Z)\\
& = (\alpha+\delta(G)) G(Z)  
\end{aligned}
\end{equation*}
This completes our proof.

\begin{figure*}[t] 
\centering 
\begin{minipage}[t]{0.45 \linewidth}
\centering
\includegraphics[width=0.7\textwidth]{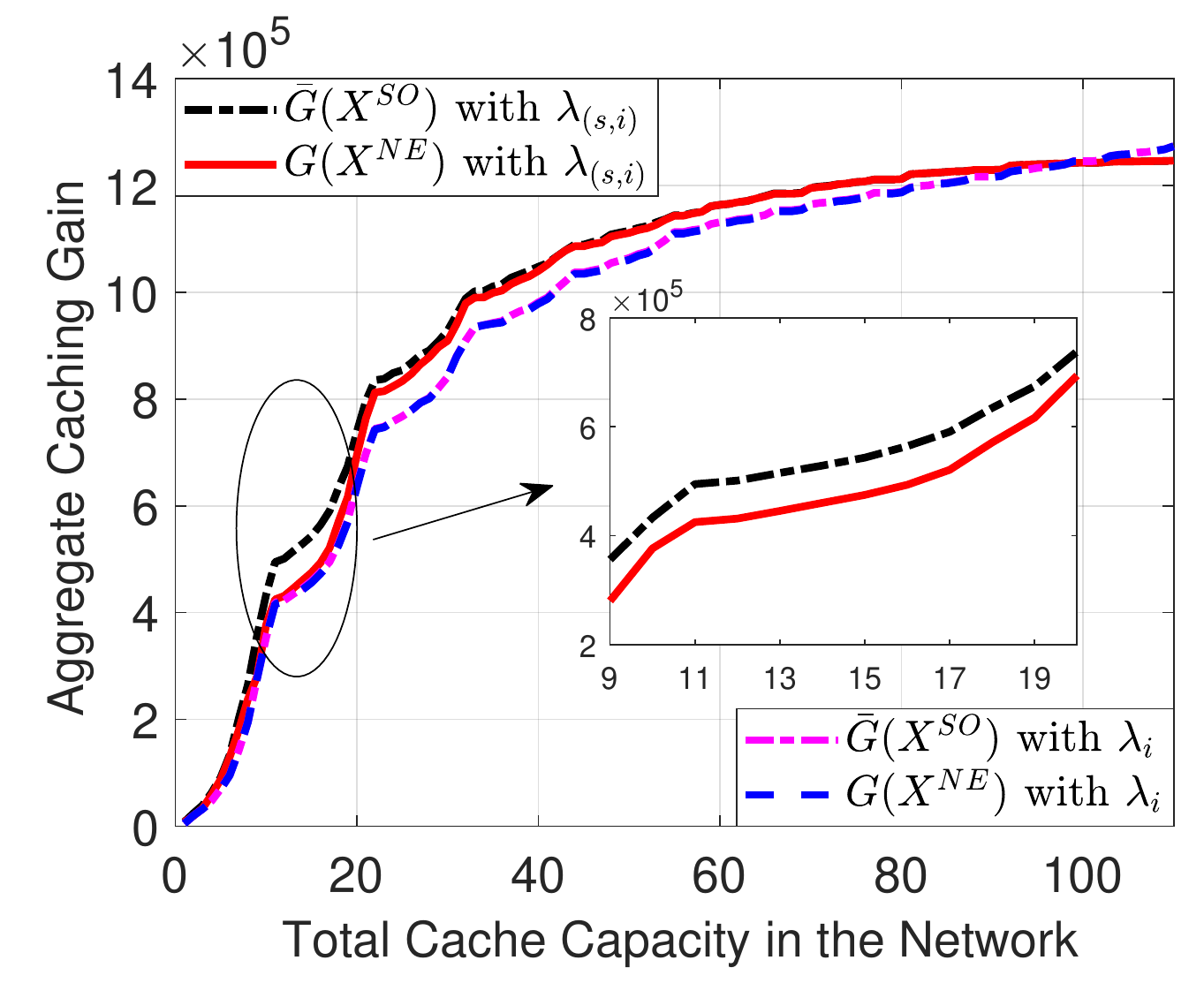}
  \caption{$G(\cdot)$ vs. $\sum_{v\in V}c_v$, under both heterogeneous and homogeneous request patterns.}\label{fig:HomoLambda_Abilene}
\end{minipage}
\begin{minipage}[t]{0.015 \linewidth}
~
\end{minipage}
\begin{minipage}[t]{0.45 \linewidth}
\centering
\includegraphics[width=0.7\textwidth]{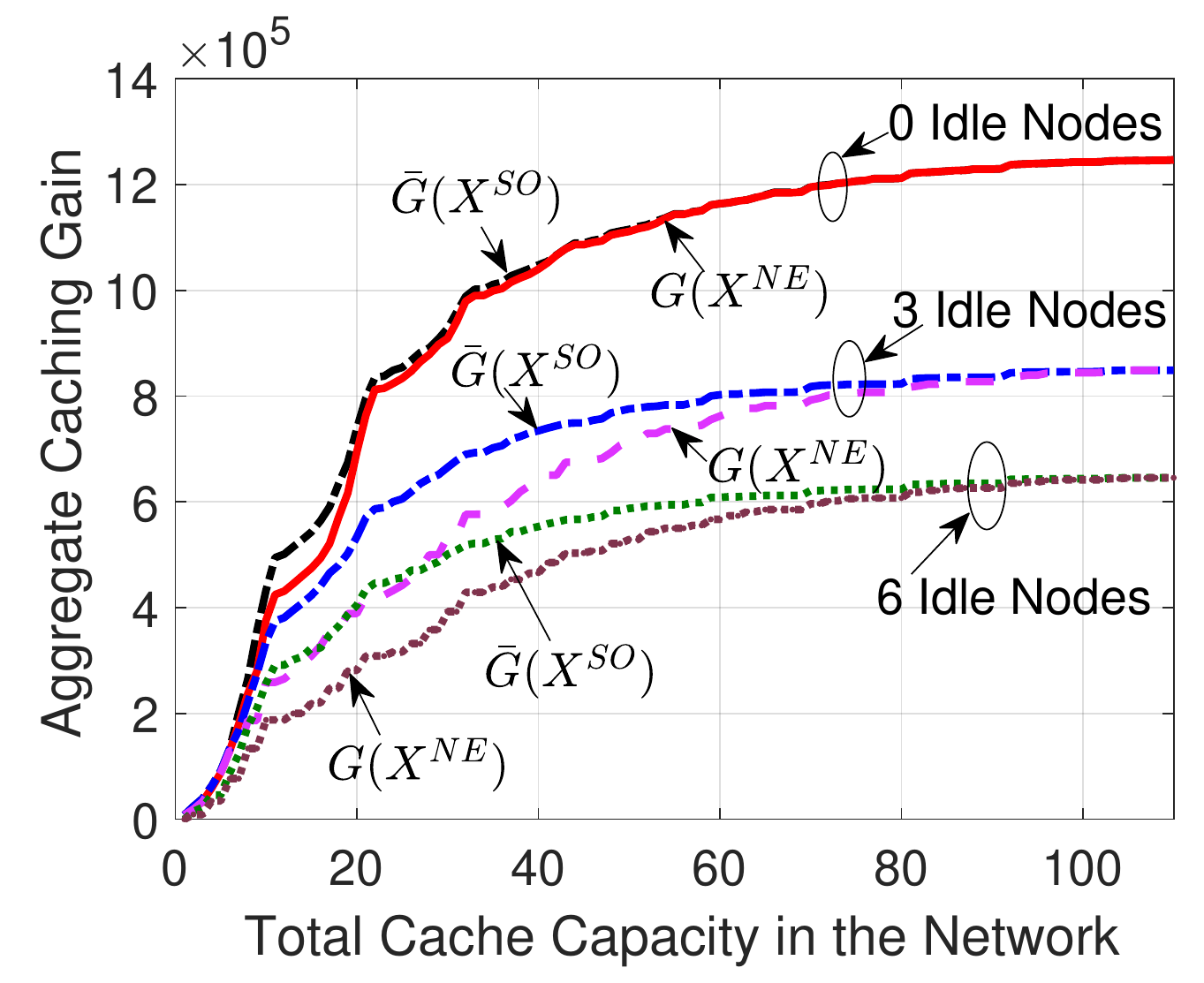}
  \caption{$G(\cdot)$ vs. $\sum_{v\in V}c_v$, under different number of idle nodes.}\label{fig:IdleNodes_Abilene}
\end{minipage}
\end{figure*}

\section{Proof of Theorem \ref{theo:ExistenceBeta}}\label{app:theo:ExistenceBeta}

We show a constructive proof. 
Since Algorithm \ref{algo:FindNECloud} can find a $\beta$-approximate Nash equilibrium, we can prove the existence of the $\beta$-approximate Nash equilibria for Game \ref{gameapp}.

\section{Proof of Theorem \ref{theo:ExistenceAPP}}\label{app:theo:ExistenceAPP}

The procedure to find a $\beta$-approximate Nash equilibrium in Algorithm \ref{algo:FindNECloud} takes at most $|V|-1$ steps. 
In each step, the equilibrium strategy of the particular player can be found by solving problem \eqref{prob:knapsack} with complexity $\mathcal{O}(|\I|)$ (see Section 9.4.2 of \cite{KnapsackBook}). 
Hence, the $\beta$-approximate Nash equilibrium can be found in $\mathcal{O}(|V||\I|)$ time.

\section{Proof of Theorem \ref{theo:PoAapp}}\label{app:theo:PoAapp}

Let $\boldsymbol{x}^{\beta-NE}$ and $\boldsymbol{x}^\ast$ denote the caching strategy under any $\beta$-approximate Nash equilibria and the socially optimal solution, respectively. 
According to the definition of the $\beta$-approximate Nash equilibrium, we have
\begin{equation*}
g_s(\boldsymbol{x}_s^{\beta-NE},\boldsymbol{x}_{-s}^{\beta-NE}) \geq \frac{1}{\beta} \cdot g_s(\boldsymbol{x}^\ast_s,\boldsymbol{x}_{-s}^{\beta-NE}), \forall s\in V.
\end{equation*}
Summing over all $s\in V$, we have
\begin{equation*} 
G(\boldsymbol{x}_s^{\beta-NE},\boldsymbol{x}_{-s}^{\beta-NE}) \geq \frac{1}{\beta} \sum_{s\in V}g_s(\boldsymbol{x}^\ast_s,\boldsymbol{x}_{-s}^{\beta-NE}), \forall s\in V.
\end{equation*}
According to~\eqref{eq:svug} and Theorem \ref{theo:PoAalpha}, we know that
\begin{equation*}
\begin{aligned}
g_s(\boldsymbol{x}^\ast_s,\boldsymbol{x}_{-s}^{\beta-NE}) &\geq \frac{1}{\alpha} \left( G(\boldsymbol{x}^\ast_s,\boldsymbol{x}_{-s}^{\beta-NE}) -G(\boldsymbol{0},\boldsymbol{x}_{-s}^{\beta-NE})  \right) \\
&\geq \frac{1}{\alpha} \left( G(\boldsymbol{x}^\ast) -G(\boldsymbol{x}^{\beta-NE})  \right).
\end{aligned}
\end{equation*}
Hence, 
\begin{equation*}
\begin{aligned}
&G(\boldsymbol{x}_s^{\beta-NE},\boldsymbol{x}_{-s}^{\beta-NE}) \\
\geq & \frac{1}{\beta}  \sum_{s\in V}g_s(\boldsymbol{x}^\ast_s,\boldsymbol{x}_{-s}^{\beta-NE}) \\
\geq & \sum_{s\in V} \frac{1}{\alpha \cdot \beta} \left[ G(\boldsymbol{x}^\ast_s,\boldsymbol{x}_{-s}^{\beta-NE}) -G(\boldsymbol{0},\boldsymbol{x}_{-s}^{\beta-NE}) \right] \\ 
\geq & \frac{1}{\alpha \cdot \beta} \left[ G(\boldsymbol{x}^\ast)-G(\boldsymbol{x}^{\beta-NE}) \right],
\end{aligned}
\end{equation*}
which leads to 
\begin{equation*}
\frac{G(\boldsymbol{x}^{\beta-NE})}{G(\boldsymbol{x}^\ast)} \geq \frac{1}{1+\alpha \cdot \beta} = \frac{1}{1+ \beta \cdot \left( \max_{v\in V,i\in\I}|p^{(v,i)}| -1\right)}.
\end{equation*}
This completes our proof.

\section{Proof of Lemma \ref{lemma:SocialWelfareUB}}\label{app:lemma:SocialWelfareUB}

Since problem \eqref{prob:relaxedG} maximizes the same function as problem \eqref{prob:maxCG} over a larger domain, we have:
$$G(\boldsymbol{\phi}^\ast) \geq G(\boldsymbol{x}^\ast).$$ 
By Goemans-Williamson inequality, we have that: for any sequence of $z_i\in[0,1],i\in\{1,\ldots,n\},$
\begin{equation*}
\left(1-\frac{1}{e}\right)\min\left\{1,\sum_{i=1}^n z_i\right\}  \leq 1-\prod_{i=1}^n (1-z_i) \leq \min\left\{1,\sum_{i=1}^n z_i\right\}
\end{equation*}
So for any feasible $\boldsymbol{\phi}$, we have:
$$ \left( 1-\frac{1}{e} \right) L(\boldsymbol{\phi}) \leq G(\boldsymbol{\phi}) \leq L(\boldsymbol{\phi}) . $$
Hence, we have
$$L(\boldsymbol{\phi}^\ast) \geq G(\boldsymbol{\phi}^\ast).$$ 
By the optimality of $\boldsymbol{\phi}^{\ast\ast}$ to problem \eqref{prob:maxL}, we have
$$L(\boldsymbol{\phi}^{\ast\ast}) \geq L(\boldsymbol{\phi}^\ast) . $$
This completes our proof.

\section{Simulations for Different Cache Capacities}\label{app:newFigs}

We perform simulations on how the aggregate caching gains $G(\boldsymbol{x}^{\rm NE})$ and $\bar{G}(\boldsymbol{x}^{\rm SO})$ change with the total cache capacity in the Abilene network. 
We start from the state where the cache capacity of each node is zero. 
In each trial of the simulation, we add one unit of cache capacity to one node (staring from node 1 to node 11). 
For example, when the total cache capacity is 3, we have $c_v=1$ for $v=1,2,3$, and $c_v=0$ for $v=4,5,\ldots, 11$.  
When the total cache capacity is 28, we have $c_v=3$ for $v=1,\ldots,6$, and $c_v=2$ for $v=7,\ldots, 11$. 

Figure \ref{fig:HomoLambda_Abilene} shows the aggregate caching gains $G(\boldsymbol{x}^{\rm NE})$ and $\bar{G}(\boldsymbol{x}^{\rm SO})$ under different total cache capacities in the network, for the case with heterogeneous request patterns $\lambda_{(s,i)}$ (the upper two curves) and for the case with homogeneous request patterns $\lambda_{(s,i)}=\lambda_i, \forall s\in V, i\in \mathcal{I}$ (the lower two curves), respectively. 
We can see that the gap between $G(\boldsymbol{x}^{\rm NE})$ and $\bar{G}(\boldsymbol{x}^{\rm SO})$ under homogeneous request patterns $\lambda_i$ is smaller than the gap under heterogeneous request patterns $\lambda_{(s,i)}$. 
This is consistent with the observation that we obtained from Figure \ref{fig:HomoLambda_Abilene10} in the paper.

In practice, some cache nodes are intermediate routers which do not request for any content items, i.e., idle nodes. 
We show the impact of the number of idle nodes in Figure \ref{fig:IdleNodes_Abilene}. 
We can see that the gap between $G(\boldsymbol{x}^{\rm NE})$ and $\bar{G}(\boldsymbol{x}^{\rm SO})$ decreases with the total cache capacity in the network, while the gap increases with the number of idles nodes. 
This is consistent with the observations that we obtained from Figure \ref{fig:IdleNodes_Abilene10} in the paper. 
This implies that the impact of the selfish behaviors is mitigated when the cache resource increases, and the selfish behaviors of idle nodes degrade the (relative) performance of Nash equilibria (since the selfish idle nodes will not cache content items at equilibrium).

\end{document}